\documentclass[acmtog,authorversion]{acmart}

\usepackage{booktabs} 
\usepackage{graphicx}
\usepackage{amsmath}
\usepackage{mathrsfs}
\usepackage{microtype}
\usepackage{enumitem}
\usepackage[ruled,linesnumbered,longend]{algorithm2e}

\setcitestyle{nosort}

\SetAlFnt{\small}
\SetAlCapFnt{\small}
\SetAlCapNameFnt{\small}
\SetAlCapHSkip{0pt}
\newcommand{\algocomment}[1]{{\tcp{#1}\vspace*{-0.25ex}}}
\newcommand{\algospace}{{\vspace*{0.3ex}}}

\newtheorem{assumption}{Assumption}[section]
\newtheorem{prop}{Proposition}[section]
\newtheorem{lemma}{Lemma}[section]

\theoremstyle{remark}
\newtheorem{remark}{Remark}[section]

\acmJournal{TOG}
\acmYear{2019}
\acmVolume{38}
\acmNumber{6}
\acmArticle{163}
\acmMonth{11} 
\acmDOI{10.1145/3355089.3356491}

\setcopyright{acmcopyright}

\newcommand{\xzu}{{$\mathbf{x}$-$\mathbf{z}$-$\mathbf{u}$}}
\newcommand{\zxu}{{$\mathbf{z}$-$\mathbf{x}$-$\mathbf{u}$}}
\newcommand{\lev}[2]{{\mathscr{L}_{#2}^{#1}}}
\newcommand{\dom}[1]{{\mathrm{dom(}{#1}\mathrm{)}}}
\newcommand{\conv}[1]{{\mathrm{conv(}{#1}\mathrm{)}}}
\newcommand{\hatg}{\hat{g}}

\DeclareMathOperator*{\argmin}{argmin}

\newcommand{\squeezemath}[1]{
\medmuskip= #1 mu
\thinmuskip= #1 mu
\thickmuskip= #1 mu
}

\begin{document}
\title{Accelerating ADMM for Efficient Simulation and Optimization}

\author{Juyong Zhang}
\email{juyong@ustc.edu.cn}
\authornote{Equal contributions.}
\authornote{Corresponding author (\href{mailto:juyong@ustc.edu.cn}{juyong@ustc.edu.cn}).}
\affiliation{%
	\institution{University of Science and Technology of China}
}

\author{Yue Peng}
\email{echoyue@mail.ustc.edu.cn}
\authornotemark[1]
\affiliation{%
	\institution{University of Science and Technology of China}
}

\author{Wenqing Ouyang}
\email{wq8809@mail.ustc.edu.cn}
\authornotemark[1]
\affiliation{%
	\institution{University of Science and Technology of China}
}

\author{Bailin Deng}
\email{DengB3@cardiff.ac.uk}
\affiliation{%
	\institution{Cardiff University}
}

\authorsaddresses{Authors' addresses:
	$\{$Juyong Zhang, Yue Peng, Wenqing Ouyang$\}$, University of Science and Technology of China, 96 Jinzhai Road, Hefei 230026, Anhui, China, $\{$juyong@ustc.edu.cn, echoyue@mail.ustc.edu.cn, wq8809@mail.ustc.edu.cn$\}$; Bailin Deng, Cardiff University, 5 The Parade, Cardiff CF24 3AA, Wales, United Kingdom, DengB3@cardiff.ac.uk.}

\begin{abstract}
The alternating direction method of multipliers (ADMM) is a popular approach for solving optimization problems that are potentially non-smooth and with hard constraints. It has been applied to various computer graphics applications, including physical simulation, geometry processing, and image processing. However, ADMM can take a long time to converge to a solution of high accuracy. Moreover, many computer graphics tasks involve non-convex optimization, and there is often no convergence guarantee for ADMM on such problems since it was originally designed for convex optimization.
In this paper, we propose a method to speed up ADMM using Anderson acceleration, an established technique for accelerating fixed-point iterations. We show that in the general case, ADMM is a fixed-point iteration of the second primal variable and the dual variable, and Anderson acceleration can be directly applied. Additionally, when the problem has a separable target function and satisfies certain conditions, ADMM becomes a fixed-point iteration of only one variable, which further reduces the computational overhead of Anderson acceleration.
Moreover, we analyze a particular non-convex problem structure that is common in computer graphics, and prove the convergence of ADMM on such problems under mild assumptions. We apply our acceleration technique on a variety of optimization problems in computer graphics, with notable improvement on their convergence speed.
\end{abstract}

%
%


%
%


\begin{teaserfigure}
	\centering
	\includegraphics[width=\textwidth]{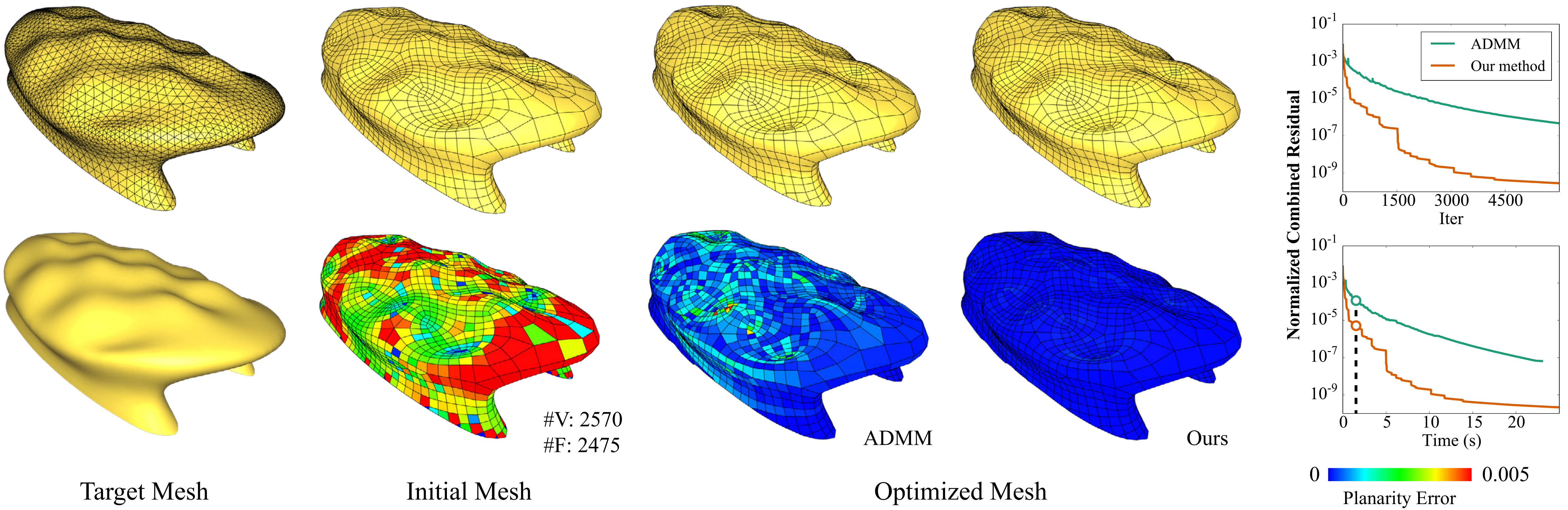}
	\caption{We apply our accelerated ADMM solver to optimize a quad mesh, subject to hard constraints of face planarity and soft constraints of closeness to a reference surface. Our solver leads to a faster decrease of combined residual than the original ADMM, achieving better satisfaction of hard constraints within the same computational time (highlighted in the plot in bottom right).}
	\label{Fig:Airport}
\end{teaserfigure}

\maketitle

\section{Introduction}

Many tasks in computer graphics involve solving optimization problems. For example, a geometry processing task may compute the vertex positions of a deformed mesh by minimizing its deformation energy~\cite{SorkineA07}, whereas a physical simulation task may optimize the node positions of a system to enforce physics laws that govern its behavior~\cite{Martin2011,Schumacher2012}. Such tasks are often formulated as \emph{unconstrained} optimization, where the target function penalizes the violation of certain conditions so that they are satisfied as much as possible by the solution. It has been an active research topic to develop fast numerical solvers for such problems, with various methods proposed in the past~\cite{SorkineA07,Liu2008,BouazizDSWP12,Liu2013,Bouaziz2014,Wang15,KovalskyGL16,LiuBK17,ShtengelPSKL17,RabinovichPPS17}.

On the other hand, some applications involve optimization with \emph{hard constraints}, i.e., conditions that need to be enforced strictly. Such \emph{constrained} optimization problems are often more difficult to solve~\cite{nocedal2006numerical}. One possible solution strategy is to introduce a quadratic penalty term for the hard constraints with a large weight, thereby converting it into an unconstrained problem that is easier to handle. However, to strictly enforce the hard constraints, their penalty weight needs to approach infinity~\cite{nocedal2006numerical}, which can cause instability for numerical solvers. More sophisticated techniques, such as sequential quadratic programming or the interior-point method, can enforce constraints without stability issues. However, these solvers often incur high computational costs and may not meet the performance requirements for graphics applications. It becomes even more challenging for non-smooth problems where the target function is not everywhere differentiable, as many constrained optimization solvers require gradient information and may not be applicable for such cases.

In recent years, the alternating direction method of multipliers (ADMM)~\cite{boyd2011distributed} has become a popular approach for solving optimization problems that are potentially non-smooth and with hard constraints. The key idea is to introduce auxiliary variables and derive an equivalent problem with a separable target function, subject to a linear compatibility constraint between the original variables and the auxiliary variables~\cite{Combettes2011}.
ADMM searches for a solution to this converted problem by alternately updating the original variables, the auxiliary variables, and the dual variables. With properly chosen auxiliary variables, each update step can reduce to simple sub-problems that can be solved efficiently, often in parallel with closed-form solutions. In addition, ADMM does not rely on the smoothness of the problem, and converges quickly to a solution of moderate accuracy~\cite{boyd2011distributed}. Such properties make ADMM an attractive choice for solving large-scale optimization problems in various applications such as signal processing~\cite{Chartrand2013,Simonetto2014}, image processing~\cite{Figueiredo2010,Almeida2013}, and computer vision~\cite{Liu2013-Tensor}. Recently, ADMM has also been applied for computer graphics problems such as geometry processing~\cite{Bouaziz2013,Neumann2013,Zhang2014-LBC,Xiong2014,Neumann2014-CMM}, physics simulation~\cite{Gregson2014,Pan2017,Overby2017}, and computational photography~\cite{Heide2016,Xiong2017,Wang2018-Megapixel}.  

Despite the effectiveness and versatility of ADMM, there are still two major limitations for its use in computer graphics. First, although ADMM converges quickly in initial iterations, its final convergence might be slow~\cite{boyd2011distributed}. This makes it impractical for problems with a strong demand for solution accuracy, such as those with strict requirements on the satisfaction of hard constraints. Recent attempts to accelerate ADMM such as~\cite{Goldstein2014,Kadkhodaie2015,Zhang2018-GMRES} are only designed for convex problems, which limits their applications in computer graphics. 
Second, ADMM was originally designed for convex problems, whereas many computer graphics tasks involve non-convex optimization. Although ADMM turns out to be effective for many non-convex problems in practice, its convergence for general non-convex optimization remains an open research question. Recent convergence results such as~\cite{li2015global,hong2016convergence,magnusson2016convergence,wang2019global} rely on strong assumptions that are not satisfied by many computer graphics problems.

This paper addresses these two issues of ADMM. First, we propose a method to accelerate ADMM for non-convex optimization problems. 
Our approach is based on Anderson acceleration~\cite{Anderson1965,Walker2011}, a well-established technique for accelerating fixed-point iterations. 
Previously, Anderson acceleration has been applied to local-global solvers for unconstrained optimization problems in computer graphics~\cite{Peng2018}. Our approach expands its applicability to many constrained optimization problems as well as other unconstrained problems where local-solver solvers are not feasible. To this end, we need to solve two problems: (i) we must find a way to interpret ADMM as a fixed-point iteration; (ii) as Anderson acceleration can become unstable, we should define criteria to accept the accelerated iterate and a fall-back strategy when it is not accepted, similar to~\cite{Peng2018}.
We show that in the general case ADMM is a fixed-point iteration of the second primal variable and the dual variable, and we can evaluate the effectiveness of an accelerated iterate via its \emph{combined residual} which is known to vanish when the solver converges.
Moreover, when the problem structure satisfies some mild conditions, one of these two variables can be determined from the other one; in this case ADMM becomes a fixed-point iteration of only one variable with less computational overhead, and we can accept an accelerated iterate based on a more simple condition.
We apply this method to a variety of ADMM solvers for computer graphics problems, and observe a notable improvement in their convergence rates.

Additionally, we provide a new convergence proof of ADMM on non-convex problems, under weaker assumptions than the convergence results in~\cite{li2015global,hong2016convergence,magnusson2016convergence,wang2019global}. For a particular problem structure that is common in computer graphics, we also provide sufficient conditions for the global linear convergence of ADMM. Our proofs shed new light on the convergence properties of non-convex ADMM solvers.
\section{Related Work}

\paragraph{Optimization solvers in computer graphics}
The development of efficient optimization solvers has been an active research topic in computer graphics. One particular type of method, called local-global solvers, has been widely used for unconstrained optimization in geometry processing and physical simulation.
For geometry processing, Sorkine and Alexa~\shortcite{SorkineA07} proposed a local-global approach to minimize deformation energy for as-rigid-as-possible mesh surface modeling.  Liu et al.~\shortcite{Liu2008} developed a similar method to perform conformal and isometric parameterization for triangle meshes. Bouaziz et al.~\shortcite{BouazizDSWP12} extended the approach to a unified framework for optimizing discrete shapes. For physical simulation, Liu et al.~\shortcite{Liu2013} proposed a local-global solver for optimization-based simulation of mass-spring systems. Bouaziz et al.~\shortcite{Bouaziz2014} extended this approach to the projective dynamics framework for implicit time integration of physical systems via energy minimization.

Local-global solvers often converge quickly to an approximate solution, but may be slow for final convergence. Other methods have been proposed to achieve improved convergence rates. For geometry processing, Kovalsky et al.~\shortcite{KovalskyGL16} achieved a fast convergence of geometric optimization by iteratively minimizing a local quadratic proxy function. Rabinovich et.al.~\shortcite{RabinovichPPS17} proposed a scalable approach to compute locally injective mappings, via local-global minimization of a reweighted proxy function.
Claici et al.~\shortcite{Claici2017} proposed a preconditioner for fast minimization of distortion energies.
Shtengel et al.~\shortcite{ShtengelPSKL17} applied the idea of majorization-minimization~\cite{Lange2004} to iteratively update and minimize a convex majorizer of the target energy in geometric optimization. Zhu et al.~\shortcite{Zhu2018-Blended} proposed a fast solver for distortion energy minimization, using a blended quadratic energy proxy together with improved line-search strategy and termination criteria.
For physical simulation, Wang~\shortcite{Wang15} proposed a Chebyshev semi-iterative acceleration technique for projective dynamics. Later, Wang and Yang~\shortcite{Wang2016} developed a GPU-friendly gradient descent method for elastic body simulation, using Jacobi preconditioning and Chebyshev acceleration. Liu et al.~\shortcite{LiuBK17} proposed an L-BFGS solver for physical simulation, with faster convergence than the projective dynamics solver from~\cite{Bouaziz2014}. Brandt et al.~\shortcite{Brandt2018-HPD} performed projective dynamics simulation in a reduced subspace, to compute fast approximate solutions for high-resolution meshes. 

\paragraph{ADMM} ADMM is a popular solver for optimization problems with separable target functions and linear side constraints~\cite{boyd2011distributed}. Using auxiliary variables and indicator functions, such formulation allows for non-smooth optimization with hard constraints, with wide applications in signal processing~\cite{Erseghe2011,Simonetto2014,Shi2014}, image processing~\cite{Figueiredo2010,Almeida2013}, computer vision~\cite{Hu2013-Fast,Liu2013-Tensor,Yang2017-Nuclear}, computational imaging~\cite{Chan2017-Plug}, automatic control~\cite{Lin2013}, and machine learning~\cite{Zhang2014-Asynchronous,Hajinezhad2016-Nonnegative}. ADMM has also been used in computer graphics to handle non-smooth optimization problems~\cite{Bouaziz2013,Neumann2013,Zhang2014-LBC,Xiong2014,Neumann2014-CMM} or to benefit from its fast initial convergence~\cite{Gregson2014,Heide2016,Xiong2017,Pan2017,Overby2017,Wang2018-Megapixel}.

ADMM was originally designed for convex optimization~\cite{gabay1975dual,fortin1983chapter,eckstein1992douglas}. For such problems, its global linear convergence has been established in~\cite{Lin2015,Deng2016-Global,Giselsson2017}, but these proofs require both terms in the target function to be convex. In comparison, our proof of global linear convergence allows for non-convex terms in the target function, which is better aligned with computer graphics problems.
In practice, ADMM works well for many non-convex problems as well~\cite{wen2012alternating,Chartrand2012,Chartrand2013,miksik2014distributed,lai2014splitting,liavas2015parallel}, but it is more challenging to establish its convergence for general non-convex problems. Only very recently have such convergence proofs been given under strong assumptions~\cite{li2015global,hong2016convergence,magnusson2016convergence,wang2019global}. We provide in this paper a general proof of convergence for non-convex problems under weaker assumptions.

It is well known that ADMM converges quickly to an approximate solution, but may take a long time to convergence to a solution of high accuracy~\cite{boyd2011distributed}. This has motivated researchers to explore acceleration techniques for ADMM. Goldstein et al.~\shortcite{Goldstein2014} and Kadkhodaie et al.~\shortcite{Kadkhodaie2015} applied Nesterov's acceleration~\cite{Nesterov83}, whereas Zhang and White~\shortcite{Zhang2018-GMRES} applied GMRES acceleration to a special class of problems where the ADMM iterates become linear. All these methods are designed for convex problems only, which limits their applicability in computer graphics.

\paragraph{Anderson acceleration} Anderson acceleration~\cite{Walker2011} is an established technique to speed up the convergence of a fixed-point iteration. It was first proposed in~\cite{Anderson1965} for solving nonlinear integral equations, and independently re-discovered later by Pulay~\shortcite{Pulay1980,Pulay1982} for accelerating the self-consistent field method in quantum chemistry. Its key idea is to utilize the $m$ previous iterates to compute a new iterate that converges faster to the fixed point. It is indeed a quasi-Newton method for finding a root of the residual function, by approximating its inverse Jacobian using previous iterates~\cite{Eyert1996,Fang2009,Rohwedder2011}. Recently, a renewed interest in this method has led to the analysis of its convergence~\cite{Toth2015,Toth2017}, as well as its application in various numerical problems~\cite{Sterck2012,Lipnikov2013,Pratapa2016,Suryanarayana2016,Ho2017}. Peng et al.~\cite{Peng2018} noted that local-global solvers in computer graphics can be treated as fixed-point iteration, and applied Anderson acceleration to improve their convergence. Additionally, to address the stability issue of classical Anderson acceleration~\cite{Walker2011,Potra2013}, they utilize the monotonic energy decrease of local-global solvers and only accept an accelerated iterate when it decreases the target energy. Fang and Saad~\shortcite{Fang2009} called classical Anderson acceleration the Type-II method in an Anderson family of multi-secant methods. Another member of the family, called the type-I method, uses quasi-Newton to approximate the Jacobian of the fixed-point residual function instead~\cite{Walker2011}, and has been analyzed recently in~\cite{Zhang2018-TypeI}.
In this paper, we focus our discussion on the type-II method.  
\section{Our Method}

\subsection{Preliminary}

\paragraph{ADMM} Let us consider an optimization problem
\begin{equation}
\min_{\mathbf{x}}~~\varPhi(\mathbf{x}, \mathbf{D} \mathbf{x} + \mathbf{h}).
\label{eq:GeneralOptimization}
\end{equation}
Here $\mathbf{x}$ can be the vertex positions of a discrete geometric shape, or the node positions of a physical system at a particular time instance. The quantity $\mathbf{D} \mathbf{x} + \mathbf{h}$ encodes a transformation of the positions $\mathbf{x}$ relevant for the optimization problem, such as the deformation gradient of each tetrahedron element in an elastic object. The notation $\varPhi(\mathbf{x}, \mathbf{D} \mathbf{x} + \mathbf{h})$ signifies that the target function contains a term that directly depends on $\mathbf{D} \mathbf{x} + \mathbf{h}$, such as elastic energy dependent on the deformation gradient.
In some applications, the optimization enforces \emph{hard constraints} on $\mathbf{x}$ or $\mathbf{D} \mathbf{x} + \mathbf{h}$, i.e., conditions that need to be strictly satisfied by the solution. Such hard constraints can be encoded using an \emph{indicator function} term within the target function. Specifically, suppose we want to enforce a condition $\mathbf{y} \in \mathcal{C}$ where $\mathbf{\mathbf{y}}$ is a subset from the components of $\mathbf{x}$ or $\mathbf{D} \mathbf{x} + \mathbf{h}$, and $\mathcal{C}$ is the \emph{feasible set}. Then we include the following term into $\varPhi$:
\[
	\sigma_{\mathcal{C}} (\mathbf{y})
	= \left\{
	\begin{array}{ll}
	0  & \textrm{if}~\mathbf{y} \in \mathcal{C}\\
	+ \infty & \textrm{otherwise}
	\end{array}
	\right..
\]
By definition, if $\mathbf{x}^\ast$ is a solution, then the corresponding components $\mathbf{y}^{\ast}$  must satisfy $\mathbf{y}^{\ast} \in \mathcal{C}$; otherwise it will result in a target function value $+\infty$ instead of the minimum. Examples of such an approach to modeling hard constraints can be found in~\cite{Deng2015}.

In many applications, the optimization problem~\eqref{eq:GeneralOptimization} can be non-linear, non-convex, and potentially non-smooth. It is challenging to solve such a problem numerically, especially when hard constraints are involved. One common technique is to introduce an auxiliary variable $\mathbf{z} = \mathbf{D} \mathbf{x} + \mathbf{h}$ to derive an equivalent problem
\begin{equation}
\min_{\mathbf{x},\mathbf{z}} ~~ \varPhi(\mathbf{x}, \mathbf{z})\quad
\textrm{s.t.} ~~ \mathbf{W}(\mathbf{z} - \mathbf{D} \mathbf{x} - \mathbf{h}) = 0,
\label{eq:ConvertedProblem}
\end{equation}
where $\mathbf{W}$ is a diagonal matrix with positive diagonal elements. $\mathbf{W}$ can be the identity matrix in the trivial case, or a diagonal scaling matrix that improves conditioning~\cite{Giselsson2017,Overby2017}.
ADMM~\cite{boyd2011distributed} is widely used to solve such problems. For ease of discussion, let us consider the problem
\begin{equation}
\min_{\mathbf{x},\mathbf{z}} ~~ \varPhi(\mathbf{x}, \mathbf{z})\quad
\textrm{s.t.} ~~ \mathbf{A} \mathbf{x} - \mathbf{B} \mathbf{z} = \mathbf{c},
\label{eq:ADMMProblem}
\end{equation}
Its solution corresponds to a stationary point of the augmented Lagrangian function
\begin{align}
L(\mathbf{x},\mathbf{z},\mathbf{u})&=\varPhi(\mathbf{x}, \mathbf{z}) + \langle \mu \mathbf{u},\mathbf{A}\mathbf{x}-\mathbf{B}\mathbf{z}-\mathbf{c} \rangle + \frac{\mu}{2} \| \mathbf{A}\mathbf{x} - \mathbf{B} \mathbf{z} - \mathbf{c} \|^2\nonumber\\
&=\varPhi(\mathbf{x}, \mathbf{z})+\frac{\mu}{2}\| \mathbf{A}\mathbf{x}-\mathbf{B}\mathbf{z}+\mathbf{u}-\mathbf{c}\|^2-\frac{\mu}{2}\| \mathbf{u}\|^2.
\end{align}
Here $\mathbf{u}$ is the \emph{dual variable} and $\mu > 0$ is the penalty parameter. Following~\cite{boyd2011distributed}, we also call $\mathbf{x}$ and $\mathbf{z}$ the \emph{primal variables}. ADMM searches for a stationary point by alternately updating $\mathbf{x}$, $\mathbf{z}$ and $\mathbf{u}$, resulting in the following iteration scheme~\cite{boyd2011distributed}:
\begin{equation}
\begin{aligned}
\mathbf{x}^{k+1}&=\argmin_{\mathbf{x}}~L(\mathbf{x}, \mathbf{z}^{k},\mathbf{u}^k),\\
\mathbf{z}^{k+1}&=\argmin_{\mathbf{z}}~L(\mathbf{x}^{k+1}, \mathbf{z},\mathbf{u}^k),\\
\mathbf{u}^{k+1}&=\mathbf{u}^k+\mathbf{A}\mathbf{x}^{k+1}-\mathbf{B}\mathbf{z}^{k+1}-\mathbf{c}.
\end{aligned}
\label{eq:xzu}
\end{equation}
We can also update $\mathbf{z}$ before $\mathbf{x}$, resulting in an alternative scheme:
\begin{equation}
\begin{aligned}
\mathbf{z}^{k+1}&=\argmin_{\mathbf{z}}~L(\mathbf{x}^{k}, \mathbf{z},\mathbf{u}^k),\\
\mathbf{x}^{k+1}&=\argmin_{\mathbf{x}}~L(\mathbf{x}, \mathbf{z}^{k+1},\mathbf{u}^k),\\
\mathbf{u}^{k+1}&=\mathbf{u}^k+\mathbf{A}\mathbf{x}^{k+1}-\mathbf{B}\mathbf{z}^{k+1}-\mathbf{c}.
\end{aligned}
\label{eq:zxu}
\end{equation}
In this paper, we refer to the scheme \eqref{eq:xzu} as \xzu{} iteration, and the scheme \eqref{eq:zxu} as \zxu{} iteration. In both cases, the updates for $\mathbf{z}$ and $\mathbf{x}$ often reduce to simple subproblems that can potentially be solved in parallel.
According to~\cite{boyd2011distributed}, the optimality condition of ADMM is that both its \emph{primal residual} and \emph{dual residual} vanish. For both iteration schemes above, the primal residual is defined as
\[
	\mathbf{r}_{\textrm{p}}^{k+1} = \mathbf{A}\mathbf{x}^{k+1}-\mathbf{B}\mathbf{z}^{k+1}-\mathbf{c}.
\]
As for the dual residual: for the \xzu{} iteration it is defined as
\begin{equation}
	\mathbf{r}_{\textrm{d}}^{k+1} = \mu \mathbf{A}^T \mathbf{B} (\mathbf{z}^{k+1} - \mathbf{z}^{k}),
	\label{eq:xzudualresidual}
\end{equation}
whereas for the \zxu{} iteration it is defined as
\begin{equation}
	\mathbf{r}_{\textrm{d}}^{k+1} = \mu \mathbf{B}^T \mathbf{A} (\mathbf{x}^{k+1} - \mathbf{x}^{k}).
	\label{eq:zxudualresidual}
\end{equation}
Intuitively, the primal residual measures the violation of the linear side constraint, whereas the dual residual measures the violation of the dual feasibility condition~\cite{boyd2011distributed}. Accordingly, ADMM is terminated when both $\|\mathbf{r}_{\textrm{p}}^{k+1}\|$ and $\|\mathbf{r}_{\textrm{d}}^{k+1}\|$ are small enough.

\paragraph{Anderson acceleration} ADMM is easy to parallelize and convergences quickly to an approximate solution. However, it can take a long time to converge to a solution of high accuracy~\cite{boyd2011distributed}. In the following subsections, we will discuss how to apply Anderson acceleration~\cite{Walker2011} to improve its convergence. Anderson acceleration is a technique to speed up the convergence of a fixed-point iteration $G: \mathbb{R}^n \mapsto \mathbb{R}^n$, by utilizing the current iterate as well as $m$ previous iterates. Let $\mathbf{q}^{k-m}, \mathbf{q}^{k-m+1}, \ldots,  \mathbf{q}^{k}$ be the latest $m+1$ iterates, and denote their residuals under mapping $G$ as $F^{k-m}, F^{k-m+1}, \ldots, F^{k}$, where $F^j = G(\mathbf{q}^{j}) - \mathbf{q}^{j}$ ($j = k-m, \ldots, k$). Then the accelerated iterate is computed as
\begin{eqnarray}
	\mathbf{q}_{\textrm{AA}}^{k+1} &=&
	(1 - \beta) \left(\mathbf{q}^k - \sum_{j=1}^m \theta_j^\ast (\mathbf{q}^{k-j+1} - \mathbf{q}^{k-j})\right)\nonumber\\
	&&+ ~\beta \left(G(\mathbf{q}^k) - \sum_{j=1}^m \theta_j^\ast (G(\mathbf{q}^{k-j+1}) - G(\mathbf{q}^{k-j}))\right),
	\label{eq:AndersonAcceleration}
\end{eqnarray}
where $(\theta_1^\ast, \ldots, \theta_m^\ast)$ is the solution to a linear least-squares problem:
\begin{equation}
	\min_{(\theta_1, \ldots, \theta_m)}~~\left\| F^k - \sum_{j=1}^m \theta_j (F^{k-j+1} - F^{k-j}) \right\|^2.
	\label{eq:AALS}
\end{equation}
In Eq.~\eqref{eq:AndersonAcceleration}, $\beta \in (0, 1]$ is a mixing parameter, and is typically set to~1~\cite{Walker2011}. We follow this convention throughout this paper.
Previously, Anderson acceleration has been applied to speed up local-global solvers in computer graphics~\cite{Peng2018}.

\subsection{Anderson acceleration of ADMM: the general approach}
\label{sec:GeneralAA}
\begin{algorithm}[t]
	\KwData{
		\hspace*{1ex}$\mathbf{x}^0$, $\mathbf{z}^0$, $\mathbf{u}^0$: initial values of variables;\\
		\hspace*{1ex}$L$: the augmented Lagrangian function;\\
		\hspace*{1ex}$m$: the number of previous iterates used for acceleration;\\
		\hspace*{1ex}$\textrm{AA}(\mathcal{G},\mathcal{F})$: Anderson accleration from a sequence $\mathcal{G}$ of fixed-point mapping results of previous iterates, and a sequence $\mathcal{F}$ of their corresponding fixed-point residuals;\\
		\hspace*{1ex}$I_{\max{}}$: the maximum number of iterations;\\
		\hspace*{1ex}$\varepsilon$: convergence threshold for combined residual.
	}
	\BlankLine
	$\mathbf{x}_{\textrm{default}} = \mathbf{x}^0$; {~~}
	$\mathbf{z}_{\textrm{default}} = \mathbf{z}^0 $; {~~}
	$\mathbf{u}_{\textrm{default}} = \mathbf{u}^0$\;
	$r_{\mathrm{prev}} = +\infty$; {~~}
	$j=0$; {~~}
	reset = TRUE; {~~}
	$k=0$\;
	\While{TRUE}{
		\algocomment{Run one iteration of ADMM}
		$\mathbf{x}_{\star} = \argmin_{\mathbf{x}} L(\mathbf{x}, \mathbf{z}^{k},\mathbf{u}^k)$\;
		$\mathbf{z}_{\star} = \argmin_{\mathbf{z}} L(\mathbf{x}_{\star}, \mathbf{z},\mathbf{u}^k)$\;
		$\mathbf{u}_{\star} = \mathbf{u}^k+\mathbf{A}\mathbf{x}_{\star}-\mathbf{B}\mathbf{z}_{\star}-\mathbf{c}$\;
		
		\algospace
		\algocomment{Compute the combined residual}
		$r = \|\mathbf{A}\mathbf{x}_{\star}-\mathbf{B}\mathbf{z}_{\star}-\mathbf{c}\|^2 + \|\mathbf{B} (\mathbf{z}_{\star} - \mathbf{z}^{k})\|^2$\;
		\eIf{ reset == TRUE \textbf{OR} $r < r_{\mathrm{prev}}$}
		{
			\algocomment{Record the latest accepted iterate}
			$\mathbf{x}_{\textrm{default}} = \mathbf{x}_{\star}$; {~~}
			$\mathbf{z}_{\textrm{default}} = \mathbf{z}_{\star}$; {~~}
			$\mathbf{u}_{\textrm{default}} = \mathbf{u}_{\star}$\;
			$r_{\mathrm{prev}} = r$; {~~} reset = FALSE\;
			\algospace
			\algocomment{Compute the accelerated iterate}
			$\mathbf{g}_j= (\mathbf{z}_{\star}, \mathbf{u}_{\star})$; \quad $\mathbf{f}_j= (\mathbf{z}_{\star} - \mathbf{z}^k, \mathbf{u}_{\star} - \mathbf{u}^k)$\;
			$j = j + 1$; {~~} $\overline{m} = \min(m - 1, j)$\;
			$(\mathbf{z}^{k+1}, \mathbf{u}^{k+1}) = \textrm{AA}\left([\mathbf{g}_j, \ldots, \mathbf{g}_{j-\overline{m}}], [\mathbf{f}_j, \ldots, \mathbf{f}_{j-\overline{m}}]\right)$\;
			$k = k+1$;
		}
		{
			\algocomment{Revert to the last accepted iterate}
			$\mathbf{z}^{k} = \mathbf{z}_{\textrm{default}}$; {~~}
			$\mathbf{u}^{k} = \mathbf{u}_{\textrm{default}}$; {~~}
			reset = TRUE;
		}
		\If(\tcp*[f]{Check termination}){$k \geq I_{\max{}}$ \textbf{OR} $r < \varepsilon$}
		{
			\Return $\mathbf{x}_{\textrm{default}}$;
			\tcp*[f]{Return the last accepted $\mathbf{x}$}
		}
	}
	\caption{Anderson acceleration for ADMM with \xzu{} iteration.}
	\label{algo:xzuGeneralAA}
\end{algorithm}
To speed up ADMM with Anderson acceleration, we must first define its iteration scheme as a fixed-point iteration. For the \xzu{} iteration, we note that $\mathbf{x}^{k+1}$ is dependent only on $\mathbf{z}^{k}$ and $\mathbf{u}^k$. Therefore, by treating $\mathbf{x}^{k+1}$ as a function of $(\mathbf{z}^{k}, \mathbf{u}^k)$, we can rewrite $\mathbf{z}^{k+1}$, and subsequently $\mathbf{u}^{k+1}$, as a function of $(\mathbf{z}^{k}, \mathbf{u}^k)$ as well. In this way, the \xzu{} iteration can be treated as a fixed-point iteration of $(\mathbf{z}, \mathbf{u})$:
\[
	(\mathbf{z}^{k+1}, \mathbf{u}^{k+1}) = G(\mathbf{z}^{k}, \mathbf{u}^{k}).
\]
Similarly, we can treat the \zxu{} scheme as a fixed-point iteration of $(\mathbf{x}, \mathbf{u})$.
In addition, to ensure stability for Anderson acceleration, we should define criteria to evaluate the effectiveness of an accelerated iterate, as well as a fall-back strategy when the criteria are not met. Goldstein et al.~\shortcite{Goldstein2014} pointed out that if the problem is convex, then its \emph{combined residual} is monotonically decreased by ADMM. For the \xzu{} iteration, the combined residual is defined as
\begin{equation}
	r_{\textrm{\xzu{}}}^{k+1} = \mu \|\mathbf{A}\mathbf{x}^{k+1}-\mathbf{B}\mathbf{z}^{k+1}-\mathbf{c}\|^2 + \mu \|\mathbf{B} (\mathbf{z}^{k+1} - \mathbf{z}^{k})\|^2.
	\label{eq:xzucombinedresidual}
\end{equation}
Here the first term is a measure of the primal residual, whereas the second term is related to the dual residual~\eqref{eq:xzudualresidual} but without the matrix $\mathbf{A}^T$. The combined residual for the \zxu{} iteration is defined as
\begin{equation}
	r_{\textrm{\zxu{}}}^{k+1} = \mu \|\mathbf{A}\mathbf{x}^{k+1}-\mathbf{B}\mathbf{z}^{k+1}-\mathbf{c}\|^2 + \mu \|\mathbf{A} (\mathbf{x}^{k+1} - \mathbf{x}^{k})\|^2.
	\label{eq:zxucombinedresidual}
\end{equation}
Although \cite{Goldstein2014} only proved the monotonic decrease of the combined residual for convex problems, our experiments show that the combined residual is decreased by the majority of iterates from the non-convex ADMM solvers considered in this paper. Indeed, if ADMM converges to a solution, then both the primal residual $\mathbf{A}\mathbf{x}^{k+1}-\mathbf{B}\mathbf{z}^{k+1}-\mathbf{c}$ and the variable changes $\mathbf{z}^{k+1} - \mathbf{z}^{k}$ and $\mathbf{x}^{k+1} - \mathbf{x}^{k}$ must converge to zero, so the combined residual must converge to zero as well. Therefore, we evaluate the effectiveness of an accelerated iterate by checking whether it decreases the combined residual compared with the previous iteration, and revert to the un-accelerated ADMM iterate if this is not the case.

\begin{algorithm}[t]
	$\mathbf{x}_{\textrm{default}} = \mathbf{x}^0$; \hfill $\mathbf{u}_{\textrm{default}} = \mathbf{u}^0$; \hfill $r_{\mathrm{prev}} = +\infty$; \hfill $j=0$; \hfill reset = TRUE; \hfill $k = 0$\;
	\While{TRUE}{
		$\mathbf{z}_{\star} = \argmin_{\mathbf{z}} L(\mathbf{x}^k, \mathbf{z},\mathbf{u}^k)$\;
		$\mathbf{x}_{\star} = \argmin_{\mathbf{x}} L(\mathbf{x}, \mathbf{z}_{\star},\mathbf{u}^k)$\;
		$\mathbf{u}_{\star} = \mathbf{u}^k+\mathbf{A}\mathbf{x}_{\star}-\mathbf{B}\mathbf{z}_{\star}-\mathbf{c}$\;
		$r = \|\mathbf{A}\mathbf{x}_{\star}-\mathbf{B}\mathbf{z}_{\star}-\mathbf{c}\|^2 + \|\mathbf{A} (\mathbf{x}_{\star} - \mathbf{x}^{k})\|^2$\;
		\eIf{ reset == TRUE \textbf{OR} $r < r_{\mathrm{prev}}$}
		{
			$\mathbf{x}_{\textrm{default}} = \mathbf{x}_{\star}$; {~~}
			$\mathbf{u}_{\textrm{default}} = \mathbf{u}_{\star}$; {~~}
			$r_{\mathrm{prev}} = r$; {~~} reset = FALSE\;
			$j = j + 1$; {~~} $\overline{m} = \min(m - 1, j)$\;
			$\mathbf{g}_j= (\mathbf{x}_{\star},\mathbf{u}_{\star})$; \quad $\mathbf{f}_j= (\mathbf{x}_{\star} - \mathbf{x}^k, \mathbf{u}_{\star} - \mathbf{u}^k)$\;
			$(\mathbf{x}^{k+1}, \mathbf{u}^{k+1}) = \textrm{AA}\left([\mathbf{g}_j, \ldots, \mathbf{g}_{j-\overline{m}}], [\mathbf{f}_j, \ldots, \mathbf{f}_{j-\overline{m}}]\right)$\;
			$k = k + 1$;
		}
		{
			$\mathbf{x}^{k} = \mathbf{x}_{\textrm{default}}$; {~~}
			$\mathbf{u}^{k} = \mathbf{u}_{\textrm{default}}$; {~~}
			reset = TRUE;
		}
		\If{$k \geq I_{\max{}}$ \textbf{OR} $r < \varepsilon$}
		{
			\Return $\mathbf{x}_{\textrm{default}}$;
		}
	}
	\caption{Anderson acceleration for ADMM with \zxu{} iteration.}
	\label{algo:zxuGeneralAA}
\end{algorithm}

Algorithm~\ref{algo:xzuGeneralAA} summarizes our Anderson acceleration approach for the \xzu{} iteration. Note that the evaluation of combined residual requires computing the change of $\mathbf{z}$ in one un-accelerated ADMM iteration. However, given an accelerated iterate $(\mathbf{z}_{\textrm{AA}}, \mathbf{u}_{\textrm{AA}})$, it is often difficult to find a pair $(\mathbf{z}_{\dagger}, \mathbf{u}_{\dagger})$ that leads to $(\mathbf{z}_{\textrm{AA}}, \mathbf{u}_{\textrm{AA}})$ after one ADMM iteration (i.e., $(\mathbf{z}_{\textrm{AA}}, \mathbf{u}_{\textrm{AA}}) = G(\mathbf{z}_{\dagger}, \mathbf{u}_{\dagger})$). Therefore, we run one ADMM iteration on $(\mathbf{z}_{\textrm{AA}}, \mathbf{u}_{\textrm{AA}})$ instead, and use the resulting values $(\mathbf{z}_{\star}, \mathbf{u}_{\star}) = G (\mathbf{z}_{\textrm{AA}}, \mathbf{u}_{\textrm{AA}})$ to evaluate the combined residual. If the accelerated iterate is accepted, then the computation of $(\mathbf{z}_{\star}, \mathbf{u}_{\star})$ can be reused in the next step of the algorithm and incurs no overhead. We can derive an acceleration method for the \xzu{} iteration in a similar way, by swapping $\mathbf{x}$ and $\mathbf{z}$ and adopting Eq.~\eqref{eq:zxudualresidual} for the computation of combined residual, as summarized in Algorithm~\ref{algo:zxuGeneralAA}.

\begin{remark}
	If the target function $\varPhi$ contains an indicator function for a hard constraint on the primal variable updated in the second step of an ADMM iteration (i.e., $\mathbf{z}$ in the \xzu{} iteration, or $\mathbf{x}$ in the \zxu{} iteration), then after each iteration this variable must satisfy the hard constraint. However, as Anderson acceleration computes the accelerated iterate via an affine combination of previous iterates, the accelerated $\mathbf{z}_{\textrm{AA}}$ or $\mathbf{x}_{\textrm{AA}}$ may violate the constraint unless its feasible set is an affine space. In other words, the accelerated iterate may not correspond to a valid ADMM iteration, and may cause issues if it is used as a solution. Therefore, to apply Anderson acceleration, we should ensure that $\varPhi$ contains no indicator function associated with the primal variable updated in the second step of the original ADMM iteration. This does not limit the applicability of our method, because it can always be achieved by introducing auxiliary variables and choosing an appropriate iteration scheme. The simulation in Fig.~\ref{fig:WindyFlag} is an example of changing the iteration scheme to allow acceleration.
	\label{remark:NoConstraint}
\end{remark}

\subsection{ADMM with a separable target function}
\label{sec:SeparableADMM}
The general approach in Section~\ref{sec:GeneralAA} does not assume any special structure of the target function. When the target function terms for $\mathbf{x}$ and $\mathbf{z}$ are separable, it is possible to improve the efficiency of acceleration further. To this end, we consider the following problem
\begin{equation}
\min_{\mathbf{x},\mathbf{z}} ~~ f(\mathbf{x}) + g(\mathbf{z}), \quad
\textrm{s.t.} ~~ \mathbf{A} \mathbf{x} - \mathbf{B} \mathbf{z} = \mathbf{c}.
\label{eq:SeparableADMMProblem}
\end{equation}
Moreover, we assume this problem satisfies the following properties:
\begin{assumption}
	Matrix $\mathbf{B}$ is invertible.
	\label{assump:InvertibleB}
\end{assumption}
\begin{assumption}
	$f(\mathbf{x})$ is a strongly convex quadratic function
	\begin{equation}
	f(\mathbf{x}) = \frac{1}{2} (\mathbf{x} - \tilde{\mathbf{x}})^T \mathbf{G} (\mathbf{x} - \tilde{\mathbf{x}}),
	\label{eq:fx}
	\end{equation}
	where $\tilde{\mathbf{x}}$ is a constant and $\mathbf{G}$ is a symmetric positive definite matrix.
	\label{assump:SPDG}
\end{assumption}
One example of such optimization is the implicit time integration of elastic bodies in~\cite{Overby2017}, where $\tilde{\mathbf{x}}$ is the predicted values of node positions ${\mathbf{x}}$ without internal forces, $\mathbf{G} = \mathbf{M}/\Delta t^2$ where $\mathbf{M}$ is the mass matrix and $\Delta t$ is the integration time step, the auxiliary variable $\mathbf{z}$ stacks the deformation gradient of each element, and $g(\mathbf{z})$ sums the elastic potential energy for all elements.
For the problem \eqref{eq:SeparableADMMProblem}, the \xzu{} iteration of ADMM becomes
\begin{align}
\mathbf{x}^{k+1}&=(\mathbf{G}+\mu\mathbf{A}^T\mathbf{A})^{-1}(\mathbf{G}\tilde{\mathbf{x}}+\mu\mathbf{A}^T(\mathbf{B}\mathbf{z}^k+\mathbf{c}-\mathbf{u}^k)),\label{eq:Sepxzu_x}\\
\mathbf{z}^{k+1}&=\argmin_{\mathbf{z}}\left(g(\mathbf{z})+\frac{\mu}{2}\| \mathbf{A}\mathbf{x}^{k+1}-\mathbf{B}\mathbf{z}-\mathbf{c}+\mathbf{u}^k\|^2\right),\label{eq:Sepxzu_z}\\
\mathbf{u}^{k+1}&=\mathbf{u}^k+\mathbf{A}\mathbf{x}^{k+1}-\mathbf{B}\mathbf{z}^{k+1}-\mathbf{c}.
\label{eq:Sepxzu_u}
\end{align}
And the \zxu{} iteration becomes
\begin{align}
\mathbf{z}^{k+1}&=\argmin_{\mathbf{z}} \left(g(\mathbf{z})+\frac{\mu}{2}\| \mathbf{A}\mathbf{x}^{k}-\mathbf{B}\mathbf{z}-\mathbf{c}+\mathbf{u}^k\|^2 \right),\label{eq:Sepzxu_z}\\
\mathbf{x}^{k+1}&=(\mathbf{G}+\mu\mathbf{A}^T\mathbf{A})^{-1}(\mathbf{G}\tilde{\mathbf{x}}+\mu\mathbf{A}^T(\mathbf{B}\mathbf{z}^{k+1}+\mathbf{c}-\mathbf{u}^k)),\label{eq:Sepzxu_x}\\
\mathbf{u}^{k+1}&=\mathbf{u}^k+\mathbf{A}\mathbf{x}^{k+1}-\mathbf{B}\mathbf{z}^{k+1}-\mathbf{c}.
\label{eq:Sepzxu_u}
\end{align}
Similar to Remark~\ref{remark:NoConstraint}, we assume that the target function contains no indicator function for the primal variable updated in the second step.
The general acceleration algorithms in Section~\ref{sec:GeneralAA} treat ADMM as a fixed-point iteration of $(\mathbf{z}, \mathbf{u})$ or $(\mathbf{x}, \mathbf{u})$. Next, we will show that if the problem satisfies certain conditions, then ADMM becomes a fixed-point iteration of only one variable, allowing us to reduce the overhead of Anderson acceleration and improve its effectiveness.

\begin{remark}
Without assuming the convexity of function $g(\cdot)$, there may be multiple solutions for the minimization problems in \eqref{eq:Sepxzu_z} and \eqref{eq:Sepzxu_z}. Throughout this paper, we assume the solver adopts a deterministic algorithm for \eqref{eq:Sepxzu_z} and \eqref{eq:Sepzxu_z}, so that given the same values of $\mathbf{x}$ and $\mathbf{u}$ it always returns the same value of $\mathbf{z}$.
\label{remark:UniqueMinimum}
\end{remark}

\subsubsection{\xzu{} iteration}
For the \xzu{} iteration~\eqref{eq:Sepxzu_x}-\eqref{eq:Sepxzu_u}, under certain conditions $\mathbf{u}^{k+1}$ can be represented as a function of $\mathbf{z}^{k+1}$:
\begin{prop}
If the optimization problem~\eqref{eq:SeparableADMMProblem} satisfies Assumptions~\ref{assump:InvertibleB} and \ref{assump:SPDG}, and the function $g(\mathbf{z})$ is differentiable, then the \xzu{} iteration~\eqref{eq:Sepxzu_x}-\eqref{eq:Sepxzu_u} satisfies
\begin{equation}
	\mathbf{u}^{k+1} = \frac{1}{\mu}\mathbf{B}^{-T} \nabla g(\mathbf{z}^{k+1}).
	\label{eq:ufromz}
\end{equation}
\label{prop:ufromz}
\end{prop}
A proof is given in Appendix~\ref{sec:proof:ufromz}. Proposition~\ref{prop:ufromz} shows that $\mathbf{u}^{k+1}$ can be recovered from $\mathbf{z}^{k+1}$. Therefore, we can treat the \xzu{} iteration~\eqref{eq:Sepxzu_x}-\eqref{eq:Sepxzu_u} as a fixed-point iteration of $\mathbf{z}$ instead of $(\mathbf{z},\mathbf{u})$, and apply Anderson acceleration to $\mathbf{z}$ alone. From the accelerated $\mathbf{z}_{\textrm{AA}}$, we recover its corresponding dual variable $\mathbf{u}_{\textrm{AA}}$ via Eq.~\eqref{eq:ufromz}. This approach brings two major benefits. First, the main computational overhead for Anderson acceleration in each iteration is to update the normal equation system for the problem~\eqref{eq:AALS}, which involves inner products of time complexity $O(mn)$ where $n$ is the dimension of variables that undergo fixed-point iteration~\cite{Peng2018}. Since $\mathbf{B}$ is invertible, $\mathbf{u}$ and $\mathbf{z}$ are of the same dimension; thus this new approach reduces the computational cost of inner products by half. Another benefit is a more simple criterion for the effectiveness of an accelerated iterate, based on the following property:
\begin{algorithm}[t]
	$r_{\mathrm{prev}} = +\infty$; {~~} $j=0$; {~~} reset = TRUE; {~~} $k = 0$\;
	\While{TRUE}{
		\algocomment{Update $\mathbf{x}$ with~\eqref{eq:Sepxzu_x} and compute residual with~\eqref{eq:xzuResidual}}
		$\mathbf{x}^{k+1} = (\mathbf{G}+\mu\mathbf{A}^T\mathbf{A})^{-1}(\mathbf{G}\tilde{\mathbf{x}}+\mu\mathbf{A}^T(\mathbf{B}\mathbf{z}^k+\mathbf{c}-\mathbf{u}^k))$\;
		$r = \|\mathbf{A} \mathbf{x}^{k+1} - \mathbf{B} \mathbf{z}^{k} - \mathbf{c}\|$\;
		\If(\tcp*[f]{Check residual}){reset == FALSE \textbf{AND} $r \geq r_{\mathrm{prev}}$}
		{
			$\mathbf{z}^{k} = \mathbf{z}_{\textrm{default}}$\ \tcp*[r]{Revert to un-accelerated $\mathbf{z}$}
			\algocomment{Re-compute $\mathbf{u}$ and $\mathbf{x}$ with~\eqref{eq:Sepxzu_u} and \eqref{eq:Sepxzu_x}}
			$\mathbf{u}^{k} = \mathbf{u}^{k-1} + \mathbf{A}\mathbf{x}^{k}-\mathbf{B}\mathbf{z}^{k}-\mathbf{c}$\;
			$\mathbf{x}^{k+1} =  (\mathbf{G}+\mu\mathbf{A}^T\mathbf{A})^{-1}(\mathbf{G}\tilde{\mathbf{x}}+\mu\mathbf{A}^T(\mathbf{B}\mathbf{z}^k+\mathbf{c}-\mathbf{u}^k))$\;
			
			\algospace
			\algocomment{Re-compute residual}
			r = $\|\mathbf{A} \mathbf{x}^{k+1} - \mathbf{B} \mathbf{z}^k - \mathbf{c}\|$; {~~}
			reset = TRUE;
		}
		\algospace
		\algocomment{Check termination criteria}
		\If{$k + 1 \geq I_{\max{}}$ \textbf{OR} $r < \varepsilon$}
		{
			\Return $\mathbf{x}^{k+1}$\;
		}
		
		\algospace
		\algocomment{Compute un-accelerated $\mathbf{z}$ value with~\eqref{eq:Sepxzu_z}}
		$\mathbf{z}_{\textrm{default}} = \argmin_{\mathbf{z}}\left(g(\mathbf{z})+\frac{\mu}{2}\|
		\mathbf{A}\mathbf{x}^{k+1}-\mathbf{B}\mathbf{z}-\mathbf{c}+\mathbf{u}^k\|^2\right)$
		
		\algospace
		\algocomment{Compute accelerated $\mathbf{z}$ value}
		$j = j + 1$; {~~} $\overline{m} = \min(m, j)$\;
		$\mathbf{g}_j= \mathbf{z}_{\textrm{default}}$; \quad $\mathbf{f}_j= \mathbf{z}_{\textrm{default}} - \mathbf{z}^k$\;
		$\mathbf{z}^{k+1} = \textrm{AA}\left([\mathbf{g}_j, \ldots, \mathbf{g}_{j-\overline{m}}], [\mathbf{f}_j, \ldots, \mathbf{f}_{j-\overline{m}}]\right)$\;
		
		\algospace
		\algocomment{Recover compatible $\mathbf{u}$ value with~\eqref{eq:ufromz}}
		$\mathbf{u}^{k+1} = \frac{1}{\mu}\mathbf{B}^{-T} \nabla g(\mathbf{z}^{k+1})$\;
		
		\algospace
		$k = k + 1$; {~~} $r_{\mathrm{prev}} = r$\;
	}
	\caption{Anderson acceleration for ADMM with \xzu{} iteration, on a problem~\eqref{eq:SeparableADMMProblem} that satisfies Assumptions~\ref{assump:InvertibleB}, \ref{assump:SPDG} and with a differentiable $g$.}
	\label{algo:xzuSeparableAA}
\end{algorithm}
\begin{prop}
Suppose the problem~\eqref{eq:SeparableADMMProblem} satisfies Assumptions~\ref{assump:InvertibleB} and \ref{assump:SPDG}, and the function $g(\mathbf{z})$ is differentiable. Let $\mathbf{z}^{k+1} = G_{\textrm{xzu}}(\mathbf{z}^{k})$ denote the fixed-point iteration of $\mathbf{z}$ induced by the \xzu{} iteration~\eqref{eq:Sepxzu_x}-\eqref{eq:Sepxzu_u}. Then $\mathbf{z}^{k+1}$ is a fixed point of mapping $G_{\textrm{xzu}}(\cdot)$ if and only if
\begin{equation}
	\mathbf{A} \mathbf{x}^{k+2} - \mathbf{B} \mathbf{z}^{k+1} - \mathbf{c} = \mathbf{0}.
	\label{eq:xzuResidual}
\end{equation}
\label{Prop:xzuResidual}
\end{prop}
A proof is given in Appendix~\ref{sec:proof:xzuResidual}. Note that the left-hand side of~\eqref{eq:xzuResidual} has a similar form as the primal residual, but involves the value of $\mathbf{x}$ in the next iteration.
Accordingly, we evaluate the effectiveness of an accelerated iterate $\mathbf{z}_{\textrm{AA}}$ and its corresponding dual variable $\mathbf{u}_{\textrm{AA}}$ by first computing a new value $\mathbf{x}_{\star}$ according to the $\mathbf{x}$-update step~\eqref{eq:Sepxzu_x}, then evaluating a residual
$\hat{\mathbf{r}}_{\textrm{\xzu{}}}  = \mathbf{A} \mathbf{x}_{\star} - \mathbf{B} \mathbf{z}_{\textrm{AA}} - \mathbf{c}.$
We only accept $\mathbf{z}_{\textrm{AA}}$ if it leads to a smaller norm of this residual compared to the previous iteration; otherwise, we revert to the last un-accelerated iterate. If  $\mathbf{z}_{\textrm{AA}}$ is accepted, then $\mathbf{x}_{\star}$ can be reused in the next step. The main benefit here is that we do not need to run an additional ADMM iteration to verify the effectiveness of $\mathbf{z}_{\textrm{AA}}$, which incurs less computational overhead when the accelerated iterate is rejected. This acceleration strategy is summarized in Algorithm~\ref{algo:xzuSeparableAA}. Fig.~\ref{fig:ThreeBars} shows an example where accelerating $\mathbf{z}$ alone leads to a faster decrease of combined residual than accelerating  $\mathbf{z}, \mathbf{u}$ together.

\subsubsection{\zxu{} iteration}
Similar to the previous discussion, when the problem satisfies certain conditions, the \zxu{} scheme is a fixed-point iteration of only one variable. In particular, we have:
\begin{prop}
	If the optimization problem~\eqref{eq:SeparableADMMProblem} satisfies Assumptions~\ref{assump:InvertibleB} and \ref{assump:SPDG}, then the \zxu{} iteration~\eqref{eq:Sepzxu_z}-\eqref{eq:Sepzxu_u} satisfies
	\begin{equation}
	\mathbf{x}^{k+1} = \tilde{x} - \mu \mathbf{G}^{-1} \mathbf{A}^T \mathbf{u}^{k+1}.
	\label{eq:xfromu}
	\end{equation}
	\label{prop:xfromu}
\end{prop}
A proof is given in Appendix~\ref{sec:proof:xfromu}. This property implies that $\mathbf{x}^{k+1}$ can be recovered from $\mathbf{u}^{k+1}$; thus we can treat the \zxu{} scheme~\eqref{eq:Sepzxu_z}-\eqref{eq:Sepzxu_u} as a fixed-point iteration of $\mathbf{u}$ instead of $(\mathbf{x}, \mathbf{u})$. In theory, we can apply Anderson acceleration to the history of $\mathbf{u}$ to obtain an accelerated iterate $\mathbf{u}_{\textrm{AA}}$, and recover the corresponding $\mathbf{x}_{\textrm{AA}}$ from Eq.~\eqref{eq:xfromu}. However, this would require solving a linear system with matrix $\mathbf{G}$, and can be computationally expensive. Instead, we note that $\mathbf{x}^{k+1}$ and $\mathbf{u}^{k+1}$ are related to by an affine map, and this relation is satisfied by any previous pair of $\mathbf{x}$ and $\mathbf{u}$ values. Then since $\mathbf{u}_{\textrm{AA}}$ is an affine combination of previous $\mathbf{u}$ values, we can apply the same affine combination coefficients to the corresponding previous $\mathbf{x}$ values to obtain $\mathbf{x}_{\textrm{AA}}$, which is guaranteed to satisfy Eq.~\eqref{eq:xfromu} with $\mathbf{u}_{\textrm{AA}}$. As the affine combination coefficients are computed from $\mathbf{u}$ only, this still reduces the computational cost compared to applying Anderson acceleration to $(\mathbf{x}, \mathbf{u})$.
Similar to the \xzu{} case, we can verify the convergence of the \zxu{} iteration by comparing $\mathbf{x}$ in the current iteration with the value of $\mathbf{z}$ in the next iteration:
\begin{prop}
	Suppose the problem~\eqref{eq:SeparableADMMProblem} satisfies Assumptions~\ref{assump:InvertibleB} and \ref{assump:SPDG}. Let $\mathbf{u}^{k+1} = G_{\textrm{zxu}}(\mathbf{u}^{k})$ denote the fixed-point iteration of $\mathbf{u}$ induced by the \xzu{} iteration~\eqref{eq:Sepzxu_z}-\eqref{eq:Sepzxu_u}. Then $\mathbf{u}^{k+1}$ is a fixed point of mapping $G_{\textrm{zxu}}(\cdot)$ if and only if
	\begin{equation}
	\mathbf{A} \mathbf{x}^{k+1} - \mathbf{B} \mathbf{z}^{k+2} - \mathbf{c} = \mathbf{0}.
	\label{eq:zxuResidual}
	\end{equation}
	\label{Prop:zxuResidual}
\end{prop}
Accordingly, we evaluate the effectiveness of $\mathbf{u}_{\textrm{AA}}$ and $\mathbf{x}_{\textrm{AA}}$ by computing from them a $\mathbf{z}_{\star}$ using Eq.~\eqref{eq:Sepzxu_z}, and evaluating the residual
$\hat{\mathbf{r}}_{\textrm{\zxu{}}} = \mathbf{A} \mathbf{x}_{\textrm{AA}} - \mathbf{B} \mathbf{z}_{\star} - \mathbf{c}$.
We accept $\mathbf{u}_{\textrm{AA}}$ if the norm of this residual is smaller than the previous iteration, and revert to the last un-accelerated iterate otherwise. If $\mathbf{u}_{\textrm{AA}}$ is accepted, then $\mathbf{z}_{\star}$ is reused in the next step. Algorithm~\ref{algo:zxuSeparableAA} summarizes our approach.

\begin{remark}
We have shown that ADMM can be reduced to a fixed-point iteration of the second primal variable or the dual variable based on Assumptions~\ref{assump:InvertibleB} and \ref{assump:SPDG}, and (for the \xzu{} iteration) the smoothness of $g$. In fact, these assumptions can be further relaxed. We refer the reader to Appendix~\ref{sec:FurtherDiscussion} for more details. Fig.~\ref{fig:AAQP} is an example of using such relaxed conditions to reduce the fixed-point iteration to one variable.
\end{remark}

\subsection{Convergence analysis}
\label{sec:ConvergenceAnalysis}

\begin{algorithm}[t]
	$r_{\mathrm{prev}} = +\infty$; {~~} $j=0$; {~~} reset = TRUE; {~~} $k = 0$\;
	\While{TRUE}{
		\algocomment{Update $\mathbf{z}$ with~\eqref{eq:Sepzxu_z} and compute residual with~\eqref{eq:zxuResidual}}
		$\mathbf{z}^{k+1}=\argmin_{\mathbf{z}}\left(g(\mathbf{z})+\frac{\mu}{2}\| \mathbf{A}\mathbf{x}^{k}-\mathbf{B}\mathbf{z}-\mathbf{c}+\mathbf{u}^k\|^2\right)$\;
		$r = \|\mathbf{A} \mathbf{x}^{k} - \mathbf{B} \mathbf{z}^{k+1} - \mathbf{c}\|$\;
		\algospace
		\algocomment{Check whether the residual increases}
		\If{reset == FALSE \textbf{AND} $r \geq r_{\mathrm{prev}}$}
		{
			\algocomment{Revert to un-accelerated $\mathbf{x}, \mathbf{u}$}
			$\mathbf{x}^{k} = \mathbf{x}_{\textrm{default}}$; {~~} $\mathbf{u}^{k} = \mathbf{u}_{\textrm{default}}$\;
			$\mathbf{z}^{k+1}=\argmin_{\mathbf{z}}\left(g(\mathbf{z})+\frac{\mu}{2}\| \mathbf{A}\mathbf{x}^{k}-\mathbf{B}\mathbf{z}-\mathbf{c}+\mathbf{u}^k\|^2\right)$\;
			r = $\|\mathbf{A} \mathbf{x}^{k} - \mathbf{B} \mathbf{z}^{k+1} - \mathbf{c}\|$\;
			reset = TRUE;
		}
		\If{$k + 1 \geq I_{\max{}}$ \textbf{OR} $r < \varepsilon$}
		{
			\Return $\mathbf{x}^{k}$\;
		}
		
		\algospace
		\algocomment{Compute un-accelerated $\mathbf{x}$ and $\mathbf{u}$}
		$\mathbf{x}_{\textrm{default}} =(\mathbf{G}+\mu\mathbf{A}^T\mathbf{A})^{-1}(\mathbf{G}\tilde{\mathbf{x}}+\mu\mathbf{A}^T(\mathbf{B}\mathbf{z}^{k+1}+\mathbf{c}-\mathbf{u}^k))$\;
		$\mathbf{u}_{\textrm{default}}=\mathbf{u}^k+\mathbf{A}\mathbf{x}_{\textrm{default}} - \mathbf{B}\mathbf{z}^{k+1}-\mathbf{c}$\;
		
		\algospace
		\algocomment{Use history of $\mathbf{u}$ to compute affine coeffients}
		$j = j + 1$; {~~} $\overline{m} = \min(m, j)$\;
		$\mathbf{g}_j^{\mathbf{x}}= \mathbf{x}_{\textrm{default}}$; {~~} $\mathbf{g}_j^{\mathbf{u}}= \mathbf{u}_{\textrm{default}}$; {~~} $\mathbf{f}_j^{\mathbf{u}} = \mathbf{u}_{\textrm{default}} - \mathbf{u}^k$\;
		$(\theta_{1}^{\ast}, \ldots, \theta_{\overline{m}}^{\ast}) = \argmin\limits_{(\theta_1, \ldots, \theta_{\overline{m}})} \left\| \mathbf{f}_j^{\mathbf{u}} - \sum_{i=1}^{\overline{m}} \theta_i (\mathbf{f}_{j-i+1}^{\mathbf{u}} - \mathbf{f}_{j-i}^{\mathbf{u}}) \right\|^2$\;
		
		\algospace
		\algocomment{Compute accelerated $\mathbf{x}$ and $\mathbf{u}$ with the coefficients}
		$\mathbf{x}^{k+1} = \mathbf{g}_j^{\mathbf{x}} - \sum_{i=1}^{\overline{m}} \theta_i^\ast \left( \mathbf{g}_{j-i+1}^{\mathbf{x}} - \mathbf{g}_{j-i}^{\mathbf{x}}\right)$\;
		$\mathbf{u}^{k+1} = \mathbf{g}_j^{\mathbf{u}} - \sum_{i=1}^{\overline{m}} \theta_i^\ast \left( \mathbf{g}_{j-i+1}^{\mathbf{u}} - \mathbf{g}_{j-i}^{\mathbf{u}}\right)$\;
		
		\algospace
		$k = k + 1$; {~~} $r_{\mathrm{prev}} = r$\;
	}
	\caption{Anderson acceleration for ADMM with \zxu{} iteration, on a problem~\eqref{eq:SeparableADMMProblem} that satisfies Assumptions~\ref{assump:InvertibleB} and \ref{assump:SPDG}.}
	\label{algo:zxuSeparableAA}
\end{algorithm}

For Anderson acceleration to be applicable, an ADMM solver must be convergent already. However, many ADMM solvers used in computer graphics lack a convergence guarantee due to the non-convexity of the problems they solve.
Although ADMM works well for many non-convex problems in practice, convergence proofs on such problems rely on strong assumptions that are often not satisfied by graphics problems~\cite{li2015global,hong2016convergence,magnusson2016convergence,wang2019global}.
In this subsection, we discuss the convergence of ADMM on the problem~\eqref{eq:SeparableADMMProblem} where the term $g$ in the target function can be non-convex. We first provide a set of conditions for linear convergence of ADMM on such problems, and then give more general convergence proofs using weaker assumptions than existing results in the literature. As the problem structure~\eqref{eq:SeparableADMMProblem} is common in computer graphics, our new results can potentially expand the applicability of ADMM for graphics problems.

To ease the presentation, we first introduce some notation. To account for the fact that the target function may be unbounded from above due to an indicator function, we suppose all the functions are mappings to $\mathbb{R}\bigcup\{+\infty\}$. Following~\cite{rockafellar2015convex}, for a function $F$ we define its effective domain and level set as:
\begin{align*}
\dom{F} &\coloneqq \{\mathbf{x} \mid f(\mathbf{x})<+\infty\},\\
\lev{F}{\alpha} & \coloneqq \{\mathbf{x} \mid f(\mathbf{x})\leq\alpha\}, ~\text{given}~\alpha \in \mathbb{R}.
\end{align*}
A function $F$ is \emph{level-bounded} if $\lev{F}{\alpha}$ is a bounded set for any $\alpha \in \mathbb{R}$.
Given a set $\mathcal{S}$, let $\mathcal{I}_\mathcal{S}$ and $\mathcal{B}_\mathcal{S}$ denote the interior and the boundary of $\mathcal{S}$, respectively.
A function $F$ is continuous on $\mathbb{R}^n$ if: (i) it is continuous within $\mathcal{I}_{\dom{F}}$ in the conventional sense; and (ii) $\forall \mathbf{x}_k \rightarrow \mathbf{x} \in \mathcal{B}_{\dom{F}}$, we have $F(\mathbf{x}_k)\rightarrow F(\mathbf{x})= +\infty$.
We say a function is \emph{Lipschitz differentiable} if it is differentiable and its gradient is Lipschitz continuous.
Unless specified otherwise, $\mathbf{I}$ denotes the identity matrix and the identity map.
The symbol $\conv{\mathcal{S}}$ denotes the convex hull of a set $\mathcal{S}$,
and $\partial F$ denotes the set of all sub-differentials for a function $F$ (see~\cite[Definition 8.3(b)]{rockafellar2009variational}). For matrix $\mathbf{Q}$, we use $\rho(\mathbf{Q})$ to represent its spectral radius.
We will discuss conditions for the ADMM iterates $\{(\mathbf{x}^k, \mathbf{z}^k, \mathbf{u}^k)\}$ to converge to a stationary point $(\mathbf{x}^{\ast}, \mathbf{z}^{\ast}, \mathbf{u}^{\ast})$ of the augmented Lagrangian for problem~\eqref{eq:SeparableADMMProblem}, which is defined by the conditions~\cite{boyd2011distributed}:
\begin{equation}
\mathbf{A}\mathbf{x}^{\ast}-\mathbf{B}\mathbf{z}^{\ast}=\mathbf{c}, \quad \mathbf{0} \in \partial f(\mathbf{x}^{\ast})+\mathbf{A}^{T}\mathbf{u}^{\ast},
\quad \mathbf{0} \in \partial g(\mathbf{z}^{\ast})-\mathbf{B}^{T}\mathbf{u}^{\ast}.
\label{eq:StationaryPoint}
\end{equation}

\begin{figure*}[t!]
	\centering
	\includegraphics[width=\textwidth]{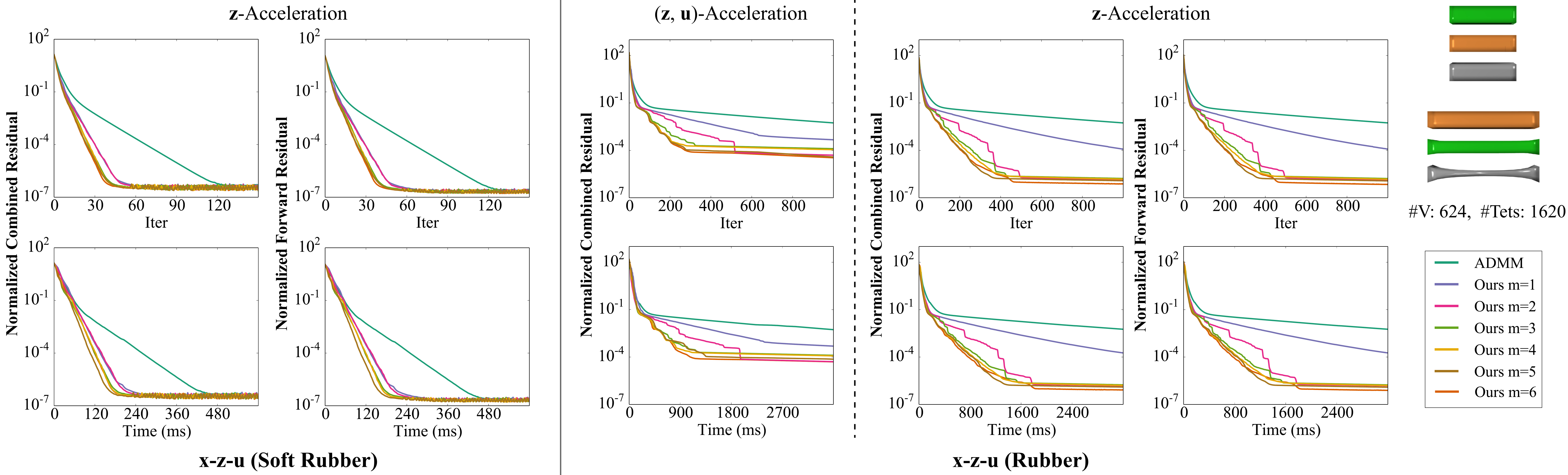}
	\caption{Comparison between the ADMM solver in~\cite{Overby2017} and our method according to Algorithm~\ref{algo:xzuSeparableAA}, for computing the same frame of a simulation sequence with three elastic bars. Two material stiffness settings (``soft rubber'' and ``rubber'') are used for testing. In both case, our method leads to faster decrease of residuals and accelerates the convergence. For the case with rubber, we also test Algorithm~\ref{algo:xzuGeneralAA} that applies Anderson acceleration to $(\mathbf{z}, \mathbf{u})$, which also speeds up the convergence but is less effective than accelerating $\mathbf{z}$ alone.}
	\label{fig:ThreeBars}
\end{figure*} 

\paragraph{Linear convergence}
Our discussion involves the following definitions related to the problem~\eqref{eq:SeparableADMMProblem} and Assumptions~\ref{assump:InvertibleB} and \ref{assump:SPDG}:
\begin{equation}
\hatg(\mathbf{z})  \coloneqq g(\mathbf{B}^{-1}\mathbf{z}), \quad
\mathbf{K}  \coloneqq \mathbf{A}\mathbf{G}^{-1}\mathbf{A}^{T}. \label{eq:DefinitionHatGAndK}
\end{equation}
We denote by $\rho(\mathbf{K})$ the spectral radius of matrix $\mathbf{K}$. To prove linear convergence of ADMM for the problem~\eqref{eq:SeparableADMMProblem} regardless of its initial value, we need the following assumption:
\begin{assumption}
	$\nabla\hat{g}$ is Lipschitz differentiable on $\mathbb{R}^n$ with a Lipschitz constant $L$, i.e.
	$\| \nabla\hat{g}(\mathbf{z}_1)-\nabla\hat{g}(\mathbf{z}_2)\|\leq L\| \mathbf{z}_1-\mathbf{z}_2\|$ $\forall \mathbf{z}_1, \mathbf{z}_1 \in \mathbb{R}^n$.
	\label{assump:GlobalConvergenceGhatAndK}
\end{assumption}
Then we have:
\begin{theorem}
	If Assumptions~\ref{assump:InvertibleB}-\ref{assump:GlobalConvergenceGhatAndK} are satisfied and $\rho(\mathbf{K})<\frac{1}{2L}$, then for a sufficiently large $\mu$ the \xzu{} iteration~\eqref{eq:Sepxzu_x}-\eqref{eq:Sepxzu_u} converges to a stationary point defined in Eq.~\eqref{eq:StationaryPoint}. Moreover,
	\[
	\| \mathbf{B}\mathbf{z}^{n+1}-\mathbf{B}\mathbf{z}^{n}\|\leq\gamma_1\| \mathbf{B}\mathbf{z}^{n}-\mathbf{B}\mathbf{z}^{n-1}\|,
	\]
	where $\gamma_1=\displaystyle\frac{\frac{\mu\rho(\mathbf{K})}{1+\mu\rho(\mathbf{K})}+\frac{L}{\mu}}{1-\frac{L}{\mu}} < 1$ is a constant.
	\label{thm:xzuConvergenceGlobalLipschitz}
\end{theorem}
\begin{theorem}
	If Assumptions~\ref{assump:InvertibleB}-\ref{assump:GlobalConvergenceGhatAndK} are satisfied, $\rho(\mathbf{K})<\frac{1}{L}$ and $\mathbf{I}-\mu\mathbf{K}$ is invertible, then for a sufficiently large $\mu$ the \zxu{} iteration~\eqref{eq:Sepzxu_z}-\eqref{eq:Sepzxu_u} converges to a stationary point defined in Eq.~\eqref{eq:StationaryPoint}. Moreover,
	\[
	\| \mathbf{v}^{k+1}-\mathbf{v}^{k}\|\leq \gamma_2\| \mathbf{v}^{k}-\mathbf{v}^{k-1}\|,
	\]
	where
	$\mathbf{v}^k=(\mathbf{I}-\mu\mathbf{K})\mathbf{u}^k$ and $\gamma_2=\frac{\mu\rho(\mathbf{K})}{1+\mu\rho(\mathbf{K})}+\frac{L}{\mu-L}<1$.
	\label{thm:zxuConvergenceGlobalLipschitz}
\end{theorem}
Proofs are provided in Appendix~\ref{sec:ConvergenceTheoremProof}. The theorems above rely on Assumption~\ref{assump:GlobalConvergenceGhatAndK} which requires the function $g$ to be globally Lipschitz differentiable. This may not be the case for some graphics problems. For example, the StVK energy used for simulation of hyperelastic materials is a quartic function of the deformation gradient, and is locally Lipschitz differentiable but not globally so. For such problems, we can still prove linear convergence with additional conditions on its initial value and penalty parameter. In the following, we use $T(\mathbf{x}, \mathbf{z})$ to denote the target function~\eqref{eq:SeparableADMMProblem}. We make the following relaxed assumption about $\hatg$:
\begin{assumption}
	\begin{itemize}
		\item[(1)] $\hatg$ is level-bounded, and $\hatg(\mathbf{z})\geq0~\forall \mathbf{z} \in \mathbb{R}^{n}$.
		\item[(2)] $\hatg$ is continuous on $\mathbb{R}^{n}$ and differentiable in $\mathcal{I}_{\dom{\hatg}}$.
		\item[(3)] $\hatg$ is Lipschitz differentiable on any compact convex set in $\dom{\hatg}$.
	\end{itemize}
	\label{assump:g}
\end{assumption}
For linear convergence of the \xzu{} iteration, we assume the following for the initial value $(\mathbf{x}^0, \mathbf{z}^0, \mathbf{u}^0)$ and penalty parameter $\mu$:
\begin{assumption}
	\begin{itemize}
		\item[(1)] $\mathbf{z}^0=\mathbf{B}^{-1}(\mathbf{A}\mathbf{x}^0-\mathbf{c}),\mathbf{u}^0=\frac{1}{\mu}\mathbf{B}^{-T}\nabla g(\mathbf{z}^{0})$. $\mathbf{z}^0\in\text{dom}(g)$.
		\item[(2)] $\mu$ is large enough such that $c_1\leq 1$, where
		\[
		c_1=\sup\limits_{\mathbf{z}\in \lev{{g}}{T^0+1}} \frac{1}{2\mu} \| \mathbf{B}^{-T}\nabla g(\mathbf{z})\|^2
		\]
		and $T^0 = T(\mathbf{x}^0,\mathbf{z}^0)$.
		Moreover, suppose $\conv{\lev{\hatg}{T^0+c_1}}\subset\dom{\hatg}$ and let $L_c$ be a Lipschitz constant of $\nabla\hatg$ over this set.
	\end{itemize}
	\label{assump:initialValue}
\end{assumption}
\begin{theorem}
	Suppose Assumptions~\ref{assump:InvertibleB}, \ref{assump:SPDG}, \ref{assump:g}, \ref{assump:initialValue} are satisfied,
	$\frac{\mu}{2}-\frac{L_c^2}{\mu}>\frac{L_c}{2}$, and $\rho(\mathbf{K})<\frac{1}{2L_c}$.
	Then for a sufficiently large $\mu$ the \xzu{} iteration~\eqref{eq:Sepxzu_x}-\eqref{eq:Sepxzu_u} converges to a stationary point defined in Eq.~\eqref{eq:StationaryPoint}, and
	$
	\| \mathbf{B}\mathbf{z}^{n+1}-\mathbf{B}\mathbf{z}^{n}\|\leq\gamma_3\| \mathbf{B}\mathbf{z}^{n}-\mathbf{B}\mathbf{z}^{n-1}\|
	$
	with $\gamma_3=\displaystyle\frac{\frac{\mu\rho(\mathbf{K})}{1+\mu\rho(\mathbf{K})}+\frac{L_c}{\mu}}{1-\frac{L_c}{\mu}}<1$.
	\label{thm:xzuConvergenceLocalLipschitz}
\end{theorem}
For the \zxu{} iteration, we need a different assumption that relies on the following proposition which is proved in Appendix~\ref{Proof_Thm3.6}:
\begin{prop}
	\label{prop:i-1}
	Let $R(\mathbf{A})$ be the range of matrix $\mathbf{A}$. Then for any $\mathbf{x} \in R(\mathbf{A})$, $\|\mathbf{K}\mathbf{x}\|\geq \eta\|\mathbf{x}\|$ where $\eta>0$ is a constant depending on $\mathbf{K}$.
\end{prop}
\begin{assumption}
	\label{initial_zxu}
	The initial value $(\mathbf{x}^0, \mathbf{z}^0, \mathbf{u}^0)$ satisfies:\begin{itemize}
		\item[(1)] $\mathbf{z}^0=\mathbf{B}^{-1}(\mathbf{A}\mathbf{x}^0-\mathbf{c}),\mathbf{x}^0=\mathbf{\tilde{x}},\mathbf{u}^0=0.$ $\mathbf{z}^0\in\text{dom}(g)$.
		\item[(2)]$\mu$ is large enough such that $c_2+c_3\leq 1$, where
		\begin{align*}
		c_2& = \sup\limits_{(\mathbf{x},\mathbf{z})\in \lev{T}{T^0+1}} \frac{2}{\eta^2\mu}\|\mathbf{A}\mathbf{x}-\mathbf{A}\mathbf{\tilde{x}}\|^2+ (\frac{2\rho(\mathbf{K})^2}{\mu\eta^2} +\frac{1}{\mu}) \| \mathbf{B}^{-T}\nabla g(\mathbf{z})\|^2,\\
		c_3& =(\frac{8\rho(\mathbf{K})^2}{\mu\eta^2}+\frac{4}{\mu})\|\mathbf{B}^{-T}\nabla g(\mathbf{z}^{0})\|^2,
		\end{align*}
		where $\eta$ is defined in Proposition~\ref{prop:i-1}.
		Moreover, let $L_d$ be a Lipschitz constant of $\nabla\hatg$ over $\conv{\lev{\hatg}{T^0+c_2+c_3}}$, and suppose $\conv{\lev{\hatg}{T^0+c_2+c_3}}\subset\dom{\hatg}$.
	\end{itemize}
\end{assumption}
\begin{theorem}
	Suppose Assumptions~\ref{assump:InvertibleB}, \ref{assump:SPDG}, ~\ref{assump:g}, \ref{initial_zxu} are satisfied, $\rho(\mathbf{K})<\frac{1}{L_d}$, and $\mathbf{I}-\mu\mathbf{K}$ is invertible. Then for a sufficiently large $\mu$ the \zxu{} iteration~\eqref{eq:Sepzxu_z}-\eqref{eq:Sepzxu_u} converges to a stationary point defined in Eq.~\eqref{eq:StationaryPoint}, and
	$
	\| \mathbf{v}^{k+1}-\mathbf{v}^{k}\|\leq \gamma_4\| \mathbf{v}^{k}-\mathbf{v}^{k-1}\|,
	$
	with $\mathbf{v}^k=(\mathbf{I}-\mu\mathbf{K})\mathbf{u}^k$ and $\gamma_4=\frac{\mu\rho(\mathbf{K})}{1+\mu\rho(\mathbf{K})}+\frac{L_d}{\mu-L_d}<1$.
	\label{thm:zxuConvergenceLocalLipschitz}
\end{theorem}
The proofs for these two theorems are given in Appendix~\ref{sec:ConvergenceTheoremProof}.

\begin{remark}
Unlike existing linear convergence proofs such as~\cite{Lin2015,Deng2016-Global,Giselsson2017}, we do not require both $f$ and $g$ to be convex.
This makes our proofs applicable to some graphics problems with a non-convex $g$, such as the elastic body simulation problem in~\cite{Overby2017} where $g$ is an elastic potential energy. In the supplementary material we provide numerical verification of  linear convergence on such a problem.
\end{remark}

\paragraph{General convergence under weak assumptions}
If a linear convergence rate is not needed, the assumptions above can be further relaxed to prove the convergence of ADMM on problem~\eqref{eq:SeparableADMMProblem}: instead of the relation between the matrix $\mathbf{K}$ and the Lipschitz constant $L$, we require the following  weak assumption on function $g$.
\begin{assumption}
	$g$ is a semi-algebraic function.
	\label{assump:semialgebraic}
\end{assumption}
A function $F: \mathbb{R}^n \mapsto \mathbb{R}$ is called semi-algebraic if its graph $\{\bigl(\mathbf{y}, F(\mathbf{y})\bigr) \mid \mathbf{y} \in \mathbb{R}^n\} \subset \mathbb{R}^{n+1}$
is a union of finitely many sets each defined by a finite number of polynomial equalities and strict inequalities~\cite{li2015global}.
This assumption covers a large range of functions used in computer graphics. For example, polynomials (such as StVK energy) and rational functions (such as NURBS) are both semi-algebraic.
Then we have:
\begin{theorem}
Suppose Assumptions~\ref{assump:InvertibleB}, \ref{assump:SPDG}, \ref{assump:g}, \ref{assump:initialValue} and \ref{assump:semialgebraic} are satisfied, and $\frac{\mu}{2}-\frac{L_c^2}{\mu}>\frac{L_c}{2}$. Then for a sufficiently large $\mu$ the \xzu{} iteration~\eqref{eq:Sepxzu_x}-\eqref{eq:Sepxzu_u} converges to a stationary point defined in Eq.~\eqref{eq:StationaryPoint}, and $\sum_{n=1}^{+\infty}\| \mathbf{z}^{k+1}-\mathbf{z}^k\|<\infty$.
\label{thm:GeneralConvergenceXZU}
\end{theorem}
\begin{theorem}
	If Assumptions~\ref{assump:InvertibleB}, \ref{assump:SPDG}, \ref{assump:g}, \ref{initial_zxu} and \ref{assump:semialgebraic} are satisfied,
	then for a sufficiently large $\mu$ the \zxu{} iteration~\eqref{eq:Sepzxu_z}-\eqref{eq:Sepzxu_u} converges to a stationary point defined in Eq.~\eqref{eq:StationaryPoint}, and $\sum_{n=1}^{+\infty}\| \mathbf{A}\mathbf{x}^{k+1}-\mathbf{A}\mathbf{x}^k\|<\infty$.
	\label{thm:GeneralConvergenceZXU}
\end{theorem}
Proofs are given in Appendix~\ref{Proof_Thm3.5} and \ref{Proof_Thm3.6}.

\begin{remark}
Compared with existing convergence results for non-convex ADMM such as~\cite{li2015global,wang2019global}, for the \xzu{} iteration we do not require the function $g$ to be globally Lipschitz differentiable, and for the \zxu{} iteration we do not require the matrix $\mathbf{A}$ to be of full row rank.  
This makes our results applicable to a wider range of problems in computer graphics. In particular, for geometry optimization, the reduction matrix $\mathbf{A}$ that relates vertex positions to auxiliary variables may not be of full row rank, potentially due to the presence of  auxiliary variables that are derived in the same way from vertex positions but involved in different constraints.
Although for the \zxu{} iteration our assumptions on $g$ are more restrictive than those in \cite{li2015global,wang2019global}, such assumptions are still general enough to be satisfied by many graphics problems.
\end{remark}


\section{Results}

\begin{figure}[t!] 
	\centering
	\includegraphics[width=\columnwidth]{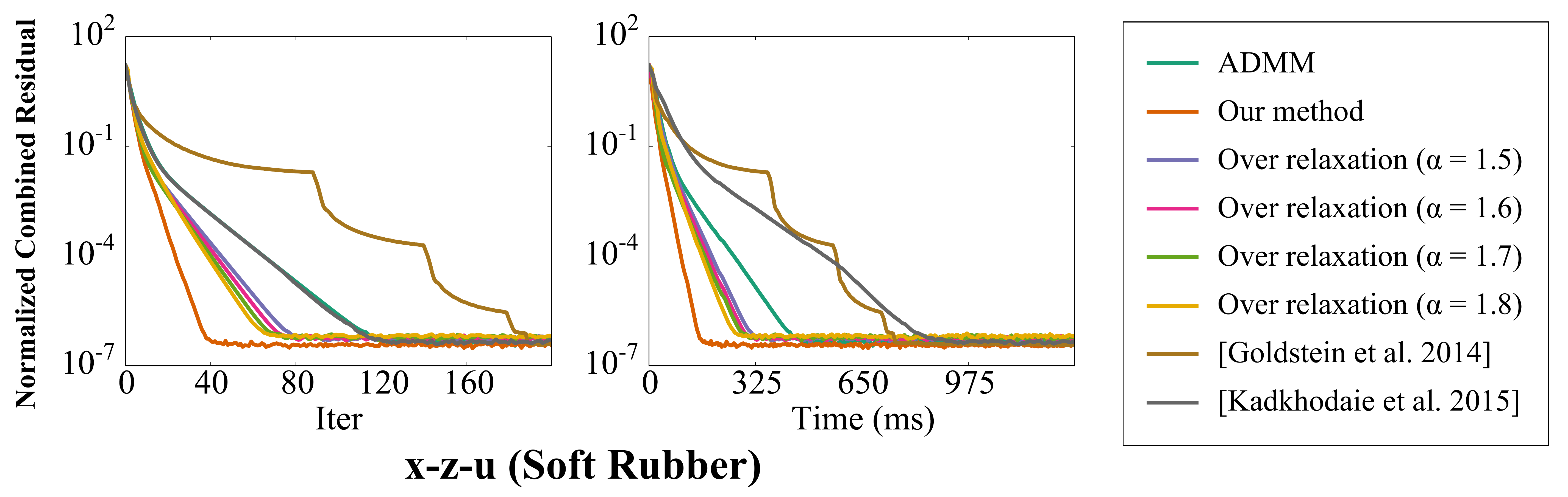}
	\caption{Comparison with other ADMM acceleration schemes on the same non-convex problem for rubber simulation as Fig.~\ref{fig:ThreeBars}. The methods from~\protect\cite{Goldstein2014} and \cite{Kadkhodaie2015}, which are designed for convex problems, are ineffective for this problem instance. Over-relaxation is effective in accelerating the convergence, but not as much as our approach.}
	\label{fig:AccelerationComparison}
\end{figure}

\begin{figure*}[t!] 
	\centering
	\includegraphics[width=\textwidth]{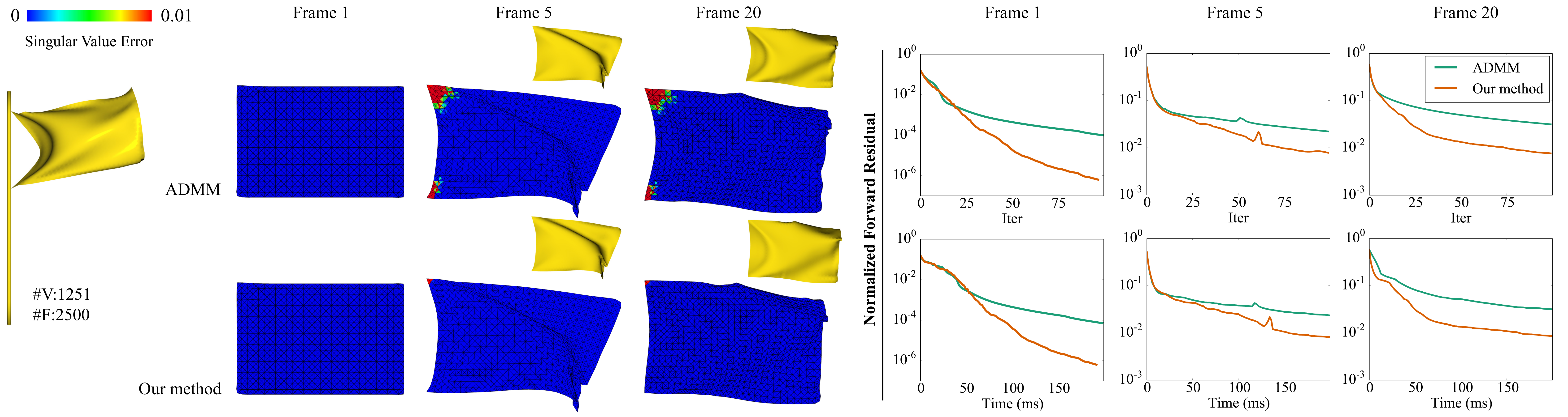}
	\caption{For the simulation of a discretized flag with hard constraints that limit its strain, our accelerated solver convergences faster than an ADMM solver. Here the color-coding shows the deviation from the deformation gradient singular values from their prescribed range. Using the same computational budget to compute a frame, the results with our solver satisfy the strain limiting constraints better.}
	\label{fig:WindyFlag}
\end{figure*}
We apply our methods to a variety of ADMM solvers in graphics. We implement Anderson acceleration following the source code released by the authors of~\cite{Peng2018}\footnote{\url{https://github.com/bldeng/AASolver}}. The source code of our implementation is available at {\url{https://github.com/bldeng/AA-ADMM}. All examples are run on a desktop PC with 32GB of RAM and a quad-core CPU at 3.6GHz. To account for the dimension and the numerical range of the variables, we use the following \emph{normalized combined residual} $R_c$ and \emph{normalized forward residual} $R_f$ to measure convergence:
\begin{equation}
	R_c = \sqrt{\frac{r_c}{N_\mathbf{z} \cdot a^2}}, \quad R_f = \sqrt{\frac{r_f}{N_\mathbf{z} \cdot a^2}},
	\label{eq:NormalizedResiduals}
\end{equation}
where $r_c$ is the combined residual computed from Eq.~\eqref{eq:xzucombinedresidual} or \eqref{eq:zxucombinedresidual}, $r_f$ is the squared norm of the residual of Eq.~\eqref{eq:xzuResidual} or \eqref{eq:zxuResidual}, $N_\mathbf{z}$ is the dimension of $\mathbf{z}$, and $a > 0$ is a scalar that indicates the typical variable range. In the following, for all physical simulation and geometry optimization problems, we set $a$ to the average edge length of the initial discretized model. For image processing problems, we simply set $a = 1$.
For the choice of parameter $m$, similar to~\cite{Peng2018} we observe that a large $m$ tends to improve the reduction of iteration count but increases the computational overhead per iteration (see Fig.~\ref{fig:ThreeBars}). We choose $m = 6$ by default.

\subsection{Physical simulation}
Overby et al.~\shortcite{Overby2017} performed physical simulation via the following optimization problem:
\begin{equation}
\min_{\mathbf{x},\mathbf{z}} ~~ f(\mathbf{x}) + g(\mathbf{z})\quad
\textrm{s.t.} ~~ \mathbf{W}(\mathbf{z} - \mathbf{D} \mathbf{x}) = 0,
\label{eq:OverbyProblem}
\end{equation}
Here $\mathbf{x}$ is the node positions of the discretized object, $f(\mathbf{x})$ is a momentum energy of the form~\eqref{eq:fx} with $\mathbf{G}$ being a scaled mass matrix, $\mathbf{D} \mathbf{x}$ collects the deformation gradient of each element, $g(\mathbf{z})$ is the elastic potential energy, and $\mathbf{W}$ is a diagonal scaling matrix that improves conditioning. This problem is solved in~\eqref{eq:OverbyProblem} using ADMM with the \xzu{} iteration. As it satisfies the assumptions in Proposition~\ref{prop:ufromz}, we apply Anderson acceleration to variable $\mathbf{z}$ according to Algorithm~\ref{algo:xzuSeparableAA}. Our method is implemented based on the source code released by the authors of~\cite{Overby2017}\footnote{\url{https://github.com/mattoverby/admm-elastic}}.  
Fig.~\ref{fig:ThreeBars} compares the simulation performance on three elastic bars subject to horizontal external forces on their two ends. We use the same material stiffness for all bars, and a different elastic potential energy model for each bar (corotational, StVK and neo-Hookean, respectively). 
We apply the original solver and our solver with different $m$ values to the same problem for a particular frame, and plot their normalized combined residuals and normalized forward residuals through the iterations. 
The methods are compared on two types of material stiffness (``soft rubber'' and ``rubber'' as defined in the code from~\cite{Overby2017}, with the latter one being stiffer).
Our method decreases both residuals much faster than the original ADMM solver for each stiffness settings. Moreover, these two residuals are highly correlated, which demonstrates the effectiveness of using the forward residual to verify accelerated iterates according to Proposition~\ref{Prop:xzuResidual}.  
On the rubber models, we also evaluate the performance of the general approach in Algorithm~\ref{algo:xzuGeneralAA} that accelerates $\mathbf{z}$ and $\mathbf{u}$ together.
We can see that accelerating $\mathbf{z}$ alone leads to a faster decrease of the combined residual. One possible reason is that Algorithm~\ref{algo:xzuSeparableAA} explicitly enforces the compatibility condition~\eqref{eq:ufromz}, so that the accelerated $\mathbf{z}$ and the recovered $\mathbf{u}$ always correspond to a valid intermediate value for a certain ADMM iterate sequence. This property does not hold for the general approach, since it only performs affine combination to obtain the accelerated $\mathbf{z}$ and $\mathbf{u}$, which is more akin to finding a new initial value for an ADMM sequence.
In Fig.~\ref{fig:AccelerationComparison}, we use the same soft rubber simulation problem to compare our method with existing ADMM acceleration techniques, including~\cite{Goldstein2014} and \cite{Kadkhodaie2015} which combined Nesterov's acceleration scheme with a restarting rule based on combined residual, as well as over-relaxation~\cite{eckstein1992douglas} with a relaxation parameter $\alpha \in [1.5, 1.8]$ as explained in~\cite[\S3.4.3]{boyd2011distributed}. As \cite{Goldstein2014,Kadkhodaie2015} rely on the convexity of the problem, they are ineffective for this non-convex problem and in fact increases the computational time. Although over-relaxation speeds up the decrease of residual, it achieves less acceleration than our method.

\begin{figure}[t!] 
	\centering
	\includegraphics[width=\columnwidth]{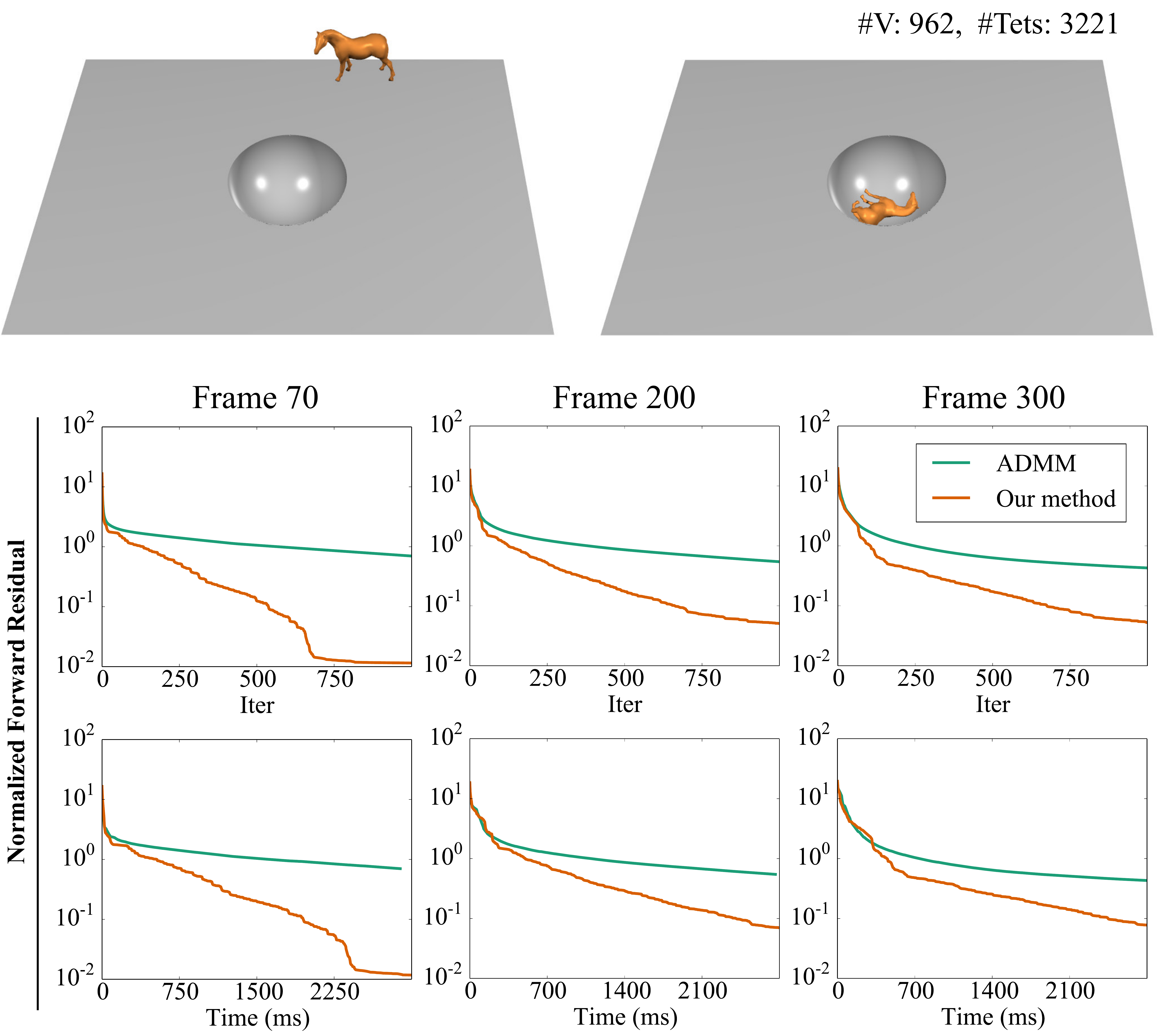}
	\caption{Simulation of a falling horse, with hard constraints on node positions that prevent them from penetrating the static objects. Our method achieves faster convergence than ADMM, as shown by the plots of normalized forward residual for three frames.}
	\label{fig:HorseHole}
\end{figure}

The solver in~\cite{Overby2017} allows enforcing hard constraints on node positions. Our method can be applied in such cases as well. In Fig.~\ref{fig:WindyFlag}, we simulate the movement of a triangulated flag under the wind force. Within $g(\mathbf{z})$, the elastic potential energy for each triangle is defined as the squared Euclidean distance from its deformation gradient to the closest rotation matrix. In addition, $g(\mathbf{z})$ contains an indicator function term for the strain limit of each triangle that requires all singular values of the deformation gradient to be within the range  $[0.95, 1.05]$.
Due to such hard constraints for $\mathbf{z}$, we cannot apply our method to the \xzu{} iteration (see Remark~\ref{remark:NoConstraint}). Instead, we adopt the \zxu{} iteration and apply Algorithm~\ref{algo:zxuSeparableAA} to accelerate $\mathbf{u}$ alone, because the iteration satisfies the assumptions in Proposition~\ref{prop:xfromu}. 
We compare the original ADMM solver with our accelerated solver with $m=6$.
To this end, we first apply our solver to compute a simulation sequence, and then re-solve the optimization problem using the original ADMM solver. 
Fig.~\ref{fig:WindyFlag} plots the normalized forward residual from each solver on three frames, where we see a faster decrease of the residual using our solver. 
In addition, for these three frames we take the results from both solvers within the same computational time, and use color-coding to illustrate the maximum deviation of its deformation gradient singular values from the prescribed range on each triangle.
We can see that our solver leads to better satisfaction of the strain limiting constraints.

Hard constraints are also used in~\cite{Overby2017} to handle collision between objects. In Figs.~\ref{fig:HorseHole} and \ref{fig:HorseObjects}, we apply our method in such scenarios. Here an elastic solid horse model falls under gravity and collides with static objects in the scene. In~\cite{Overby2017}, this is handled by enforcing hard constraints on $\mathbf{x}$ that prevent the nodes from penetrating the static objects. As this would reduce the $\mathbf{x}$-update step to a time-consuming quadratic programming problem, \cite{Overby2017} linearizes the constraints and solve the resulting linear system. 
However, with such modification it is no longer an ADMM algorithm.
Therefore, we apply the constraints on $\mathbf{z}$ instead and solve the problem using \zxu{} iteration, with acceleration according to Algorithm~\ref{algo:zxuSeparableAA}. Figs.~\ref{fig:HorseHole} and \ref{fig:HorseObjects} plot the normalized forward residual for computing certain frames in the simulation sequence, showing a faster decrease of the residual with our method.

\begin{figure}[t!] 
	\centering
	\includegraphics[width=0.9\columnwidth]{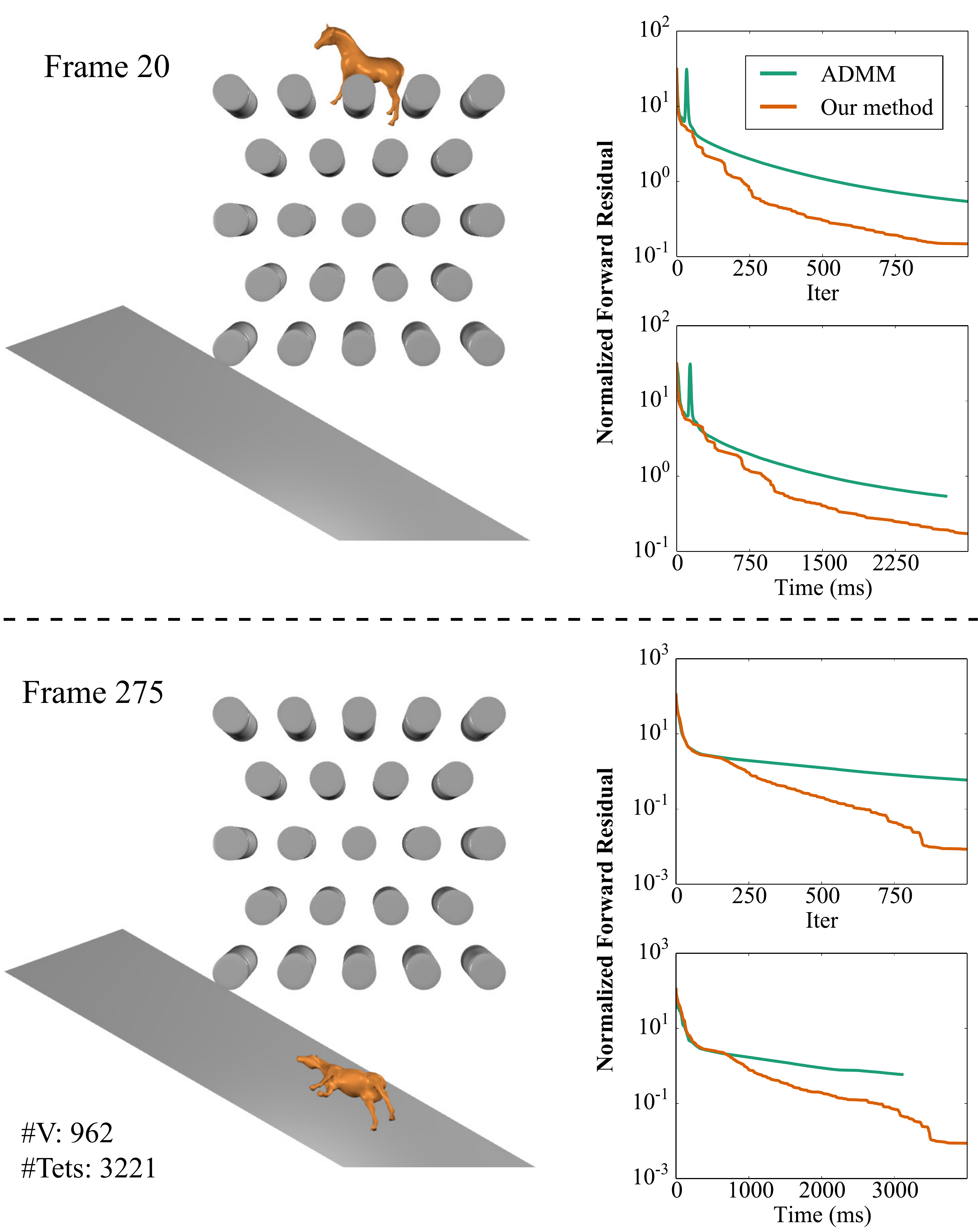}
	\caption{The same simulation of a falling horse as in Fig.~\ref{fig:HorseHole}, with more complex arrangement of static objects. Our acceleration approach remains effective.}
	\label{fig:HorseObjects}
\end{figure}

\begin{figure*}[t!] 
	\centering
	\includegraphics[width=0.95\textwidth]{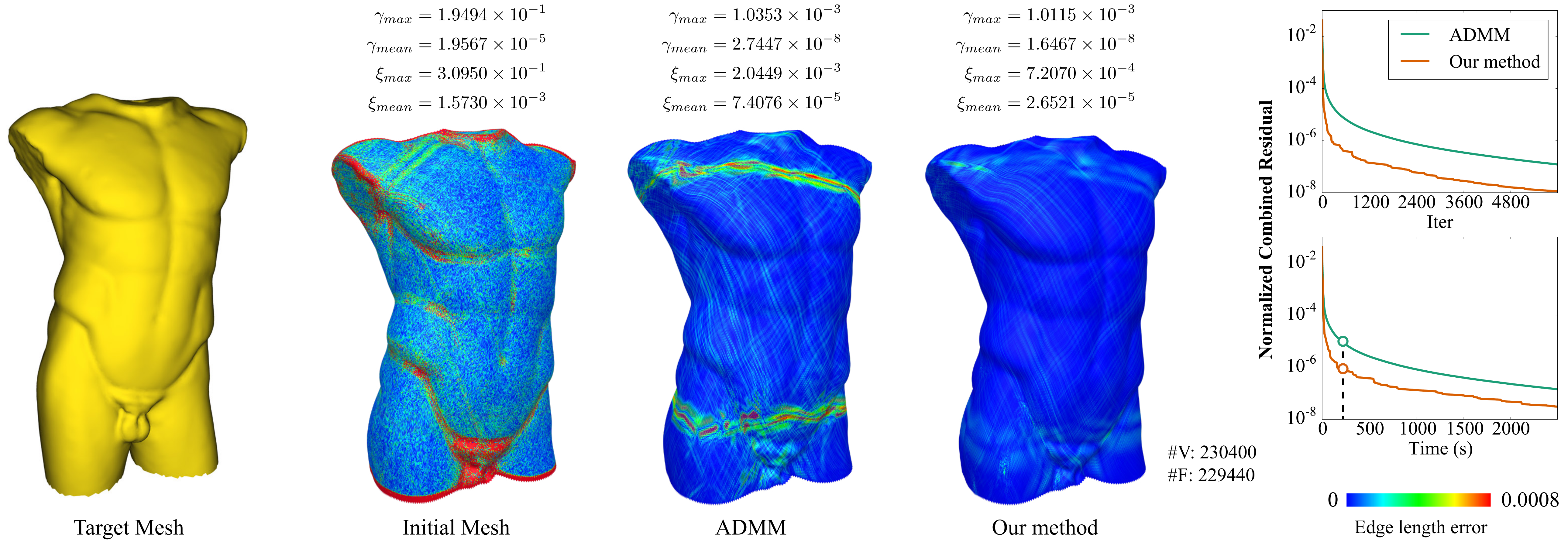}
	\caption{Our method accelerates an ADMM solver for wire mesh optimization, as shown by the normalized combined residual plots. We also show two results computed using ADMM and our accelerated solver within the same computational time (indicated in the bottom-right plot), and evaluate their violation of the angle constraints and edge length constraints using the error metrics in Eq.~\eqref{eq:WireMeshErrorMetrics}. Our result satisfies these constraints better.}
	\label{fig:WireMesh}
\end{figure*}

\begin{figure*}[t!] 
	\centering
	\includegraphics[width=0.95\textwidth]{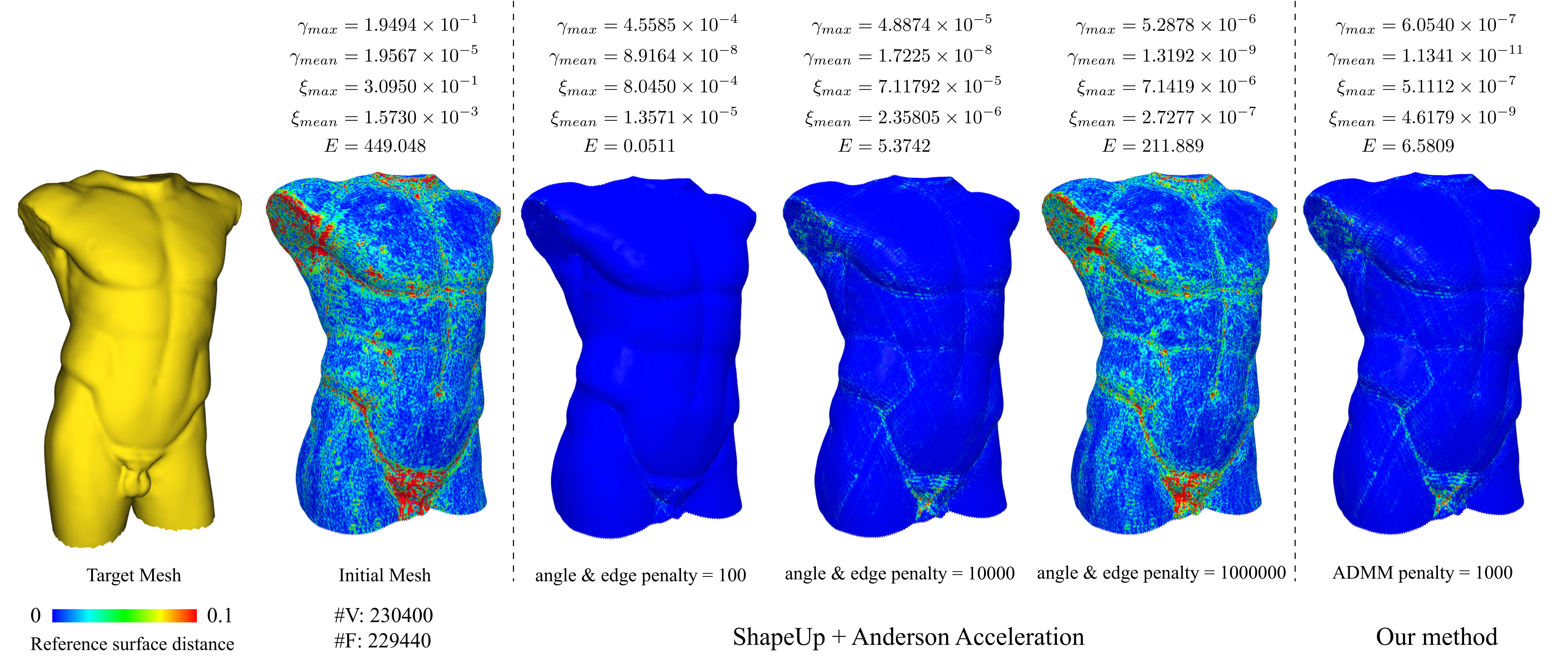}
	\caption{Comparison of wire mesh optimzation results using our accelerated ADMM solver and an accelerated quadratic penalty method as described in~\protect\cite{Peng2018}. The error metric $E$ is the sum of squared distances from the mesh vertices to the reference shape, and the color-coding illustrates the distance for each vertex. 
	Although the quadratic penalty method can improve satisfaction of the angle and edge length constraints with a larger penalty weight, this leads to greater deviation from the reference shape.}
	\label{fig:WireMeshShapeUp}
\end{figure*}

\begin{figure*}[t!] 
	\centering
	\includegraphics[width=0.98\textwidth]{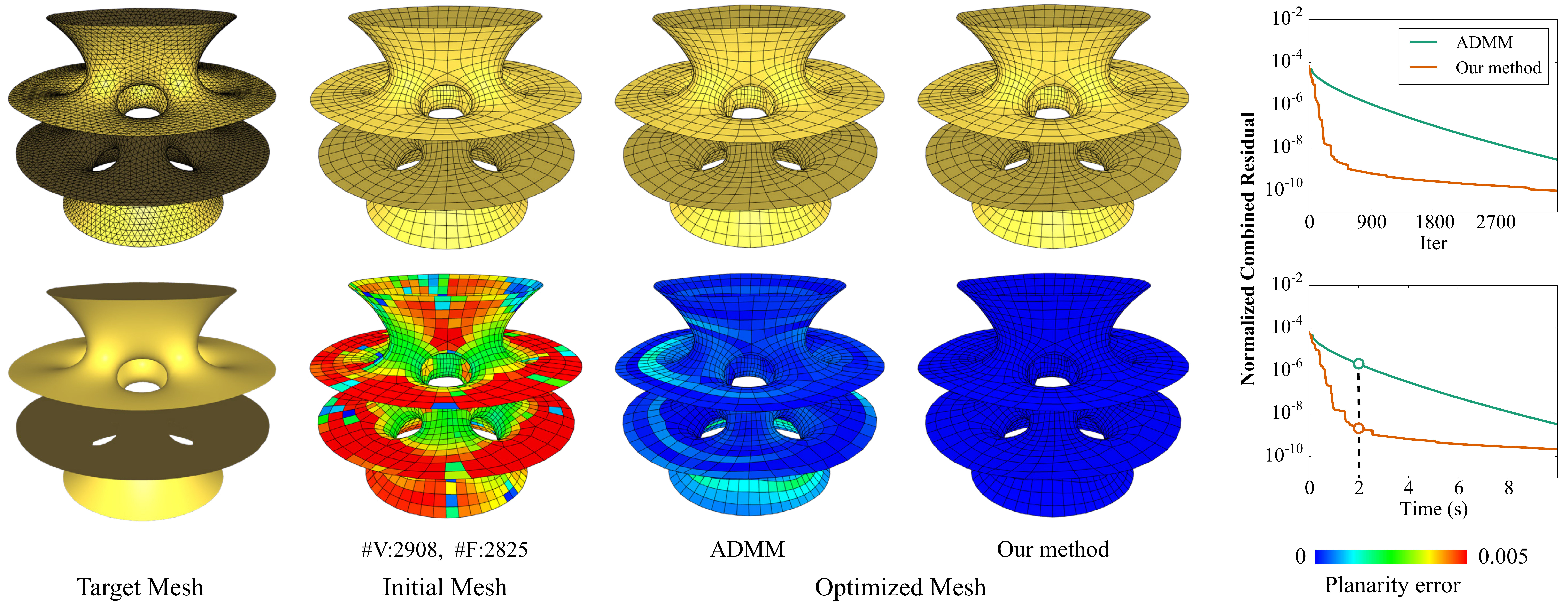}
	\caption{PQ mesh optimization using our accelerated solver convergences faster than ADMM, and achieves better satisfaction of the planarity constraints within the same computational time (highlighted in the plot in bottom right).}
	\label{fig:Costa}
\end{figure*}

\subsection{Geometry processing}
We also apply our method to an ADMM solver for mesh optimization subject to both soft and hard constraints based on~\cite{Deng2015}. The input is a mesh with vertex positions $\mathbf{x}$, soft constraints $\mathbf{A}_i \mathbf{x} \in \mathcal{C}_i $ ($i \in \mathcal{S}$), and hard constraints $\mathbf{A}_j \mathbf{x} \in \mathcal{C}_j $ ($j \in \mathcal{H}$). Here each reduction matrix $\mathbf{A}_i$ and $\mathbf{A}_j$ selects vertex positions relevant to the constraint and (where appropriate) compute their differential coordinates with respect to either their mean position or one of the vertices. \cite{Deng2015} introduce auxiliary variables $\mathbf{z}_i \in \mathcal{C}_i$ ($i \in \mathcal{S}$) and $\mathbf{z}_j \in \mathcal{C}_j$ ($j \in \mathcal{H}$) to derive an optimization problem
\begin{align}
	\min_{\mathbf{x}, \mathbf{z}}~~ &\frac{1}{2}\|\mathbf{L}(\mathbf{x} - \tilde{\mathbf{x}})\|^2 + \sum_{i \in \mathcal{S}} \left(\frac{w_i}{2} \|\mathbf{A}_i \mathbf{x} - \mathbf{z}_i\|^2 +  \sigma_{\mathcal{C}_i}(\mathbf{z}_i) \right) + 
\sum_{j \in \mathcal{H}} \sigma_{\mathcal{C}_j}(\mathbf{z}_j)\nonumber\\
\textrm{s.t.}~~& \mathbf{A}_j {\mathbf{x}} - \mathbf{z}_j = \mathbf{0}, \quad \forall~j \in \mathcal{H}.
\label{eq:GeometryOptimizatoinProblem}
\end{align}
Here $\|\mathbf{L}(\mathbf{x} - \tilde{\mathbf{x}})\|^2$ is an optional Laplacian fairing energy for the vertex positions and/or for their displacement from initial positions, whereas $\|\mathbf{A}_i \mathbf{x} - \mathbf{z}_i\|^2$ penalizes the violation of a soft constraint with a user-specified weight $w_i$. This problem is solved in~\cite{Deng2015} using the augmented Lagrangian method (ALM),
where each iteration performs multiple alternate updates of $\mathbf{z}$ and $\mathbf{x}$ followed by a single update of $\mathbf{u}$, using the same formulas as~\eqref{eq:zxu}. Wu et al.~\shortcite{Wu2011-ALM} pointed out that it is more efficient to perform only one alternate update of primal variables per iteration, in which case ALM reduces to ADMM. Therefore, we solve the problem using ADMM with the \zxu{} iteration, and apply the general approach in Algorithm~\ref{algo:zxuGeneralAA} for acceleration because the target function is not separable.

In Fig.~\ref{fig:WireMesh}, we apply our method with $m=6$ to the wire mesh optimization problem from~\cite{Grinspun14b}. The input is a regular quad mesh subject to the following constraints:
\begin{itemize}[leftmargin=*]
	\item Hard constraints: all edges have the same length $l$; within a face, each angle formed by two incident edges is in the range $[\frac{\pi}{4}, \frac{3 \pi}{4}]$.
	\item Soft constraint: each vertex lies on a given reference surface.
\end{itemize}
The mesh is optimized without the Laplacian fairing term, i.e., $\mathbf{L} = \mathbf{0}$. Our method leads to a faster decrease of the combined residual with respect to both the iteration count and the computational time. We also evaluate the violation of hard constraints using the following error metrics for angle $\alpha$ and edge length $e$:
\begin{equation}
	\xi(e) = \frac{|e - l|}{l}, \quad 
	\gamma(\alpha) = \left\{
	\begin{array}{ll}
	\frac{\pi}{4} - \alpha & \textrm{if}~\alpha < \frac{\pi}{4},\\
	\alpha - \frac{3\pi}{4} & \textrm{if}~\alpha > \frac{3\pi}{4},\\
	0 & \textrm{otherwise}.
	\end{array}
	\right.
	\label{eq:WireMeshErrorMetrics}
\end{equation}
The data and color-coding in Fig.~\ref{fig:WireMesh} show that within the same computational time, the result from our method satisfies the hard constraints better than the original ADMM.

Besides ADMM, another popular approach for enforcing hard constraints is the quadratic penalty method, which replaces the original constrained problem by an unconstrained problem with quadratic terms in the target function to penalize the violation of hard constraints~\cite{nocedal2006numerical}. Fig.~\ref{fig:WireMeshShapeUp} compares the effectiveness of these two approaches in enforcing hard constraints while decreasing the original target function. For the quadratic penalty method, we use ShapeUp~\cite{BouazizDSWP12} with Anderson acceleration as described in~\cite{Peng2018}, and solve three problem instances with different penalty weights for hard constraints and fixed weights for the other terms. Each solver is run to full convergence for comparison. We can see that although a larger penalty weight for hard constraints improves their satisfaction, it also leads to relatively less penalty and greater violation of the soft constraints. In particular, with a large penalty weight to satisfy the hard constraints to a similar level as ADMM, the result from the quadratic penalty method deviates much more from the reference surface than ADMM. It shows that ADMM is more effective in satisfying hard constraints without compromising the minimization of the target function, and our method further improves its efficiency.

In Figs.~\ref{Fig:Airport} and \ref{fig:Costa}, we apply our method to planar quad mesh optimization, a classical problem in architectural geometry~\cite{Liu2006}. The input is a quad mesh subject to the following constraints:
\begin{itemize}[leftmargin=*]
	\item Hard constraint: vertices within each face lie on a common plane.
	\item Soft constraint: each vertex lies on a given reference surface.
\end{itemize}
Following~\cite{BouazizDSWP12}, the reduction matrix for each hard constraint represents the mean-centering operator for the vertices on a common face. The target function includes a Laplacian fairness energy and a relative fairness energy for the vertex positions, as described in~\cite{Liu2011}. We measure the planarity error for each face $F$ of a given mesh using the metric ${d_{\max}(F)}/{\overline{e}}$, where $d_{\max}(F)$ is the maximum distance from a vertex of $F$ to the best fitting plane of its vertices, and $\overline{e}$ is the average edge length of the mesh.  
In both Fig.~\ref{Fig:Airport} and Fig.~\ref{fig:Costa}, our method accelerates the decrease of the combined residual, producing a result with lower planarity error than the original ADMM within the same computational time.

\subsection{Image processing}
In Fig.~\ref{fig:GaussianNoise},
we apply our method to the ADMM solver from the \mbox{ProxImaL} image optimization framework~\cite{Heide2016}. 
We compare our method with the original solver on the following problem that computes a deconvoluted image $\mathbf{x}$ from an observation image $\mathbf{f}$ with Gaussian noise and a convolution operator $\mathbf{K}$:
\begin{equation}
\min_{\mathbf{x}, \mathbf{z}}~\lambda_1 \|\mathbf{z}_1 - \mathbf{f}\|^2 + \lambda_2 \|\mathbf{z}_2^{i,j}\|
\quad
\textrm{s.t.}~\mathbf{K}{\mathbf{x}} = \mathbf{z}_1,
{~~}
(\nabla \mathbf{x})_{i,j} = \mathbf{z}_2^{i,j}~\forall i, j,
\label{eq:GaussianDenoising}
\end{equation}
where $(\nabla \mathbf{x})_{i,j}$ is the image gradient of $\mathbf{x}$ at pixel $(i,j)$.
This is solved in~\cite{Heide2016} using ADMM with the \xzu{} iteration, and we accelerate it using Algorithm~\ref{algo:xzuGeneralAA} with $m = 6$.
We modify the source code of the ProxImaL library~\footnote{\url{https://github.com/comp-imaging/ProxImaL}} to implement our accelerated solver, and use conjugate gradient to solve the linear systems in the update steps. 
Fig.~\ref{fig:GaussianNoise} shows that our method requires less computational time and lower iteration count to achieve the same residual value.

Finally, in Fig.~\ref{fig:AAQP}, we accelerate the ADMM solver used by the Coded Wavefront Sensor in~\cite{Wang2018-Megapixel} for computing the observed wavefront from a captured image. The wavefront $\mathbf{x}$ is computed by solving an optimization problem
\begin{equation}
	\min_{\mathbf{x},\mathbf{z}}~~\lambda \|\nabla \mathbf{x}\|^2 + g(\mathbf{z})
	\quad
	\textrm{s.t.}~~\nabla \mathbf{x} = \mathbf{z}, 
	\label{eq:WavefrontSolver}
\end{equation}
where $\mathbf{z}$ is an auxiliary variable for image gradient, and $g(\mathbf{z})$ is a quadratic term that measures the consistency between the wavefront and the captured image. From the general condition presented in Appendix~\ref{appendix:ufromz}, we know that in each iteration the dual variable $\mathbf{u}^{k}$ can be represented as a function of $\mathbf{z}^k$ via Eq.~\eqref{neweq:27}. Therefore, we apply Anderson acceleration to $\mathbf{z}$ alone. Moreover, as $g(\mathbf{z})$ is quadratic, Eq.~\eqref{neweq:27} implies that $\mathbf{z}^k$ and $\mathbf{u}^{k}$ are related by a linear map. Thus we use the history $\mathbf{z}$ to compute the affine combination coefficients for Anderson acceleration, and apply them to both $\mathbf{z}$ and $\mathbf{u}$ to derive the accelerated $\mathbf{z}$ and its compatible $\mathbf{u}$, similar to Algorithm~\ref{algo:zxuSeparableAA}. We modify the source code released by the authors of~\cite{Wang2018-Megapixel}~\footnote{\url{https://github.com/vccimaging/MegapixelAO}} to implement our accelerated solver. Fig.~\ref{fig:AAQP} compares the normalized combined residual plots between the two solvers, using a test example provided in the released source code. Compared to the original ADMM, our method leads to a significant reduction of computational time and iteration count for the same accuracy. 
Also included in the comparison is the GMRES acceleration for ADMM proposed in~\cite{Zhang2018-GMRES}, which is designed specifically for strongly convex quadratic problems. Following~\cite{Zhang2018-GMRES}, we restart GMRES every 10 iterations to reduce computational cost. As a general method, our approach is outperformed by GMRES acceleration, but only by a small margin. 

\begin{figure}[t!] 
	\centering
	\includegraphics[width=0.9\columnwidth]{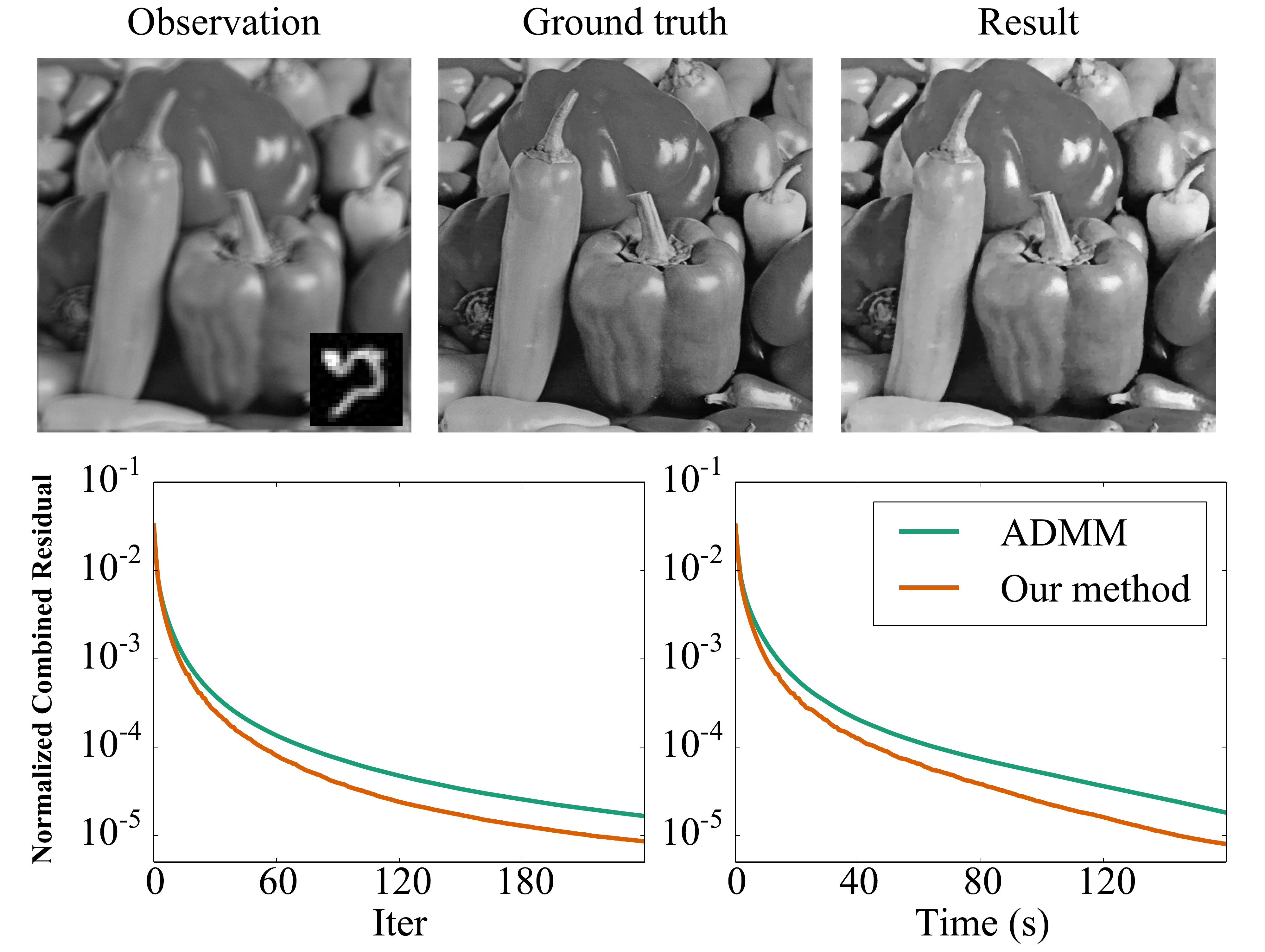}
	\caption{Our method accelerates the ADMM solver in~\protect\cite{Heide2016} for the deconvolution of a $512\times 512$ image with Gaussian noise using Eq.~\eqref{eq:GaussianDenoising}. The given convolution operator $\mathbf{K}$ is visualized in the bottom right of the observation image.}
	\label{fig:GaussianNoise}
\end{figure}

\section{Conclusion and Future Work}

In this paper, we apply Anderson acceleration to improve the convergence of ADMM on computer graphics problems. We show that ADMM can be interpreted as a fixed-point iteration of the second primal variable and the dual variable in the general case, and of only one of them when the problem has a separable target function and satisfies certain conditions.  Such interpretation allows us to directly apply Anderson acceleration in the former case, and further reduce its computational overhead in the latter case. Moreover, for each case we propose a simple residual for measuring the convergence, and use it to determine whether to accept an accelerated iterate. We apply this method to a variety of ADMM solvers in graphics, with applications ranging from physics simulation, geometry processing, to image processing. Our method shows its effectiveness on all these problems, with a notable reduction of iteration account and computational time required to reach the same accuracy. On the theoretical front, we also prove the convergence of ADMM for a common non-convex problem structure in computer graphics under weak assumptions. 
Our work addresses two main limitations of ADMM especially on non-convex problems, which will help to expand its applicability in computer graphics as a versatile solver for optimization problems that are potentially non-smooth, non-convex, and with hard constraints.

\begin{figure}[t!] 
	\centering
	\includegraphics[width=\columnwidth]{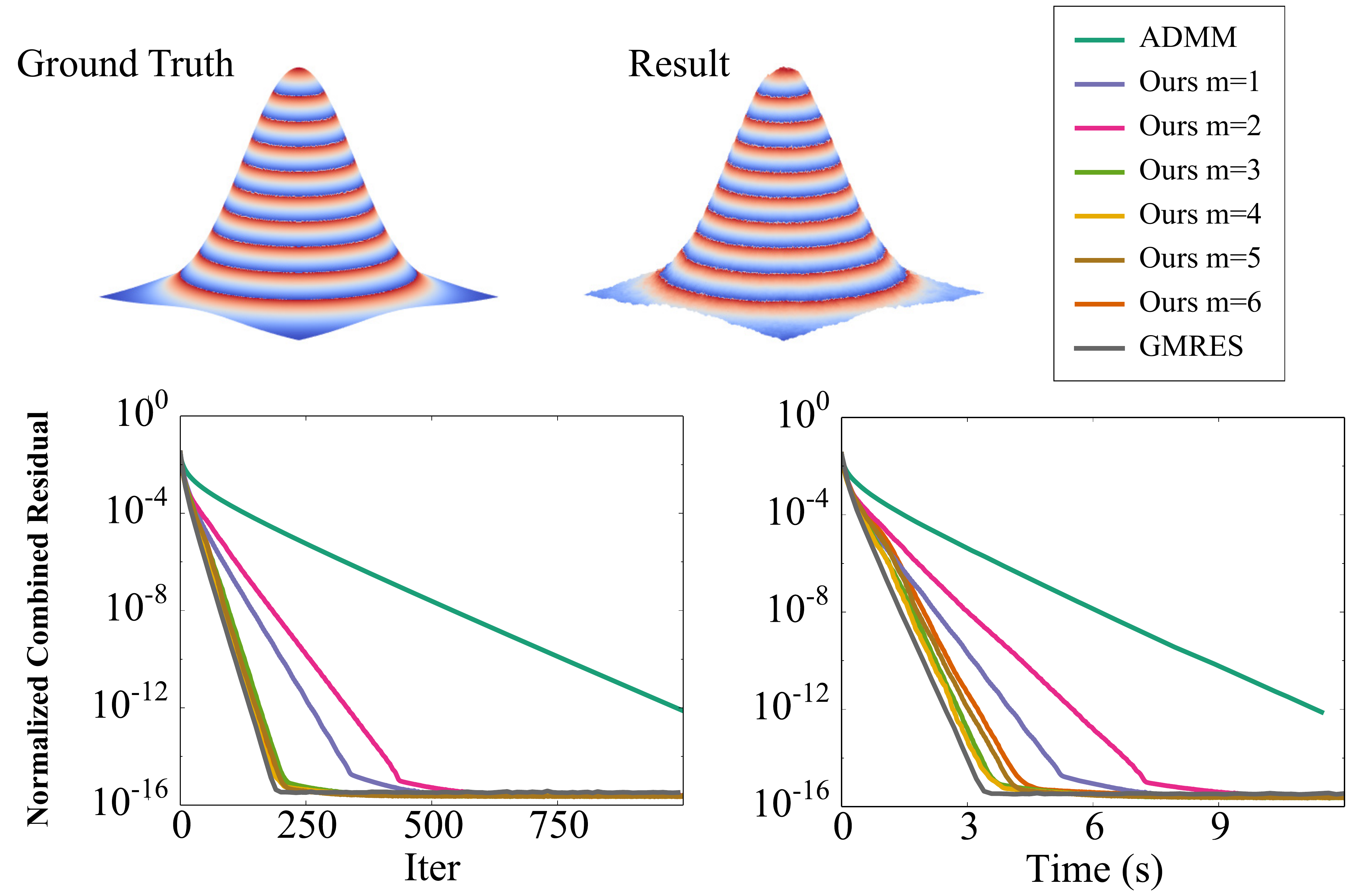}
	\caption{Our method accelerates the ADMM solver in~\protect\cite{Wang2018-Megapixel} for computing the observed wavefront from a captured image, and achieves similar performance as the specialized GMRES acceleration~\protect\cite{Zhang2018-GMRES} despite being a general acceleration technique.}
	\label{fig:AAQP}
\end{figure}

One limitation of our method is that it can be less effective for ADMM solvers with very low computational cost per iteration. In this case, the overhead of Anderson acceleration can cause a large relative increase of computational time, which partly cancels out the speedup gained from the reduction of iteration count. One such example is Fig.~\ref{fig:HeatMethod}, where we apply our method to the ADMM solver in~\cite{Tao2019} for correcting a vector field into an integrable gradient field of geodesic distance. Although our method reduces the number of iterations, its large relative overhead actually increases the computational time for achieving the same residual.

Our experiments show that Anderson acceleration is effective in reducing the number of iterations, but we do not have a theoretical guarantee for such property. This is still an open research problem, and the only existing result we are aware of is \cite{Evans2018}, which proves that Anderson acceleration improves the convergence rate for linearly converging fixed-point methods if a set of strong assumptions is satisfied. Further theoretical analysis of our method is needed to understand and guarantee its performance.

Currently we follow the convention and set the mixing parameter $\beta = 1$ for Anderson acceleration. Although it is effective in our experiments, other values of $\beta = 1$ can potentially improve the performance~\cite{Eyert1996}. The optimal choice of mixing parameter remains an open research problem, and should be explored further. 

The convergence of ADMM can also be affected by the choice of the penalty parameter and the conditioning of linear side constraints. Recently, researchers have started to analyze the optimal choice of penalty parameter and conditioning for ADMM, but only on simple convex problems~\cite{Ghadimi2015,Giselsson2017}. Overby et al.~\shortcite{Overby2017} proposed a heuristic for choosing such parameters for non-convex physical simulation problems, but there is still no theoretical guarantee for its effectiveness. A potential future research is to perform such analysis on non-convex problems, as well as how they can be used in conjunction with Anderson acceleration to further improve convergence of ADMM.

\begin{figure}[t!] 
	\centering
	\includegraphics[width=\columnwidth]{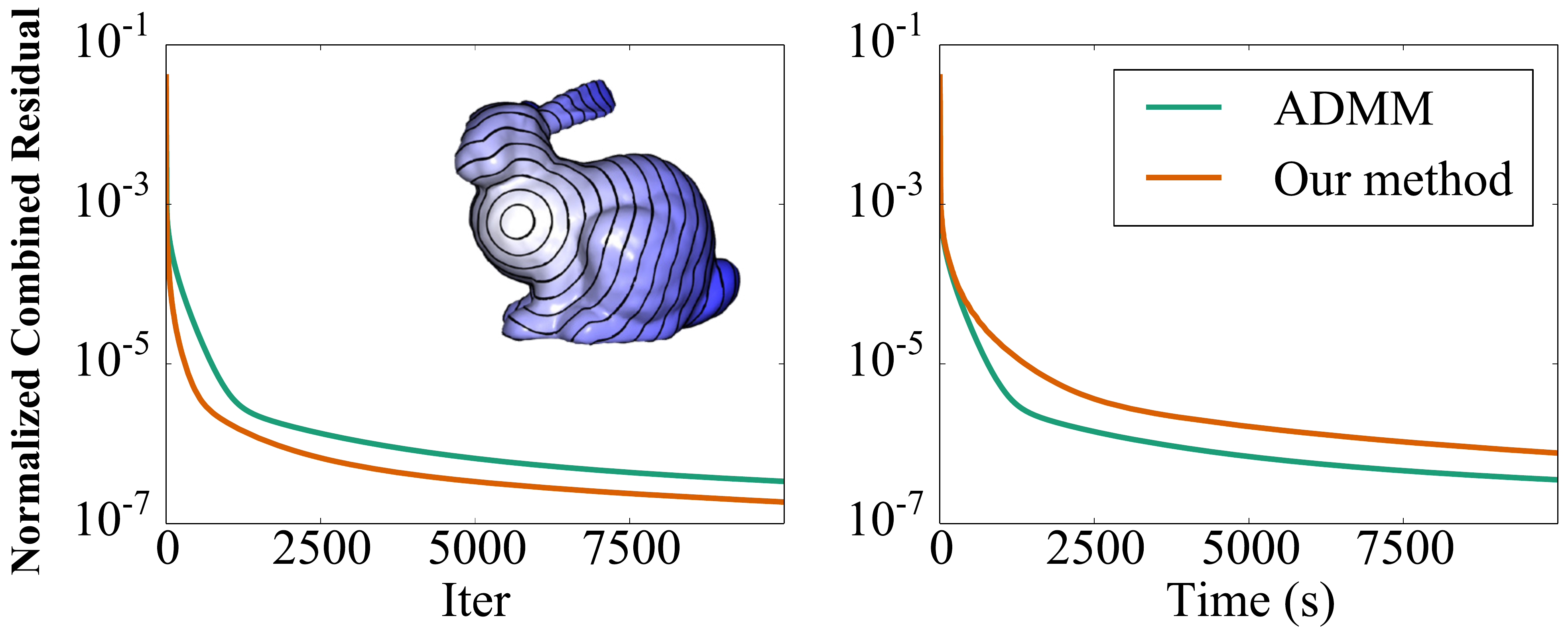}
	\caption{We apply our method to the ADMM solver in~\protect\cite{Tao2019} for correcting a vector field into an integrable gradient field. Due to the very low computational cost per iteration in the original solver, Anderson acceleration incurs a large relative overhead. As a result, although our method reduces the number of iterations, it actually increases the computational time.}
	\label{fig:HeatMethod}
\end{figure}

Finally, as ADMM is a popular solver across different problem domains, we can apply our method to problems outside computer graphics. In this paper we have focused on a problem structure common for graphics tasks. Applications in other domains may involve other problem structures and require different analyses and strategies, which will be an interesting future work.

\begin{acks}
	The target model in Figure~\ref{fig:WireMesh}, \href{https://www.thingiverse.com/thing:146386}{``Male Torso, Diadumenus Type''} by \href{https://www.thingiverse.com/CosmoWenman/about}{Cosmo Wenman}, is licensed under \href{https://creativecommons.org/licenses/by/3.0/}{CC BY 3.0}. This work was supported by National Natural Science Foundation of China (No. 61672481), and Youth Innovation Promotion Association CAS (No. 2018495).
\end{acks}

\bibliographystyle{ACM-Reference-Format}
\bibliography{AA-ADMM} 

\appendix

\section{Proof for Proposition~\ref{prop:ufromz}}
\label{sec:proof:ufromz}
By the optimality condition of \eqref{eq:Sepxzu_z} we have:
\begin{equation}
\label{eq:15}
\nabla g(\mathbf{z}^{k+1})-\mu\mathbf{B}^T(\mathbf{A}\mathbf{x}^{k+1}-\mathbf{B}\mathbf{z}^{k+1}+\mathbf{u}^k-\mathbf{c})=0.
\end{equation}
Put \eqref{eq:Sepxzu_u} into \eqref{eq:15}:
\begin{equation}
\nabla g(\mathbf{z}^{k+1})=\mu\mathbf{B}^T\mathbf{u}^{k+1},
\label{eq:16}
\end{equation}
which completes the proof. \qed

\section{Proof for Proposition~\ref{Prop:xzuResidual}}
\label{sec:proof:xzuResidual}
For the first part, suppose $\mathbf{z}^{k+1}$ is the fixed-point of the \xzu{} iteration, which means that
\begin{equation}
\mathbf{z}^{k+2}=\mathbf{z}^{k+1}.
\label{neweq:1}
\end{equation}
Then we have
\begin{align*}
&~~ \mathbf{u}^{k+2}=\mathbf{u}^{k+1} && \text{by \eqref{eq:16}}\\
\implies &~~ \mathbf{A}\mathbf{x}^{k+2}-\mathbf{B}\mathbf{z}^{k+2}-\mathbf{c}=0 && \text{by \eqref{eq:Sepxzu_u}}\\
\implies &~~ \mathbf{A}\mathbf{x}^{k+2}-\mathbf{B}\mathbf{z}^{k+1}-\mathbf{c}=0 && \text{by \eqref{neweq:1}.}
\end{align*}
For the second part, if $\mathbf{A}\mathbf{x}^{k+2}-\mathbf{B}\mathbf{z}^{k+1}-\mathbf{c}=0$ then from \eqref{eq:Sepxzu_u}:
\begin{equation}
\mathbf{A}\mathbf{x}^{k+2}-\mathbf{c}+\mathbf{u}^{k+1}=\mathbf{A}\mathbf{x}^{k+1}-\mathbf{c}+\mathbf{u}^{k}.
\end{equation}
And from \eqref{eq:Sepxzu_z} and Remark~\ref{remark:UniqueMinimum}
\begin{align*}
	\mathbf{z}^{k+1} & =\argmin_{\mathbf{z}}\left(g(\mathbf{z})+\frac{\mu}{2}\| (\mathbf{A}\mathbf{x}^{k+1}-\mathbf{c}+\mathbf{u}^k)-\mathbf{B}\mathbf{z}\|^2\right)\\
	& = \argmin_{\mathbf{z}}\left(g(\mathbf{z})+\frac{\mu}{2}\| (\mathbf{A}\mathbf{x}^{k+2}-\mathbf{c}+\mathbf{u}^{k+1})-\mathbf{B}\mathbf{z}\|^2\right) = \mathbf{z}^{k+2},
\end{align*}
which completes the proof. \qed

\section{Proof for Proposition~\ref{prop:xfromu}}
\label{sec:proof:xfromu}
By \eqref{eq:Sepzxu_u} we have:
\begin{equation}\label{eq:23}
\mathbf{B}\mathbf{z}^{k+1}-\mathbf{u}^k+\mathbf{c}=\mathbf{A}\mathbf{x}^{k+1}-\mathbf{u}^{k+1}
\end{equation}
Put \eqref{eq:23} into \eqref{eq:Sepzxu_x}:
\begin{align}
&~~ (\mathbf{G}+\mu \mathbf{A}^{T}\mathbf{A})x^{k+1}=\mathbf{G}\tilde{\mathbf{x}}+\mu\mathbf{A}^T(\mathbf{A}\mathbf{x}^{k+1}-\mathbf{u}^{k+1})\nonumber\\
\implies &~~ \mathbf{G}x^{k+1} = \mathbf{G}\tilde{\mathbf{x}} - \mu\mathbf{A}^T \mathbf{u}^{k+1}\nonumber\\
\implies &~~ \mathbf{x}^{k+1}=\tilde{\mathbf{x}}-\mu\mathbf{G}^{-1}\mathbf{A}^T\mathbf{u}^{k+1},
\label{eq:25}
\end{align}
which completes the proof. \qed

\section{Proof for Proposition~\ref{Prop:zxuResidual}}
\label{sec:proof:zxuResidual}
For the first part, suppose $\mathbf{u}^{k+1}$ is the fixed-point of the \zxu{} iteration, so that
\begin{equation}
\label{neweq:6}
\mathbf{u}^{k+2}=\mathbf{u}^{k+1}.
\end{equation}
Then by \eqref{eq:25} and \eqref{neweq:6}:
\begin{equation}
	\mathbf{x}^{k+2} = \mathbf{x}^{k+1}.
	\label{eq:FixedX}
\end{equation}
Therefore
\begin{align*}
	&~~ \mathbf{A}\mathbf{x}^{k+2}-\mathbf{B}\mathbf{z}^{k+2}-\mathbf{c}= \mathbf{0} && \text{by \eqref{eq:Sepzxu_u} and \eqref{neweq:6}}\\
	\implies &~~ \mathbf{A}\mathbf{x}^{k+1}-\mathbf{B}\mathbf{z}^{k+2}-\mathbf{c}=0 && \text{by \eqref{eq:FixedX}.}
\end{align*}
For the second part, suppose
\begin{equation}
\mathbf{A}\mathbf{x}^{k+1}-\mathbf{B}\mathbf{z}^{k+2}-\mathbf{c} = \mathbf{0}.
\label{eq:zxuResidualCondition}
\end{equation}
Then we have
\begin{align*}
&~~ \mathbf{u}^{k+1}-\mathbf{B}\mathbf{z}^{k+2}=\mathbf{u}^{k}-\mathbf{B}\mathbf{z}^{k+1} && \text{by \eqref{eq:Sepzxu_u} and \eqref{eq:zxuResidualCondition}}\\
\implies &~~ \mathbf{x}^{k+2}=\mathbf{x}^{k+1} && \text{by \eqref{eq:Sepzxu_x}}\\
\implies &~~ \mathbf{A}\mathbf{x}^{k+2}-\mathbf{B}\mathbf{z}^{k+2}-\mathbf{c} = \mathbf{0} && \text{by \eqref{eq:zxuResidualCondition}}\\
\implies &~~ \mathbf{u}^{k+2}=\mathbf{u}^{k+1} && \text{by \eqref{eq:Sepzxu_u}},
\end{align*}
which completes the proof. \qed

\section{Further Discussion for Propositions~\ref{prop:ufromz}-\ref{Prop:zxuResidual}}
\label{sec:FurtherDiscussion}
We now consider the general condition such that between the second updated primal variable and the dual variable, one of them is a function of the other. We consider the most general case:
\begin{equation}
\min_{\mathbf{x},\mathbf{z}} ~~ f(\mathbf{x}) + g(\mathbf{z}) \quad
\textrm{s.t.} ~~ \mathbf{A} \mathbf{x} - \mathbf{B} \mathbf{z} = \mathbf{c}.
\label{neweq:20}
\end{equation}
Unlike Section~\ref{sec:SeparableADMM}, we do not assume any specific form of $f$ and $g$. We then only need to discuss the following \xzu{} iteration because the conclusion for \zxu{} iteration is similar:
\begin{align}
\mathbf{x}^{k+1}&\in\argmin_{\mathbf{x}}~L(\mathbf{x}, \mathbf{z}^{k},\mathbf{u}^k), \label{neweq:21}\\
\mathbf{z}^{k+1}&\in\argmin_{\mathbf{z}}~L(\mathbf{x}^{k+1}, \mathbf{z},\mathbf{u}^k),\label{neweq:22}\\
\mathbf{u}^{k+1}&=\mathbf{u}^k+\mathbf{A}\mathbf{x}^{k+1}-\mathbf{B}\mathbf{z}^{k+1}-\mathbf{c}.\label{neweq:23}
\end{align}
We first need the subproblem \eqref{neweq:21} and \eqref{neweq:22} to be well-defined, for which the next condition is sufficient
:\begin{itemize}
\item[(C1)] $f$ and $g$ are bounded from below and lower-semi continuous.
\end{itemize}
Then we rewrite the ADMM iteration as:\begin{align}
&-\mathbf{A}^T(\mathbf{A}\mathbf{x}^{k+1}-\mathbf{B}\mathbf{z}^k-\mathbf{c}+\mathbf{u}^k)\in\partial f(\mathbf{x}^{k+1}),\label{neweq:24}\\
&\mathbf{B}^T(\mathbf{A}\mathbf{x}^{k+1}-\mathbf{B}\mathbf{z}^{k+1}-\mathbf{c}+\mathbf{u}^{k})\in\partial g(\mathbf{z}^{k+1}),\label{neweq:25}\\
&\mathbf{u}^{k+1}=\mathbf{u}^k+\mathbf{A}\mathbf{x}^{k+1}-\mathbf{B}\mathbf{z}^{k+1}-\mathbf{c}.\label{neweq:26}
\end{align}

\subsection{$\mathbf{u}$ as a function of $\mathbf{z}$}
\label{appendix:ufromz}
By \eqref{neweq:25} and \eqref{neweq:26}:
\begin{equation}\label{neweq:27}
\mathbf{B}^T\mathbf{u}^{k+1}\in\partial g(\mathbf{z}^{k+1}).
\end{equation}
Now we can see that $\mathbf{u}^{k+1}$ is a function $\mathbf{z}^{k+1}$ if and only if:
\begin{itemize}
  \item[(C2)] $\mathbf{B}$ is invertible.
  \item[(C3)] $\partial g(\mathbf{z})$ contains exactly one element $\forall\mathbf{z}\in \dom{\partial g}$.
\end{itemize}
From ~\cite[Theorem 9.18]{rockafellar2009variational}  we know that the next condition is sufficient:
\begin{itemize}
  \item[(C3$'$)]$g(\mathbf{z})$ is strictly differentiable $\forall\mathbf{z}\in \dom{\partial g}$.
\end{itemize}
Moreover, we need additional conditions in order to use Anderson acceleration on $\mathbf{z}$. Note that Anderson acceleration generates $\mathbf{z}^{AA}$ by affine combination. So if we want to use \eqref{neweq:27} to compute $\mathbf{u}_{\textrm{AA}}$ from $\mathbf{z}_{\textrm{AA}}$, the following condition is needed:
\begin{itemize}
\item[(C4)] The domain of $\partial g$, defined as $\{\mathbf{z} \mid  \partial g(\mathbf{z})\neq \emptyset \}$, is affine.
\end{itemize}

\subsection{$\mathbf{z}$ as a function of $\mathbf{u}$}
From \eqref{neweq:27} we know that $\mathbf{z}$ is a function of $\mathbf{u}$ if and only if:
\begin{itemize}
\item[(C5)] The inverse mapping of set-valued mapping $\partial g(\mathbf{z})$ is a single-valued mapping.
\end{itemize}
The next condition is sufficient to ensure (C5) but not necessary:
\begin{itemize}
\item[(C5$'$)] $g(\mathbf{z})$ is strictly convex.
\end{itemize}
Also, similar to the argument in Appendix~\ref{appendix:ufromz}, in order to apply Anderson acceleration on $\mathbf{u}$ we need the following condition:
\begin{itemize}
\item[(C6)] The range of $\partial g$, defined as $\bigcup_{\mathbf{z}\in R^n}\partial g(\mathbf{z})$, is affine.
\end{itemize}

\section{Proofs for Convergence Theorems}
\label{sec:ConvergenceTheoremProof}

This section proves the linear convergence theorems when $g$ is locally Lipschitz differentiable (Theorems~\ref{thm:xzuConvergenceLocalLipschitz} and \ref{thm:zxuConvergenceLocalLipschitz}) and the general convergence theorems (Theorems~\ref{thm:GeneralConvergenceXZU} and \ref{thm:GeneralConvergenceZXU}).
The proofs for Theorems~\ref{thm:xzuConvergenceGlobalLipschitz} and \ref{thm:zxuConvergenceGlobalLipschitz} are similar to those for Theorems~\ref{thm:xzuConvergenceLocalLipschitz} and \ref{thm:zxuConvergenceLocalLipschitz}, so we will not give their complete proofs but only summarize the main steps. Because of the order in which some lemmas are used in the proofs, we will prove Theorem~\ref{thm:GeneralConvergenceXZU} and \ref{thm:xzuConvergenceLocalLipschitz} first.
Without loss of generality, we assume $\mathbf{c} = \mathbf{0}$ in Eq.~\eqref{eq:SeparableADMMProblem} to simplify notation.

\subsection{Proof for Theorem~\ref{thm:GeneralConvergenceXZU}}
\label{Proof_Thm3.5}
Recall that Theorem~\ref{thm:GeneralConvergenceXZU} is about general convergence of the \xzu{} iteration. We first note that:
\begin{align}
&~~\nabla\hatg(\mathbf{z})=\mathbf{B}^{-T}\nabla g(\mathbf{B}^{-1}\mathbf{z})\label{eqE1:1}\\
\implies&~~\nabla\hatg(\mathbf{B}\mathbf{z})=\mathbf{B}^{-T}\nabla g(\mathbf{z}).\label{eqE1:2}
\end{align}
These two equations will be used frequently in the following. 
Note that from Assumption~\ref{assump:g}(2) we can derive \eqref{eq:15} from \eqref{eq:Sepxzu_z}. Moreover, based on the definition of $L_c$ in Assumption~\ref{assump:initialValue} we have:
\begin{prop}
Suppose the Lipschitz constant of $\nabla\hat{g}(\mathbf{z})$ over $\conv{\lev{\hatg}{\alpha}}$ is $L_1$, then
$\forall~\mathbf{B}\mathbf{z}_1, \mathbf{B}\mathbf{z}_2 \in \lev{\hatg}{\alpha}$, we have
\begin{equation}\label{eq:66}
|\hatg(\mathbf{B}\mathbf{z}_1)-\hatg(\mathbf{B}\mathbf{z}_2)-\langle\nabla \hatg(\mathbf{B}\mathbf{z}_2),\mathbf{B}\mathbf{z}_1-\mathbf{B}\mathbf{z}_2\rangle|\leq\frac{L_1}{2}\| \mathbf{B}\mathbf{z}_1-\mathbf{B}\mathbf{z}_2\|^2.
\end{equation}
Moreover, if $\mu>L_1$, and $\mathbf{z}_2\in\argmin_{\mathbf{z}}(g(\mathbf{z})+\frac{\mu}{2}\| \mathbf{B}\mathbf{z}-\mathbf{q}\|^2)$,
then:
\begin{equation*}
g(\mathbf{z}_2)+\frac{\mu}{2}\| \mathbf{B}\mathbf{z}_2-\mathbf{q}\|^2 \leq g(\mathbf{z}_1)+\frac{\mu}{2}\| \mathbf{B}\mathbf{z}_1-\mathbf{q}\|^2\notag -\frac{\mu-L_1}{2}\| \mathbf{B}\mathbf{z}_1-\mathbf{B}\mathbf{z}_2\|^2.
\end{equation*}
\label{prop:GeneralProof1}
\end{prop}
The proof is standard so we omit it. Also see~\cite[Lemma 1.2.3 \& Theorem 2.1.8]{nesterov2013introductory}. The next lemma is important.
\begin{lemma}
If Assumption~\ref{assump:g} and \ref{assump:initialValue} hold, $\frac{\mu}{2}-
\frac{L_c^2}{\mu}>\frac{L_c}{2}$, and the \xzu{} iteration satisfies $g(\mathbf{z}^k)\leq T(\mathbf{x}^0,\mathbf{z}^0)+c_1$ and  $L(\mathbf{x}^k,\mathbf{z}^k,\mathbf{u}^k)\leq L(\mathbf{x}^0,\mathbf{z}^0,\mathbf{u}^0)=T(\mathbf{x}^0,\mathbf{z}^0)$. Then
\begin{equation}
g(\mathbf{z}^{k+1})\leq T(\mathbf{x}^0,\mathbf{z}^0)+c_1,
~~~~L(\mathbf{x}^{k+1},\mathbf{z}^{k+1},\mathbf{u}^{k+1})\leq T(\mathbf{x}^0,\mathbf{z}^0).\label{eq:GeneralConvergenceLemma}
\end{equation}
\label{lemma:GeneralProof}
\end{lemma}
\begin{proof}
By the definition of $\mathbf{z}^{k+1}$ in \eqref{eq:Sepxzu_z}:
\begin{equation*}
g(\mathbf{z}^{k+1})+\frac{\mu}{2}\| \mathbf{A}\mathbf{x}^{k+1}-\mathbf{B}\mathbf{z}^{k+1}+\mathbf{u}^k\|^2 \leq g(\mathbf{z}^{k})+\frac{\mu}{2}\| \mathbf{A}\mathbf{x}^{k+1}-\mathbf{B}\mathbf{z}^{k}+\mathbf{u}^k\|^2.
\end{equation*}
And notice the definition of $\mathbf{x}^{k+1}$ in \eqref{eq:Sepxzu_x}:
\begin{equation}
f(\mathbf{x}^{k+1})+\frac{\mu}{2}\| \mathbf{A}\mathbf{x}^{k+1}-\mathbf{B}\mathbf{z}^k+\mathbf{u}^k\|^2 \leq f(\mathbf{x}^k)+\frac{\mu}{2}\| \mathbf{A}\mathbf{x}^{k}-\mathbf{B}\mathbf{z}^k+\mathbf{u}^k\|^2.
\label{eq:72}
\end{equation}
Combine the two equations above:
\begin{equation*}
T(\mathbf{x}^{k+1},\mathbf{z}^{k+1})+\frac{\mu}{2}\| \mathbf{A}\mathbf{x}^{k+1}-\mathbf{B}\mathbf{z}^{k+1}+\mathbf{u}^k\|^2 \leq L(\mathbf{x}^{k},\mathbf{z}^{k},\mathbf{u}^k)+\frac{\mu}{2}\| \mathbf{u}^k\|^2.
\end{equation*}
By \eqref{eq:Sepxzu_u} and \eqref{eq:16}:
\begin{equation}
T(\mathbf{x}^{k+1},\mathbf{z}^{k+1})+\frac{\mu}{2}\| \mathbf{u}^{k+1}\|^2 \leq
L(\mathbf{x}^{k},\mathbf{z}^{k},\mathbf{u}^k)+\frac{1}{2\mu}\| \mathbf{B}^{-T}\nabla g(\mathbf{z}^k)\|^2.
\label{eq:74}
\end{equation}
Notice that $L(\mathbf{x}^k,\mathbf{z}^k,\mathbf{u}^k)\leq T(\mathbf{x}^0,\mathbf{z}^0)$ and by the definition of $c_1$:
\begin{equation*}
g(\mathbf{z}^{k+1})\leq T(\mathbf{x}^0,\mathbf{z}^0)+c_1.
\end{equation*}
Thus we have proved the first part. For the second part, we have:
\begin{align}
&L(\mathbf{x}^{k+1},\mathbf{z}^k,\mathbf{u}^k)\leq L(\mathbf{x}^{k},\mathbf{z}^k,\mathbf{u}^k),\label{eq:76}\\
&L(\mathbf{x}^{k+1},\mathbf{z}^{k+1},\mathbf{u}^k)\leq L(\mathbf{x}^{k+1},\mathbf{z}^k,\mathbf{u}^k)-\frac{\mu-L_c}{2}\| \mathbf{B}\mathbf{z}^{k+1}-\mathbf{B}\mathbf{z}^k\|^2, \label{eq:77}\\
&L(\mathbf{x}^{k+1},\mathbf{z}^{k+1},\mathbf{u}^{k+1})= L(\mathbf{x}^{k+1},\mathbf{z}^{k+1},\mathbf{u}^k)+\mu\| \mathbf{u}^{k+1}-\mathbf{u}^k\|^2.
\label{eq:78}\end{align}
Here \eqref{eq:76} is derived from \eqref{eq:72}, \eqref{eq:77} is derived from Assumption~\ref{assump:initialValue}(2) and Proposition~\ref{prop:GeneralProof1}, and \eqref{eq:78} is trivial. Add them together, and then use \eqref{eq:16} and the fact that $\frac{\mu}{2}-
\frac{L_c^2}{\mu}>\frac{L_c}{2}$:
\begin{align}
L(\mathbf{x}^{k+1},\mathbf{z}^{k+1},\mathbf{u}^{k+1})&\leq L(\mathbf{x}^{k},\mathbf{z}^{k},\mathbf{u}^k)-(\frac{\mu}{2}-
\frac{L_c^2}{\mu}-\frac{L_c}{2})\| \mathbf{B}\mathbf{z}^{k+1}-\mathbf{B}\mathbf{z}^{k}\|^2\notag\\&\leq T(\mathbf{x}^0,\mathbf{z}^0),
\label{eq:79}
\end{align}
Which completes the proof.
\end{proof}
From Assumption~\ref{assump:initialValue}(1) and Lemma~\ref{lemma:GeneralProof}, we have:
\begin{prop}
\label{propG-2}
Suppose Assumptions~\ref{assump:g} and \ref{assump:initialValue} hold, and $\frac{\mu}{2}-
\frac{L_c^2}{\mu}>\frac{L_c}{2}$. Then the \xzu{} iteration satisfies
\begin{align}\label{eq:80}
g(\mathbf{z}^{k})\leq T(\mathbf{x}^0,\mathbf{z}^0)+c_1,~~~~ L(\mathbf{x}^{k},\mathbf{z}^{k},\mathbf{u}^{k})\leq T(\mathbf{x}^0,\mathbf{z}^0).
\end{align}
\label{prop:GeneralProof2}
\end{prop}
By Proposition~\ref{prop:GeneralProof2}, Assumption~\ref{assump:g}(3) has the same effect as the Lipschitz differentiability assumption. The next step is similar to the convergence proof in~\cite{wang2019global}, which requires the following properties for the sequence $(\mathbf{x}^k,\mathbf{z}^k,\mathbf{u}^k)$:
\begin{itemize}
\item[(P1)] \emph{Boundedness}: the generated sequence $(\mathbf{x}^k,\mathbf{z}^k,\mathbf{u}^k)$ is bounded, and $L(\mathbf{x}^k,\mathbf{z}^k,\mathbf{u}^k)$ is lower bounded.
\item[(P2)] \emph{Sufficient descent}: there is a constant $C_1(\mu)>0$ such that for sufficiently large $k$, we have:
\begin{align*}
 &L(\mathbf{x}^k,\mathbf{z}^k,\mathbf{u}^k)-L(\mathbf{x}^{k+1},\mathbf{z}^{k+1},\mathbf{u}^{k+1})\\
 \geq ~~& C_1(\mu)(\| \mathbf{B}(\mathbf{z}^{k+1}-\mathbf{z}^k )\|^2+\| \mathbf{A}(\mathbf{x}^{k+1}-\mathbf{x}^k)\|^2).
\end{align*}
\item[(P3)] \emph{Subgradient bound}: there is a constant $C_2(\mu)>0$ and $\mathbf{d}^{k+1}\in\partial L(\mathbf{x}^{k+1},\mathbf{y}^{k+1},\mathbf{u}^{k+1})$ such that
	\begin{equation*}
    \| \mathbf{d}^{k+1}\|\leq C_2(\mu)(\| \mathbf{B}(\mathbf{z}^{k+1}-\mathbf{z}^k)\|+\| \mathbf{A}(\mathbf{x}^{k+1}-\mathbf{x}^k)\|).
    \end{equation*}
\item[(P4)] \emph{Limiting continuity}: if $(\mathbf{x}^{*},\mathbf{z}^{*},\mathbf{u}^{*})$ is the limit point of the sub-sequence $(\mathbf{x}^{k_s},\mathbf{z}^{k_s},\mathbf{u}^{k_s})$ for $s\in \mathbb{N}$, then we have:
\begin{equation*}
    \lim_{s\rightarrow\infty}L(\mathbf{x}^{k_s},\mathbf{z}^{k_s},\mathbf{u}^{k_s})=L(\mathbf{x}^{*},\mathbf{z}^{*},\mathbf{u}^{*}).
\end{equation*}
\end{itemize}
Note that although the \xzu{} iteration is not same as the one defined in~\cite{li2015global}, the proof for \cite[Theorem 3]{li2015global} is not affected by the difference. Combining it with~\cite[Proposition 2]{wang2019global}, we can prove Theorem~\ref{thm:GeneralConvergenceXZU}:
\begin{proof}[Proof for Theorem~\ref{thm:GeneralConvergenceXZU}]
From~\cite[Proposition 2]{wang2019global}, \cite[Theorem 3]{li2015global}, and Proposition~\ref{prop:GeneralProof2} in our paper, we only need to show (P1)-(P4) hold for $(\mathbf{x}^k,\mathbf{z}^k,\mathbf{u}^k)$.

For (P1), from \eqref{eq:74} we have:\begin{align}
T(\mathbf{x}^k,\mathbf{z}^k)\leq T(\mathbf{x}^0,\mathbf{z}^0)+c_1.
\end{align}
From Assumption~\ref{assump:g}(1) $g(\mathbf{z})$ is level-bounded and $\mathbf{G}$ is invertible so $f(\mathbf{x})$ is also level-bounded, thus $(\mathbf{x}^k,\mathbf{z}^k)$ is bounded. The boundedness of $\mathbf{u}^k$ can be derived from~\eqref{eq:15}. The lower boundedness of $L(\mathbf{x}^{k},\mathbf{z}^{k},\mathbf{u}^{k})$ comes from Assumption~\ref{assump:initialValue}(2) and the fact that $T(\mathbf{x},\mathbf{z})\geq0$. In fact we have:
$L(\mathbf{x}^{k},\mathbf{z}^{k},\mathbf{u}^{k})\geq -c_1$.

In the derivation of \eqref{eq:76}, we did not use the fact that $f(\mathbf{x})$ is quadratic. If we take this into consideration, then \eqref{eq:76} becomes:\begin{align}\label{neq:1}
L(\mathbf{x}^{k+1},\mathbf{z}^k,\mathbf{u}^k)&\leq L(\mathbf{x}^{k},\mathbf{z}^k,\mathbf{u}^k)-l\| \mathbf{x}^{k+1}-x^{k}\|^2.
\end{align}
Here $l>0$ is some constant. \eqref{neq:1},\eqref{eq:77} and \eqref{eq:78} show that (P2) holds.

(P4) is trivial for our problem. For (P3) we have:
\begin{align}
 \nabla_{\mathbf{x}}L(\mathbf{x}^{k+1},\mathbf{z}^{k+1},\mathbf{u}^{k+1})&=\mathbf{G}(\mathbf{x}^{k+1}-\tilde{\mathbf{x}})+\mu \mathbf{A}^T(\mathbf{A}\mathbf{x}^{k+1}-\mathbf{B}\mathbf{z}^{k+1}+\mathbf{u}^{k+1})\notag\\
 &=\mu \mathbf{A}^T(\mathbf{B}\mathbf{z}^{k}-\mathbf{B}\mathbf{z}^{k+1}+\mathbf{u}^{k+1}-\mathbf{u}^{k}),\label{eq:86}\\
\nabla_{\mathbf{z}}L(\mathbf{x}^{k+1},\mathbf{z}^{k+1},\mathbf{u}^{k+1})&=\nabla g(\mathbf{z}^{k+1})+\mu\mathbf{B}^{T}(\mathbf{B}\mathbf{z}^{k+1}-\mathbf{A}\mathbf{x}^{k+1}-\mathbf{u}^{k+1})\notag\\ &=\mu\mathbf{B}^{T}(\mathbf{u}^{k}-\mathbf{u}^{k+1}),\label{eq:87} \\
\nabla_{\mathbf{u}}L(\mathbf{x}^{k+1},\mathbf{z}^{k+1},\mathbf{u}^{k+1}) & =\mu(\mathbf{A}\mathbf{x}^{k+1}-\mathbf{B}\mathbf{z}^{k+1}) =\mu(\mathbf{u}^{k+1}-\mathbf{u}^{k}).
  \label{eq:88}
 \end{align}
Here we use \eqref{eq:Sepxzu_x} and \eqref{eq:Sepxzu_u} for \eqref{eq:86}; \eqref{eq:15} for \eqref{eq:87}; \eqref{eq:Sepxzu_u} for \eqref{eq:88}. By \eqref{eq:15}, Assumption~\ref{assump:g}(3), and Assumption~\ref{assump:initialValue}:
\begin{align*}
&\|  \nabla_{\mathbf{x}}L(\mathbf{x}^{k+1},\mathbf{z}^{k+1},\mathbf{u}^{k+1})\|\leq \sqrt{\rho(\mathbf{A}^T\mathbf{A})} (\mu+L_c) \|\mathbf{B}\mathbf{z}^{k+1}-\mathbf{B}\mathbf{z}^{k}\|, \\
&\|  \nabla_{\mathbf{z}}L(\mathbf{x}^{k+1},\mathbf{z}^{k+1},\mathbf{u}^{k+1})\|\leq\sqrt{\rho(\mathbf{B}^T\mathbf{B})}L_c\| \mathbf{B}\mathbf{z}^{k+1}-\mathbf{B}\mathbf{z}^{k}\|.
\end{align*}
And notice that $ \nabla_{\mathbf{u}}L(\mathbf{x}^{k+1},\mathbf{z}^{k+1},\mathbf{u}^{k+1})=-\mathbf{B}^T\nabla_{\mathbf{z}}L(\mathbf{x}^{k+1},\mathbf{z}^{k+1},\mathbf{u}^{k+1})$, then we get the result.
\end{proof}
\subsection{Proof for Theorem~\ref{thm:xzuConvergenceLocalLipschitz}}
\label{Proof_Thm3.3}
Recall that Theorem~\ref{thm:xzuConvergenceLocalLipschitz} is about linear convergence of the \xzu{} iteration.
To simplify the notation, we define:
\begin{equation}
\mathbf{N}(\mathbf{z}) \coloneqq \mathbf{z}+\frac{1}{\mu}\mathbf{B}^{-T}\nabla g(\mathbf{B}^{-1}\mathbf{z}).
\label{eq:MappingN}
\end{equation}
\begin{prop}
	The \xzu{} iteration~\eqref{eq:Sepxzu_x}-\eqref{eq:Sepxzu_u} satisfies
	\begin{equation}
	\mathbf{N}(\mathbf{B}\mathbf{z}^{k+1})=(\mathbf{I}+\mu\mathbf{K})^{-1}(\mathbf{A}\tilde{\mathbf{x}}+\mu \mathbf{K}\mathbf{B}z^{k}+\frac{1}{\mu}\mathbf{B}^{-T}\nabla g(\mathbf{z}^k)),
	\label{eq:xzuN}
	\end{equation}
	where matrix $\mathbf{K}$ is defined in \eqref{eq:DefinitionHatGAndK}.
	\label{prop:xzuN}
\end{prop}
\begin{proof}
	By \eqref{eq:Sepxzu_u} we have:
	\begin{align}
	&~~\mathbf{B}\mathbf{z}^k-\mathbf{u}^k = \mathbf{A}\mathbf{x}^{k+1}+\mathbf{B}\mathbf{z}^k-\mathbf{B}\mathbf{z}^{k+1}-\mathbf{u}^{k+1}.\nonumber\\
	\overset{\text{by \eqref{eq:Sepxzu_x}}}{\implies}&~~(\mathbf{G}+\mu \mathbf{A}^T\mathbf{A})\mathbf{x}^{k+1} =\mathbf{G}\tilde{\mathbf{x}} +\mu\mathbf{A}^T(\mathbf{A}\mathbf{x}^{k+1}+\mathbf{B}\mathbf{z}^k-\mathbf{B}\mathbf{z}^{k+1}-\mathbf{u}^{k+1})\nonumber\\
	\implies&~~\mathbf{x}^{k+1} =\tilde{\mathbf{x}}+\mu\mathbf{G}^{-1}\mathbf{A}^T(\mathbf{B}\mathbf{z}^k-\mathbf{B}\mathbf{z}^{k+1}-\mathbf{u}^{k+1}) \nonumber\\
	\implies&~~\mathbf{A}\mathbf{x}^{k+1} =\mathbf{A}\tilde{\mathbf{x}}+\mu \mathbf{A}\mathbf{G}^{-1}\mathbf{A}^T(\mathbf{B}\mathbf{z}^k-\mathbf{B}\mathbf{z}^{k+1}-\mathbf{u}^{k+1}).\nonumber\\
	\overset{\text{by~\eqref{eq:Sepxzu_u}}}{\implies}&~~(\mathbf{I}+\mu\mathbf{A}\mathbf{G}^{-1}\mathbf{A}^T)(\mathbf{u}^{k+1}+\mathbf{B}\mathbf{z}^{k+1})=\mathbf{A}\tilde{\mathbf{x}}+\mu \mathbf{A}\mathbf{G}^{-1}\mathbf{A}^T\mathbf{B}\mathbf{z}^k+\mathbf{u}^k.\nonumber\\
	\overset{\text{by~\eqref{eq:16}}}{\implies}&~~(\mathbf{I}+\mu\mathbf{A}\mathbf{G}^{-1}\mathbf{A}^T)(\frac{1}{\mu}\mathbf{B}^{-T}\nabla g (\mathbf{z}^{k+1})+\mathbf{B}\mathbf{z}^{k+1})\nonumber\\
	& \quad = \mathbf{A}\tilde{\mathbf{x}}+\mu \mathbf{A}\mathbf{G}^{-1}\mathbf{A}^T\mathbf{B}\mathbf{z}^k+\frac{1}{\mu}\mathbf{B}^{-T}\nabla g(\mathbf{z}^k).\nonumber
	\end{align}
	From the definitions of $\mathbf{N}$ and $\mathbf{K}$, the last equation above becomes:
	\begin{align*}
	&~~(\mathbf{I}+\mu\mathbf{K})~\mathbf{N}(\mathbf{B}\mathbf{z}^{k+1})=\mathbf{A}\tilde{\mathbf{x}}+\mu\mathbf{K}\mathbf{B}\mathbf{z}^{k}+\frac{1}{\mu}\mathbf{B}^{-T}\nabla g(\mathbf{z}^{k})\\
	\implies&~~\mathbf{N}(\mathbf{B}\mathbf{z}^{k+1})=(\mathbf{I}+\mu\mathbf{K})^{-1}(\mathbf{A}\tilde{\mathbf{x}}+\mu \mathbf{K}\mathbf{B}z^{k}+\frac{1}{\mu}\mathbf{B}^{-T}\nabla g(\mathbf{z}^k)),
	\end{align*}
	which completes the proof.
\end{proof}
Next we show a sufficient condition for the convergence to a stationary point:
\begin{prop}
If the sequence $\{\mathbf{z}^k\}$ converges, then $\{\mathbf{x}^k, \mathbf{z}^k, \mathbf{u}^k\}$ converges to a stationary point defined in~\eqref{eq:StationaryPoint}.
\label{prop:StationaryPointConvergenceXZU}
\end{prop}
\begin{proof}
Suppose $\mathbf{z}^k\rightarrow\mathbf{z}^{*}$. Then by \eqref{eq:ufromz}, $\mathbf{u}^k\rightarrow\mathbf{u}^{*}=\mathbf{B}^{-T}\nabla g(\mathbf{z}^{*})$, which proves $\nabla g(\mathbf{z}^{*})-\mathbf{B}^{T}\mathbf{u}^{*}=0$. By \eqref{eq:Sepxzu_x}, $\mathbf{x}^k\rightarrow\mathbf{x}^{*}$ where
\begin{equation}\label{eq:E1}
\mathbf{x}^{*}=(\mathbf{G}+\mu\mathbf{A}^T\mathbf{A})^{-1}(\mathbf{G}\tilde{\mathbf{x}}+\mu\mathbf{A}^T(\mathbf{B}\mathbf{z}^*+\mathbf{c}-\mathbf{u}^*))
\end{equation}
In \eqref{eq:Sepxzu_u}, let $k\rightarrow\infty$ then we have
\begin{align}\label{eq:E2}
\mathbf{A}\mathbf{x}^{*}-\mathbf{B}\mathbf{z}^{*}=\mathbf{c}
\end{align}
The identity $\nabla f(\mathbf{x}^{*})+\mathbf{A}^T\mathbf{u}^{*}=0$ then follows from \eqref{eq:E1} and \eqref{eq:E2}.
\end{proof}
We now show that $\{\mathbf{z}^k\}$ converge linearly:
\begin{proof}[Proof for Theorem~\ref{thm:xzuConvergenceLocalLipschitz}]
	From~\eqref{eq:xzuN}:
	\begin{equation}
	\label{eq:H1}
	\begin{aligned}
	&\mathbf{N}(\mathbf{B}\mathbf{z}^{k+1})-\mathbf{N}(\mathbf{B}\mathbf{z}^{k})\\
	=~~&(\mathbf{I}+\mu \mathbf{K})^{-1}(\mu\mathbf{K} \mathbf{B}(\mathbf{z}^{k}-\mathbf{z}^{k-1}) +\frac{1}{\mu}\mathbf{B}^{-T}(\nabla g(\mathbf{z}^k)- \nabla g(\mathbf{z}^{k-1}))).
	\end{aligned}
	\end{equation}
	By Proposition~\ref{propG-2},
	$g(\mathbf{z}^{k})\leq T(\mathbf{x}^0,\mathbf{z}^0)+c_1$, $\forall k\in\mathbb{N}$.
	Then by the definition of $c_1$(see Assumption~\ref{assump:initialValue}) and Assumption~\ref{assump:g}(3):
	\begin{align}
	& \|\nabla\hatg(\mathbf{B}\mathbf{z}^{k+1})-\nabla\hatg(\mathbf{B}\mathbf{z}^{k})\|\leq L_c\|\mathbf{B}\mathbf{z}^{k+1}-\mathbf{B}\mathbf{z}^{k}\|, ~~\forall k\in\mathbb{N} \nonumber \\
	\implies &
	\| \mathbf{N}(\mathbf{B}\mathbf{z}^{k+1})-\mathbf{N}(\mathbf{B}\mathbf{z}^{k})\|\geq(1-\frac{L_c}{\mu})\| \mathbf{B}\mathbf{z}^{k+1}-\mathbf{B}\mathbf{z}^k\|. \label{eq:H4}
	\end{align}
	For the right hand side of \eqref{eq:H1}:
	\begin{align*}
	&\|(\mathbf{I}+\mu \mathbf{K})^{-1}(\mu\mathbf{K} \mathbf{B}(\mathbf{z}^{k}-\mathbf{z}^{k-1}) +\frac{1}{\mu}\mathbf{B}^{-T}(\nabla g(\mathbf{z}^k)- \nabla g(\mathbf{z}^{k-1})))\|\\
	\leq~~&\| (\mathbf{I}+\mu \mathbf{K})^{-1}(\mu \mathbf{K}\mathbf{B}(\mathbf{z}^{k}-\mathbf{z}^{k-1})\| +\|\frac{1}{\mu}(\mathbf{I}+\mu\mathbf{K})^{-1}\mathbf{B}^{-T}(\nabla g(\mathbf{z}^{k})-\nabla g(\mathbf{z}^{k-1})) \|.
	\end{align*}
	By the spectral mapping theorem:\begin{align}\label{eq:H6}
	&\| (\mathbf{I}+\mu\mathbf{K})^{-1}(\mu\mathbf{K})\|=\rho\left( (\mathbf{I}+\mu\mathbf{K})^{-1}(\mu\mathbf{K})\right)=\frac{\mu\rho(\mathbf{K})}{1+\mu\rho(\mathbf{K})}.
	\end{align}
	And notice that $\mathbf{K}$ is positive semi-definite:
	\begin{equation}\label{eq:H7}
	\|\frac{1}{\mu}(\mathbf{I}+\mu\mathbf{K})^{-1}\mathbf{B}^{-T}(\nabla g(\mathbf{z}^{k})-\nabla g(\mathbf{z}^{k-1}))) \|\leq \frac{L_c}{\mu}\| \mathbf{B}\mathbf{z}^{k}-\mathbf{B}\mathbf{z}^{k-1} \|.
	\end{equation}
	Combine \eqref{eq:H6} with \eqref{eq:H7}:\begin{align}\label{eq:H8}
	&\|(\mathbf{I}+\mu\mathbf{K})^{-1}(\mu \mathbf{K}(\mathbf{B}\mathbf{z}^{k}-\mathbf{B}\mathbf{z}^{k-1})+(\frac{1}{\mu}(\mathbf{B}^{-T}\nabla g(\mathbf{z}^k)-\mathbf{B}^{-T}\nabla g( \mathbf{z}^{k-1})))\| \notag\\
	&\leq (\frac{\mu\rho(\mathbf{K})}{1+\mu\rho(\mathbf{K})}+\frac{L_c}{\mu})\| \mathbf{B}\mathbf{z}^{k}-\mathbf{B}\mathbf{z}^{k-1}\|.
	\end{align}
	By \eqref{eq:H4} and \eqref{eq:H8} we have:\begin{align}\label{eq:H9}
	\| \mathbf{B}\mathbf{z}^{k+1}-\mathbf{B}\mathbf{z}^{k}\|\leq\frac{\frac{\mu\rho(\mathbf{K})}{1+\mu\rho(\mathbf{K})}+\frac{L_c}{\mu}}{1-\frac{L_c}{\mu}}\| \mathbf{B}\mathbf{z}^k-\mathbf{B}\mathbf{z}^{k-1}\|.
	\end{align}
	If $\mu>\max\left\{\frac{1}{\frac{1}{2L_c}-\rho(\mathbf{K})},\frac{1}{L_c}\right\}$ then $\gamma_1<1$, which completes the proof.
\end{proof}

\subsection{Proof for Theorem~\ref{thm:GeneralConvergenceZXU}}
\label{Proof_Thm3.6}
Theorem~\ref{thm:GeneralConvergenceZXU} is about general convergence of the \xzu{} iteration.
We first prove Proposition~\ref{prop:i-1} that defines the value $\eta$.
\begin{proof}[Proof for Proposition~\ref{prop:i-1}]
	By the definition of $\mathbf{K}$ in \eqref{eq:DefinitionHatGAndK}, we know that $\mathbf{K}(R(\mathbf{A}))\subset R(\mathbf{A})$. Since $R(\mathbf{A})$ is a linear subspace and $\mathbf{K}$ is a linear operator, for the proof it suffices to show
	$\text{ker}(\mathbf{K})\cap R(\mathbf{A})=\{0\}$,
	where $\text{ker}(\mathbf{K})$ is the kernel of $\mathbf{K}$. Now assume $\mathbf{y}\in\text{ker}(\mathbf{K})$, then for any $\mathbf{z}\in \mathbb{R}^{q}$ where $q$ is the number of rows in matrix $\mathbf{A}$, we have:
	\begin{align*}
	\langle \mathbf{A}\mathbf{G}^{-1}\mathbf{A}^T\mathbf{y},\mathbf{z}\rangle=0
	&~~\implies\langle\mathbf{G}^{-1}\mathbf{A}^T\mathbf{y}, \mathbf{A}^T\mathbf{z}\rangle=0 \\
	&~~\implies\langle\mathbf{G}^{-1}\mathbf{A}^T\mathbf{y}, \mathbf{A}^T\mathbf{y}\rangle=0 \quad \text{(take $\mathbf{z}=\mathbf{y}$)}.
	\end{align*}
	Notice that $\mathbf{G}^{-1}$ is positive definite, so we have $\mathbf{A}^T\mathbf{y}=0$,
	which is equivalent to $\mathbf{y}\perp R(\mathbf{A})$. Hence we get $\text{ker}(\mathbf{K})\cap R(\mathbf{A})=\{0\}$, which completes the proof.
\end{proof}
The next proposition provides a characterization of $\mathbf{u}^{k+1}$:
\begin{prop}
	The \zxu{} iteration~\eqref{eq:Sepzxu_z}-\eqref{eq:Sepzxu_u} satisfies:
	\begin{align}\label{eq:H123}
	\mathbf{u}^{k+1}=\mathbf{A}\mathbf{x}^{k+1}-\mathbf{A}\mathbf{x}^{k}+\frac{1}{\mu}\mathbf{B}^{-T}\nabla g(\mathbf{z}^{k+1}).
	\end{align}
\end{prop}
\begin{proof}
	From \eqref{eq:Sepzxu_u}:
	\begin{align}\label{eq:H13}
	\mathbf{u}^{k+1}-\mathbf{A}\mathbf{x}^{k+1}&=\mathbf{u}^k-\mathbf{B}\mathbf{z}^{k+1}.
	\end{align}
	From \eqref{eq:xfromu}:
		\begin{align}
		&~~\mathbf{A}\mathbf{x}^{k}+\mathbf{u}^k=\mathbf{A}\tilde{\mathbf{x}}-\mu\mathbf{K}\mathbf{u}^k+\mathbf{u}^k \nonumber \\
		\overset{\text{by~\eqref{eq:Sepzxu_z}}}{\implies} &~~\mathbf{B}\mathbf{z}^{k+1}+\frac{1}{\mu}\mathbf{B}^{-T}\nabla g(\mathbf{z}^{k+1})= \mathbf{A}\mathbf{x}^{k}+\mathbf{u}^k\nonumber\\
		\implies &~~ \frac{1}{\mu}\mathbf{B}^{-T}\nabla g(\mathbf{z}^{k+1})= \mathbf{A}\mathbf{x}^{k}+\mathbf{u}^k-\mathbf{B}\mathbf{z}^{k+1}. \label{eq:H14}
		\end{align}
	Combine \eqref{eq:H13} with \eqref{eq:H14} then we can get the result.
\end{proof}
Now we are able to bound both $\|\mathbf{u}^k\|$ and $\|\mathbf{u}^{k+1}-\mathbf{u}^k\|$:
\begin{prop}
	\label{propi-3}
	For \zxu{} iteration~\eqref{eq:Sepzxu_z}-\eqref{eq:Sepzxu_u} and $k\geq 1$ we have:
	{
	\squeezemath{2.5}
	\begin{align}
	&\|\mathbf{u}^k\|^2\leq\frac{4}{\eta^2\mu^2}\|\mathbf{A}\mathbf{x}^k-\mathbf{A}\mathbf{\tilde{x}}\|^2 +(\frac{4\rho(\mathbf{K})^2}{\mu^2\eta^2}+\frac{2}{\mu^2})\|\mathbf{B}^{-T}\nabla g(\mathbf{z}^{k})\|^2,\label{eq:H15}\\
	&\|\mathbf{u}^{k+1}-\mathbf{u}^k\|^2\leq\frac{4}{\mu^2\eta^2}\|\mathbf{A}\mathbf{x}^{k+1}-\mathbf{A}\mathbf{x}^{k}\|^2\notag\\
	&\hspace*{6em} +(\frac{4\rho(\mathbf{K})^2}{\mu^2\eta^2}+\frac{2}{\mu^2})\|\mathbf{B}^{-T}(\nabla g(\mathbf{z}^{k+1})-\nabla g(\mathbf{z}^{k}))\|^2.  \label{eq:H16}
	\end{align}
	}
\end{prop}
\begin{proof}
	To prove \eqref{eq:H15}, note that from \eqref{eq:H123}:\begin{align}\label{eq:H17}
	\|\mathbf{u}^{k}\|^2\leq 2\|\mathbf{A}\mathbf{x}^{k}-\mathbf{A}\mathbf{x}^{k-1}\|^2+\frac{2}{\mu^2}\|\mathbf{B}^{-T}\nabla g(\mathbf{z}^{k}) \|^2.
	\end{align}
	And from \eqref{eq:xfromu}:
	\begin{align}
	&~~\mathbf{A}\mathbf{x}^{k+1}=\mathbf{A}\mathbf{\tilde{x}}-\mu\mathbf{K}\mathbf{u}^{k+1} \label{eq:H27}\\
	\overset{\text{by \eqref{eq:H123}}}{\implies}&~~\mathbf{A}\mathbf{x}^{k+1}=\mathbf{A}\mathbf{\tilde{x}}-\mu\mathbf{K}(\mathbf{A}\mathbf{x}^{k+1}-\mathbf{A}\mathbf{x}^k)-\mathbf{K}\mathbf{B}^{-T}\nabla g(\mathbf{z}^{k+1})\nonumber\\
	\implies&~~\mathbf{A}\mathbf{x}^{k+1}-\mathbf{A}\mathbf{\tilde{x}}+\mathbf{K}\mathbf{B}^{-T}\nabla g(\mathbf{z}^{k+1})=-\mu\mathbf{K}(\mathbf{A}\mathbf{x}^{k+1}-\mathbf{A}\mathbf{x}^k).\label{eq:H29}
	\end{align}
	Hence by Proposition~\ref{prop:i-1}:
	\begin{equation*}
	\mu^2\eta^2\|\mathbf{A}\mathbf{x}^{k+1}-\mathbf{A}\mathbf{x}^k\|^2\leq2\|\mathbf{A}\mathbf{x}^{k+1}-\mathbf{A}\mathbf{\tilde{x}}\|^2+2\rho(\mathbf{K})^2\|\mathbf{B}^{-T}\nabla g(\mathbf{z}^{k+1})\|^2,
	\end{equation*}
	and \eqref{eq:H15} follows from this equation and \eqref{eq:H17}. For \eqref{eq:H16}, from \eqref{eq:H123}:
	\begin{align}
	& \mathbf{u}^{k+1}-\mathbf{u}^k=\mathbf{A}(\mathbf{x}^{k+1}-2\mathbf{x}^{k}+\mathbf{x}^{k-1})+\frac{1}{\mu}\mathbf{B}^{-T}(\nabla g(\mathbf{z}^{k+1})-\nabla g(\mathbf{z}^{k}))\nonumber\\
	\implies& \|\mathbf{u}^{k+1}-\mathbf{u}^k\|^2 \leq2\|\mathbf{A}(\mathbf{x}^{k+1}-2\mathbf{x}^{k}+\mathbf{x}^{k-1})\|^2 \nonumber\\
	&\qquad \qquad  \qquad \; \; \, +\frac{2}{\mu^2}\|\mathbf{B}^{-T}(\nabla g(\mathbf{z}^{k+1})-\nabla g(\mathbf{z}^{k})) \|^2. \label{eq:H32}
	\end{align}
	And by \eqref{eq:H29}:
	\begin{equation*}
		\mathbf{A}\mathbf{x}^{k+1}-\mathbf{A}\mathbf{x}^k+\mathbf{K}\mathbf{B}^{-T}(\nabla g(\mathbf{z}^{k+1})-\nabla g(\mathbf{z}^{k}))=-\mu\mathbf{K}\mathbf{A}(\mathbf{x}^{k+1}-2\mathbf{x}^{k}+\mathbf{x}^{k-1}).
	\end{equation*}
	Hence:
	\begin{equation}
	\begin{aligned}
	& \; \mu^2\eta^2\|\mathbf{A}(\mathbf{x}^{k+1}-2\mathbf{x}^{k}+\mathbf{x}^{k-1})\|^2\\
	\leq & \; 2\|\mathbf{A}\mathbf{x}^{k+1}-\mathbf{A}\mathbf{x}^k\|^2 + 2\rho(\mathbf{K})^2\|\mathbf{B}^{-T}(\nabla g(\mathbf{z}^{k+1})-\nabla g(\mathbf{z}^{k})) \|^2.
	\end{aligned}
	\label{eq:H34}
	\end{equation}
	Then \eqref{eq:H16} follows from \eqref{eq:H32} and \eqref{eq:H34}.
\end{proof}
Similar to Proposition~\ref{propG-2}, we can prove:\begin{prop}
	\label{propI-4}
	Suppose Assumptions~\ref{assump:g} and \ref{initial_zxu} hold, and $\mu$ is sufficiently large. The the \zxu{} iteration satisfies:
	\begin{equation}
	T(\mathbf{x}^k,\mathbf{z}^{k})\leq T(\mathbf{x}^0,\mathbf{z}^0)+c_2+c_3,~~~ L(\mathbf{x}^{k},\mathbf{z}^{k},\mathbf{u}^{k})\leq T(\mathbf{x}^0,\mathbf{z}^0)+c_3.
	\label{eq:H35}
	\end{equation}
\end{prop}
\begin{proof}
	We will prove this by induction. For $k=0$ this is trivial, now assume \eqref{eq:H35} holds for every $k\leq l$. Consider $k=l+1$. By the definition of $\mathbf{z}^{l+1}$ in \eqref{eq:Sepzxu_z}:
	\begin{equation}
	g(\mathbf{z}^{l+1})+\frac{\mu}{2}\| \mathbf{A}\mathbf{x}^{l}-\mathbf{B}\mathbf{z}^{l+1}+\mathbf{u}^l\|^2 \leq g(\mathbf{z}^{l})+\frac{\mu}{2}\| \mathbf{A}\mathbf{x}^{l}-\mathbf{B}\mathbf{z}^{l}+\mathbf{u}^l\|^2.
	\label{eq:H36}
	\end{equation}
	By the definition of $\mathbf{x}^{l+1}$ in \eqref{eq:Sepzxu_z}:
	\begin{equation}
	f(\mathbf{x}^{l+1})+\frac{\mu}{2}\| \mathbf{A}\mathbf{x}^{l+1}-\mathbf{B}\mathbf{z}^{l+1}+\mathbf{u}^l\|^2 \leq f(\mathbf{x}^l)+\frac{\mu}{2}\| \mathbf{A}\mathbf{x}^{l}-\mathbf{B}\mathbf{z}^{l+1}+\mathbf{u}^l\|^2.
	\label{eq:H37}
	\end{equation}
	add \eqref{eq:H37} to \eqref{eq:H36}:
	\begin{equation*}
	T(\mathbf{x}^{l+1},\mathbf{z}^{l+1})\leq L(\mathbf{x}^{l},\mathbf{z}^l,\mathbf{u}^l)+\frac{\mu}{2}\|\mathbf{u}^l\|^2.
	\end{equation*}
	By induction:
	\begin{equation*}
	L(\mathbf{x}^{l},\mathbf{z}^l,\mathbf{u}^l)\leq T(\mathbf{x}^0,\mathbf{z}^0)+c_3.
	\end{equation*}
	Since $l+1\geq1$, by Proposition~\ref{propi-3}:
	\begin{equation*}
	\frac{\mu}{2}\|\mathbf{u}^l\|^2\leq\frac{2}{\eta^2\mu}\|\mathbf{A}\mathbf{x}^l-\mathbf{A}\mathbf{\tilde{x}}\|^2 +(\frac{2\rho(\mathbf{K})^2}{\mu\eta^2}+\frac{1}{\mu})\|\mathbf{B}^{-T}\nabla g(\mathbf{z}^{l})\|^2.
	\end{equation*}
	By induction:\begin{equation*}
	T(\mathbf{x}^l,\mathbf{z}^l)\leq T(\mathbf{x}^0,\mathbf{z}^0)+c_2+c_3\leq T(\mathbf{x}^0,\mathbf{z}^0)+1.
	\end{equation*}
	By the definition of $c_2$,
	$\frac{\mu}{2}\|\mathbf{u}^l\|^2\leq c_2$.
	Hence:
	\begin{equation*}
	T(\mathbf{x}^{l+1},\mathbf{z}^{l+1})\leq  T(\mathbf{x}^0,\mathbf{z}^0)+c_2+c_3,
	\end{equation*}
	which proves the first part.
	For the second part, we first prove that the conclusion holds for $l=0$ ($k=1$). From the first part we know:
	\begin{equation*}
	T(\mathbf{x}^1,\mathbf{z}^1)\leq T(\mathbf{x}^0,\mathbf{z}^0)+c_2+c_3.
	\end{equation*}
	Notice that $f(\mathbf{x}^1)\geq0$ so we have $g(\mathbf{z}^1)\leq T(\mathbf{x}^0,\mathbf{z}^0)+c_2+c_3$. Hence by Proposition~\ref{prop:GeneralProof1}:
	\begin{equation*}
	L(\mathbf{x}^0,\mathbf{z}^1,\mathbf{u}^0)\leq L(\mathbf{x}^0,\mathbf{z}^0,\mathbf{u}^0)-\frac{\mu-L_d}{2}\|\mathbf{B}\mathbf{z}^1-\mathbf{B}\mathbf{z}^0\|^2.
	\end{equation*}
	And by Assumption~\ref{assump:SPDG}:
	\begin{equation*}
	L(\mathbf{x}^1,\mathbf{z}^1,\mathbf{u}^0)\leq L(\mathbf{x}^0,\mathbf{z}^1,\mathbf{u}^0)-\frac{\mu}{2}\|\mathbf{A}\mathbf{x}^1-\mathbf{A}\mathbf{x}^0\|^2.
	\end{equation*}
	Moreover, we have:
	\begin{equation*}
	L(\mathbf{x}^1,\mathbf{z}^1,\mathbf{u}^1)=L(\mathbf{x}^1,\mathbf{z}^1,\mathbf{u}^0)+\mu\|\mathbf{u}^1-\mathbf{u}^0\|^2
	\end{equation*}
	By \eqref{eq:H15} and $\mathbf{u}^0=0$:
	\begin{align*}
	\mu\|\mathbf{u}^1-\mathbf{u}^0\|^2&=\mu\|\mathbf{u}^1\|^2\\
	&=\frac{4}{\eta^2\mu}\|\mathbf{A}\mathbf{x}^1-\mathbf{A}\mathbf{\tilde{x}}\|^2 +(\frac{4\rho(\mathbf{K})^2}{\mu\eta^2}+\frac{2}{\mu})\|\mathbf{B}^{-T}\nabla g(\mathbf{z}^{1})\|^2.
	\end{align*}
	Moreover, we have:\begin{align*}
	\|\mathbf{B}^{-T}\nabla g(\mathbf{z}^{1})\|^2&\leq 2\|\mathbf{B}^{-T}\nabla g(\mathbf{z}^{1})-\mathbf{B}^{-T}\nabla g(\mathbf{z}^{0}) \|^2+2\|\mathbf{B}^{-T}\nabla g(\mathbf{z}^{0})\|^2\\
	&\leq2L_d\|\mathbf{B}\mathbf{z}^1-\mathbf{B}\mathbf{z}^0\|^2+2\|\mathbf{B}^{-T}\nabla g(\mathbf{z}^{0})\|^2.
	\end{align*}
	So if $\frac{\mu}{2}\geq\frac{4}{\eta^2\mu}$ and $\frac{\mu-L_d}{2}\geq2L_d(\frac{4\rho(\mathbf{K})^2}{\mu\eta^2}+\frac{2}{\mu})$, then we have:\begin{align*}
	L(\mathbf{x}^1,\mathbf{z}^1,\mathbf{u}^1)&\leq L(\mathbf{x}^0,\mathbf{z}^0,\mathbf{u}^0)+(\frac{8\rho(\mathbf{K})^2}{\mu\eta^2}+\frac{4}{\mu})\|\mathbf{B}^{-T}\nabla g(\mathbf{z}^{0})\|^2\\
	&=T(\mathbf{x}^0,\mathbf{z}^0)+(\frac{8\rho(\mathbf{K})^2}{\mu\eta^2}+\frac{4}{\mu})\|\mathbf{B}^{-T}\nabla g(\mathbf{z}^{0})\|^2.
	\end{align*}
	By the definition of $c_3$ we have
	$L(\mathbf{x}^1,\mathbf{z}^1,\mathbf{u}^1)\leq T(\mathbf{x}^0,\mathbf{z}^0)+c_3$.
	Now suppose $l\geq1$. Similar to the proof of the case $l=0$ we have:\begin{align*}
	L(\mathbf{x}^l,\mathbf{z}^{l+1},\mathbf{u}^{l})&\leq L(\mathbf{x}^l,\mathbf{z}^l,\mathbf{u}^l)-\frac{\mu-L_d}{2}\|\mathbf{B}\mathbf{z}^{l+1}-\mathbf{B}\mathbf{z}^l\|^2,\\
	L(\mathbf{x}^{l+1},\mathbf{z}^{l+1},\mathbf{u}^l)&\leq L(\mathbf{x}^l,\mathbf{z}^{l+1},\mathbf{u}^l)-\frac{\mu}{2}\|\mathbf{A}\mathbf{x}^{l+1}-\mathbf{A}\mathbf{x}^l\|^2,\\
	L(\mathbf{x}^{l+1},\mathbf{z}^{l+1},\mathbf{u}^{l+1})&= L(\mathbf{x}^{l+1},\mathbf{z}^{l+1},\mathbf{u}^l)+\mu\|\mathbf{u}^{l+1}-\mathbf{u}^l\|^2.
	\end{align*}
	By \eqref{eq:H16} we have:
	\begin{align*}
	\mu\|\mathbf{u}^{l+1}-\mathbf{u}^l\|^2&\leq \frac{4}{\mu\eta^2}\|\mathbf{A}\mathbf{x}^{l+1}-\mathbf{A}\mathbf{x}^{l}\|^2\\
	&+(\frac{4\rho(\mathbf{K})^2}{\mu\eta^2}+\frac{2}{\mu})\|\mathbf{B}^{-T}(\nabla g(\mathbf{z}^{l+1})-\nabla g(\mathbf{z}^{l}))\|^2.
	\end{align*}
	Since $g(\mathbf{z}^l),g(\mathbf{z}^{l+1})\leq T(\mathbf{x}^0,\mathbf{z}^0)+c_2+c_3$, by the definition of $L_d$:
	\begin{equation}\label{eq:H53}
	\|\mathbf{B}^{-T}(\nabla g(\mathbf{z}^{l+1})-\nabla g(\mathbf{z}^{l}))\|\leq L_d\|\mathbf{B}\mathbf{z}^{l+1}-\mathbf{B}\mathbf{z}^l\|.
	\end{equation}
	Hence we have:\begin{align*}
	\mu\|\mathbf{u}^{l+1}-\mathbf{u}^l\|^2&\leq \frac{4}{\mu\eta^2}\|\mathbf{A}\mathbf{x}^{l+1}-\mathbf{A}\mathbf{x}^{l}\|^2\\
	&+(\frac{4\rho(\mathbf{K})^2L_d^2}{\mu\eta^2}+\frac{2L_d^2}{\mu})\|\mathbf{B}\mathbf{z}^{l+1}-\mathbf{B}\mathbf{z}^l\|^2.
	\end{align*}
	If $\frac{\mu}{2}\geq\frac{4}{\eta^2\mu}$ and $\frac{\mu-L_d}{2}\geq(\frac{4\rho(\mathbf{K})^2L_d^2}{\mu\eta^2}+\frac{2L_d^2}{\mu})$, then we have:\begin{align}\label{eq:H55}
	L(\mathbf{x}^{l+1},\mathbf{z}^{l+1},\mathbf{u}^{l+1})\leq L(\mathbf{x}^l,\mathbf{z}^{l},\mathbf{u}^{l})\leq T(\mathbf{x}^0,\mathbf{z}^0)+c_3
	\end{align}
	which completes the proof.
\end{proof}

Similar to the proof of Theorem~\ref{thm:GeneralConvergenceXZU}, we need to show (P1)-(P4) hold for \zxu{} iteration. Sufficient descent has already been shown in the proof of Proposition~\ref{propI-4}. The remaining part is the same as the proof of Theorem~\ref{thm:GeneralConvergenceXZU} so we omit it.

\subsection{Proof for Theorem~\ref{thm:zxuConvergenceLocalLipschitz}}
\label{Proof_Thm3.4}
Theorem~\ref{thm:zxuConvergenceLocalLipschitz} is about linear convergence of the \zxu{} iteration.
Similar to Proposition~\ref{prop:StationaryPointConvergenceXZU}, for the convergence of the \zxu{} iteration to a stationary point, it suffices to show that the sequence $\{\mathbf{u}^k\}$ converges.
%
Then for the main proof:
\begin{proof}[Proof for Theorem~\ref{thm:zxuConvergenceLocalLipschitz}]
	By \eqref{eq:H14}:
	\begin{align*}
	& \mathbf{B}\mathbf{z}^{k+1}+\frac{1}{\mu}\mathbf{B}^{-T}\nabla g(\mathbf{z}^{k+1})= \mathbf{A}\tilde{\mathbf{x}}-\mathbf{v}^{k}\\
	\implies & \mathbf{B}(\mathbf{z}^{k+1}-\mathbf{z}^{k})+\frac{1}{\mu}\mathbf{B}^{-T}(\nabla g(\mathbf{z}^{k+1})-\nabla g(\mathbf{z}^{k}))=-(\mathbf{v}^{k}-\mathbf{v}^{k-1}).
	\end{align*}
	By \eqref{eq:H53}:
	\begin{equation*}
	(1-\frac{L_d}{\mu})\|\mathbf{B} \mathbf{z}^{k+1}-\mathbf{B} \mathbf{z}^{k} \|\leq\|\mathbf{v}^{k}-\mathbf{v}^{k-1}\|.
	\end{equation*}
	Hence we have:
	\begin{align}
	\frac{1}{\mu}\|\mathbf{B}^{-T}(\nabla g(\mathbf{z}^{k+1})-\nabla g(\mathbf{z}^{k}))\|&\leq \frac{L_d}{\mu}\|\mathbf{B}\mathbf{z}^{k+1}-\mathbf{B}\mathbf{z}^k\|\notag\\
	&\leq\frac{L_d}{\mu-L_d}\|\mathbf{v}^{k}-\mathbf{v}^{k-1}\|. \label{eq:H60}
	\end{align}
	By \eqref{eq:H123} and \eqref{eq:H27}:
	\begin{align*}
	&(\mathbf{I}+\mu\mathbf{K})\mathbf{u}^{k+1}=\mu\mathbf{K}\mathbf{u}^k+\frac{1}{\mu}\mathbf{B}^{-T}\nabla g(\mathbf{z}^{k+1})\\
	\implies &
	\mathbf{v}^{k+1}=(\mathbf{I}+\mu\mathbf{K})^{-1}\mu\mathbf{K}\mathbf{v}^k+(\mathbf{I}+\mu\mathbf{K})^{-1}(\mathbf{I}-\mu\mathbf{K})\frac{1}{\mu}\mathbf{B}^{-T}\nabla g(\mathbf{z}^{k+1}).
	\end{align*}
	Hence we have:
	\begin{align*}
	\|\mathbf{v}^{k+1}-\mathbf{v}^{k}\|&\leq \|(\mathbf{I}+\mu\mathbf{K})^{-1}\mu\mathbf{K}(\mathbf{v}^{k}-\mathbf{v}^{k-1})\|\\
	&+\frac{1}{\mu}\|(\mathbf{I}+\mu\mathbf{K})^{-1}(\mathbf{I}-\mu\mathbf{K})\mathbf{B}^{-T}(\nabla g(\mathbf{z}^{k+1})-\nabla g(\mathbf{z}^{k}))\|.
	\end{align*}
	Similar to \eqref{eq:H6} we have:
	\begin{align*}
	&~~\|(\mathbf{I}+\mu\mathbf{K})^{-1}\mu\mathbf{K}(\mathbf{v}^{k}-\mathbf{v}^{k-1})\|\leq\frac{\mu\rho(\mathbf{K})}{1+\mu\rho(\mathbf{K})}\|\mathbf{v}^{k}-\mathbf{v}^{k-1}\|\\
	\overset{\text{by \eqref{eq:H60}}}{\implies} &~~\|\mathbf{v}^{k+1}-\mathbf{v}^{k}\|\leq(\frac{\mu\rho(\mathbf{K})}{1+\mu\rho(\mathbf{K})}+\frac{L_d}{\mu-L_d} )\|\mathbf{v}^{k}-\mathbf{v}^{k-1}\|.
	\end{align*}
	Then let $\mu>\max\left\{\frac{2}{\frac{1}{L_d}-\rho(\mathbf{K})},\frac{1}{L_d}\right\}$ and we get the result.
\end{proof}
\subsection{Sketch of proofs for Theorems~\ref{thm:xzuConvergenceGlobalLipschitz} and \ref{thm:zxuConvergenceGlobalLipschitz}}
\label{Proof_Thm3.1_3.2}
For the proof of Theorem~\ref{thm:xzuConvergenceGlobalLipschitz}, the derivation can start from \eqref{eq:H4} without assumptions on the initial values. The rest of the proofs is the same as the proof of Theorem~\ref{thm:xzuConvergenceLocalLipschitz}.

For the proof of Theorem~\ref{thm:zxuConvergenceGlobalLipschitz}, the derivation of \eqref{eq:H53} does not rely on the initial value. The rest is the same.

\end{document}


\title{Supplementary Material for\\ \textit{Accelerating ADMM for Efficient Simulation and Optimization}}
\authorsaddresses{}
\maketitle
\thispagestyle{empty}

\section{Background}

In this supplementary material, we will verify the linear convergence theorems with a target function $g$ that is locally Lipschitz differentiable (Theorems 3.3 and 3.4). We will use ADMM to solve the following optimization problem from~\cite{Overby2017} for physical simulation:
\begin{equation}
\min_{\mathbf{x},\mathbf{z}} ~~ f(\mathbf{x}) + g(\mathbf{z})\quad
\textrm{s.t.} ~~ \mathbf{W}(\mathbf{z} - \mathbf{D} \mathbf{x}) = 0,
\label{eq:OverbyProblem}
\end{equation}
Here $\mathbf{x}$ is the node positions of the discretized object. $f(\mathbf{x})$ is a momentum energy of the form
\[
f(\mathbf{x}) = \frac{1}{2} (\mathbf{x} - \tilde{\mathbf{x}})^T \mathbf{G} (\mathbf{x} - \tilde{\mathbf{x}}),
\]
with $\tilde{\mathbf{x}}$ being a constant vector, and $\mathbf{G}$ being a scaled mass matrix. $\mathbf{D} \mathbf{x}$ collects the deformation gradient of each element. $\mathbf{W}$ is a diagonal scaling matrix that improves conditioning. $g(\mathbf{z})$ is an elastic potential energy. Compared to the following form of optimization problems discussed in our paper:
\begin{equation}
\min_{\mathbf{x},\mathbf{z}} ~~ f(\mathbf{x}) + g(\mathbf{z})\quad
\textrm{s.t.} ~~ \mathbf{A}\mathbf{x} - \mathbf{B} \mathbf{z} = \mathbf{c},
\label{eq:SeparableProblem}
\end{equation}
we can see that problems~\ref{eq:OverbyProblem} and \ref{eq:SeparableProblem} are equivalent if $\mathbf{A} =\mathbf{W}\mathbf{D}, \mathbf{B} = \mathbf{W}, \mathbf{c} = \mathbf{0}$.
In this report, we assume the simulation object to be a tetrahedral mesh and use the following potential energy:
\[
	g(\mathbf{z})=\sum_{i=1}^{n_t}v_i\psi(\mathbf{F}^i),
\]
where $v_i$ is the volume of each tetrahedron, and $\mathbf{F}^i \in \mathbb{R}^{3 \times 3}$ is its deformation gradient with respect ot the rest shape, and $\psi$ is the strain energy density function of StVK  materials~\cite{Sifakis2012}:
\begin{align}
\label{1}
\psi(\mathbf{F})=\lambda_1 \mathbf{E}:\mathbf{E}+\frac{1}{2}\lambda_2 tr^2(\mathbf{E}),
\end{align}
where $\lambda_1, \lambda_2$ are given material parameters, and
\[
\mathbf{E} = \frac{1}{2}(\mathbf{F}^T\mathbf{F}-\mathbf{I}) \in \mathbb{R}^{3 \times 3}
\]
and $\mathbf{I}$ is the identity matrix.

The linear convergence theorems we will verify require a local Lipschitz constant for the gradient of $g$. Thus we need to analyze the Lipschitz differentiability of $\psi$. The gradient of $\psi$ is the first Piola-Kirchhoff stress tensor:
\begin{align}
\label{3}
\mathbf{P}(\mathbf{F}) = \mathbf{F}\mathbf{S}(\mathbf{F}).
\end{align}
where $\mathbf{S}$ is the second Piola-Kirchhoff stress tensor:
\[
	\mathbf{S}(\mathbf{F}) = 2\lambda_1 \mathbf{E}+\lambda_2 tr(\mathbf{E})\mathbf{I},
\]
Note that the value of $\mathbf{E}$ depends on $\mathbf{F}$. In the following, we will denote $\mathbf{E}_j = \mathbf{E}(\mathbf{F}_j)$ for a subscript $j$. Throughout this report, for a matrix $\mathbf{A}$, $\|\mathbf{A}\|$ denotes its Frobenius norm and $\|\mathbf{A}\|_2$ denotes its $l_2$ norm.
Our task is to estimate the Lipschitz constant of $\mathbf{P}(\mathbf{F})$ with respect to $\mathbf{F}$.
\begin{prop}
\label{prop1}
$\|\mathbf{P}(\mathbf{F}_1)-\mathbf{P}(\mathbf{F}_2)\|\leq((\lambda_1+\frac{\sqrt{3}}{2}\lambda_2)\|\mathbf{F}_1\|(\|\mathbf{F}_1\|+\|\mathbf{F}_2\|)+(2\lambda_1+3\lambda_2)\|\mathbf{E}_2\|)\|\mathbf{F}_1-\mathbf{F}_2 \|.$
\end{prop}
\begin{proof}
We have:\begin{align}\label{4}
\mathbf{E}_1-\mathbf{E}_2&=\frac{\mathbf{F}_1^T\mathbf{F}_1-\mathbf{F}_2^T\mathbf{F}_2}{2} \notag\\
 &=\frac{1}{2}(\mathbf{F}_1^T\mathbf{F}_1-\mathbf{F}_1^T\mathbf{F}_2+\mathbf{F}_1^T\mathbf{F}_2-\mathbf{F}_2^T\mathbf{F}_2) \notag\\
 &=\frac{1}{2}(\mathbf{F}_1^T(\mathbf{F}_1-\mathbf{F}_2)+(\mathbf{F}_1^T-\mathbf{F}_2^T)\mathbf{F}_2 )
\end{align}
Hence:\begin{align}\label{5}
\|\mathbf{E}_1-\mathbf{E}_2\|\leq\frac{1}{2}(\|\mathbf{F}_1\|+\|\mathbf{F}_2\|)\|\mathbf{F}_1-\mathbf{F}_2\|
\end{align}
And:\begin{align}\label{6}
\|tr(\mathbf{E}_1)\mathbf{I}-tr(\mathbf{E}_2)\mathbf{I}\|&=\sqrt{3}|tr(\mathbf{E}_1)-tr(\mathbf{E}_2)|\notag\\
&=\sqrt{3}|tr(\mathbf{E}_1-\mathbf{E}_2)|\notag\\
&=\frac{\sqrt{3}}{2}|tr(\mathbf{F}_1^T(\mathbf{F}_1-\mathbf{F}_2)+(\mathbf{F}_1^T-\mathbf{F}_2^T)\mathbf{F}_2)|\notag\\
&\leq\frac{\sqrt{3}}{2}(\|\mathbf{F}_1\|+\|\mathbf{F}_2\|)\|\mathbf{F}_1-\mathbf{F}_2\|
\end{align}

For $\mathbf{P}(\mathbf{F})$ we have:\begin{align}\label{7}
\mathbf{P}(\mathbf{F}_1)-\mathbf{P}(\mathbf{F}_2)&=\mathbf{F}_1\mathbf{S}(\mathbf{F}_1)-\mathbf{F}_2\mathbf{S}(\mathbf{F}_2)\notag\\
&=\mathbf{F}_1\mathbf{S}(\mathbf{F}_1)-\mathbf{F}_1\mathbf{S}(\mathbf{F}_2)+\mathbf{F}_1\mathbf{S}(\mathbf{F}_2)-\mathbf{F}_2\mathbf{S}(\mathbf{F}_2)\notag\\
&=\mathbf{F}_1(\mathbf{S}(\mathbf{F}_1)-\mathbf{S}(\mathbf{F}_2))+(\mathbf{F}_1-\mathbf{F}_2)\mathbf{S}(\mathbf{F}_2)
\end{align}
Therefore:\begin{align}\label{8}
\|\mathbf{P}(\mathbf{F}_1)-\mathbf{P}(\mathbf{F}_2)\|\leq\|\mathbf{F}_1\|\|\mathbf{S}(\mathbf{F}_1)-\mathbf{S}(\mathbf{F}_2)\|+\|\mathbf{F}_1-\mathbf{F}_2\|\|\mathbf{S}(\mathbf{F}_2)\|
\end{align}
By \eqref{5} and \eqref{6} we have:\begin{align}\label{9}
\|\mathbf{S}(\mathbf{F}_1)-\mathbf{S}(\mathbf{F}_2)\|\leq (\lambda_1+\frac{\sqrt{3}}{2}\lambda_2)(\|\mathbf{F}_1\|+\|\mathbf{F}_2\|)\|\mathbf{F}_1-\mathbf{F}_2\|
\end{align}
We next estimate $\|\mathbf{S}(\mathbf{F}_2)\|$:\begin{align}\label{10}
\|\mathbf{S}(\mathbf{F}_2)\|&\leq 2\lambda_1\|\mathbf{E}_2\|+\lambda_2\sqrt{3}|tr(\mathbf{E}_2)|\notag\\
&\leq (2\lambda_1+3\lambda_2)\|\mathbf{E}_2\|
\end{align}
The result comes from \eqref{5}, \eqref{6}, \eqref{8} and \eqref{10}.
\end{proof}
\begin{prop}
\label{prop2}
Assume $\|\mathbf{F}\|^2\geq27$, then we have:\begin{align}\label{11}
\psi(\mathbf{F})\geq(\frac{16}{729}\lambda_1+\frac{72}{729}\lambda_2)\|\mathbf{F}\|^4
\end{align}
\end{prop}
\begin{proof}
Assume the singular values of $\mathbf{F}$ are $\sigma_1\geq\sigma_2\geq\sigma_3$. Then we have:\begin{align}
\|\mathbf{F}\|^2&=\sum_{i=1}^{3}\sigma_i^2\\
\|\mathbf{E}\|^2&=\frac{1}{4}\sum_{i=1}^{3}(\sigma_i^2-1)^2\label{12}\\
tr^2(\mathbf{E})&=\frac{1}{4}(\sum_{i=1}^{3}\sigma_i^2-3)^2
\end{align}
Since $\|\mathbf{F}\|^2\geq27$ we have $\sigma_1\geq3$, and we have:\begin{align}\label{13}
(\sigma_1^2-1)^2=\max_{i=1,2,3}(\sigma_i^2-1)^2
\end{align}
Hence:\begin{align}\label{14}
\|\mathbf{E}\|^2&\geq\frac{1}{4}(\sigma_1^2-1)^2\geq\frac{1}{4}(\frac{8}{9}\sigma_1^2)^2\geq\frac{1}{4}(\frac{8}{27}\|\mathbf{F}\|^2)^2\\
tr^2(\mathbf{E})&\geq\frac{1}{4}(\frac{24}{27}\|\mathbf{F}\|^2)^2
\end{align}
The result comes from \eqref{13} and \eqref{14}.
\end{proof}
\begin{prop}
\label{prop3}
$\conv{\lev{a}{\psi}}\subset \{\mathbf{F}\mid\|\mathbf{F}\|\leq b\}$,
where $b=\max\{3\sqrt{3},\sqrt[4]{\frac{a}{\frac{16}{729}\lambda_1+\frac{72}{729}\lambda_2}}\}$.
\end{prop}
\begin{proof}
Since $\{\mathbf{F}\mid\|\mathbf{F}\|\leq b\}$ is convex, it suffice to show $\lev{a}{\psi}\subset \{\mathbf{F}\mid\|\mathbf{F}\|\leq b\}$. Now assume $\phi(\mathbf{F})\leq a$. If $\|\mathbf{F}\|\leq 3\sqrt{3}$ this is trivial, otherwise by Proposition~\ref{prop2} we have:\begin{align}\label{15}
(\frac{16}{729}\lambda_1+\frac{72}{729}\lambda_2)\|\mathbf{F}\|^4\leq a
\end{align}
which completes the proof.
\end{proof}
\begin{prop}
\begin{itemize}
  \item[(1):] The value $((2\lambda_1+\sqrt{3}\lambda_2)b^2+(\lambda_1+\frac{3}{2}\lambda_2)\sqrt{b^4+3})$ is a Lipschitz constant of $\mathbf{P}(\mathbf{F})$ over the set $\conv{\lev{a}{\psi}}$,
  where $b$ is defined in Proposition~\ref{prop3}.
  \item[(2):] $\sup\limits_{\mathbf{F}\in\lev{a}{\psi}}\|\mathbf{P}(\mathbf{F})\|^2\leq (2\lambda_1+3\lambda_2)^2b^2(\frac{1}{4}b^4+\frac{3}{4})$
\end{itemize}
\end{prop}
\begin{proof}
Assume the singular values of $F$ are $\sigma_1\geq\sigma_2\geq\sigma_3$, by \eqref{12} we have:\begin{align}\label{16}
\|\mathbf{E}\|^2\leq \frac{1}{4}\sum_{i=1}^{3}\sigma_i^4+\frac{3}{4}\leq\frac{1}{4}(\sum_{i=1}^{3}\sigma_i^2)^2+\frac{3}{4}=\frac{1}{4}\|\mathbf{F}\|^4+\frac{3}{4}
\end{align}
Now assume $\mathbf{F}_1,\mathbf{F}_2\in \conv{\lev{a}{\psi}}$, by Proposition~\ref{prop3} we have:\begin{align}\label{17}
\|\mathbf{F}_1\|\leq b,\|\mathbf{F}_2\|\leq b
\end{align}
By Proposition~\ref{prop1} and \eqref{17} we have:\begin{align}\label{18}
\|\mathbf{P}(\mathbf{F}_1)-\mathbf{P}(\mathbf{F}_2)\|&\leq ((2\lambda_1+\sqrt{3}\lambda_2)b^2+(\lambda_1+\frac{3}{2}\lambda_2)\sqrt{b^4+3})\|\mathbf{F}_1-\mathbf{F}_2\|
\end{align}
which proves (1).

For (2), by \eqref{10}, Proposition~\ref{prop3} and \eqref{16}:\begin{align}\label{19}
\|\mathbf{P}(\mathbf{F})\|^2\leq b^2\|\mathbf{S}(\mathbf{F})\|^2\leq b^2(2\lambda_1+3\lambda_3)^2\|\mathbf{E}\|^2\leq (2\lambda_1+3\lambda_2)^2b^2(\frac{1}{4}b^4+\frac{3}{4})
\end{align}
\end{proof}

\section{Verification for the \xzu{} iteration}
\begin{figure}[t]
	\centering
	\includegraphics[width=6.5cm]{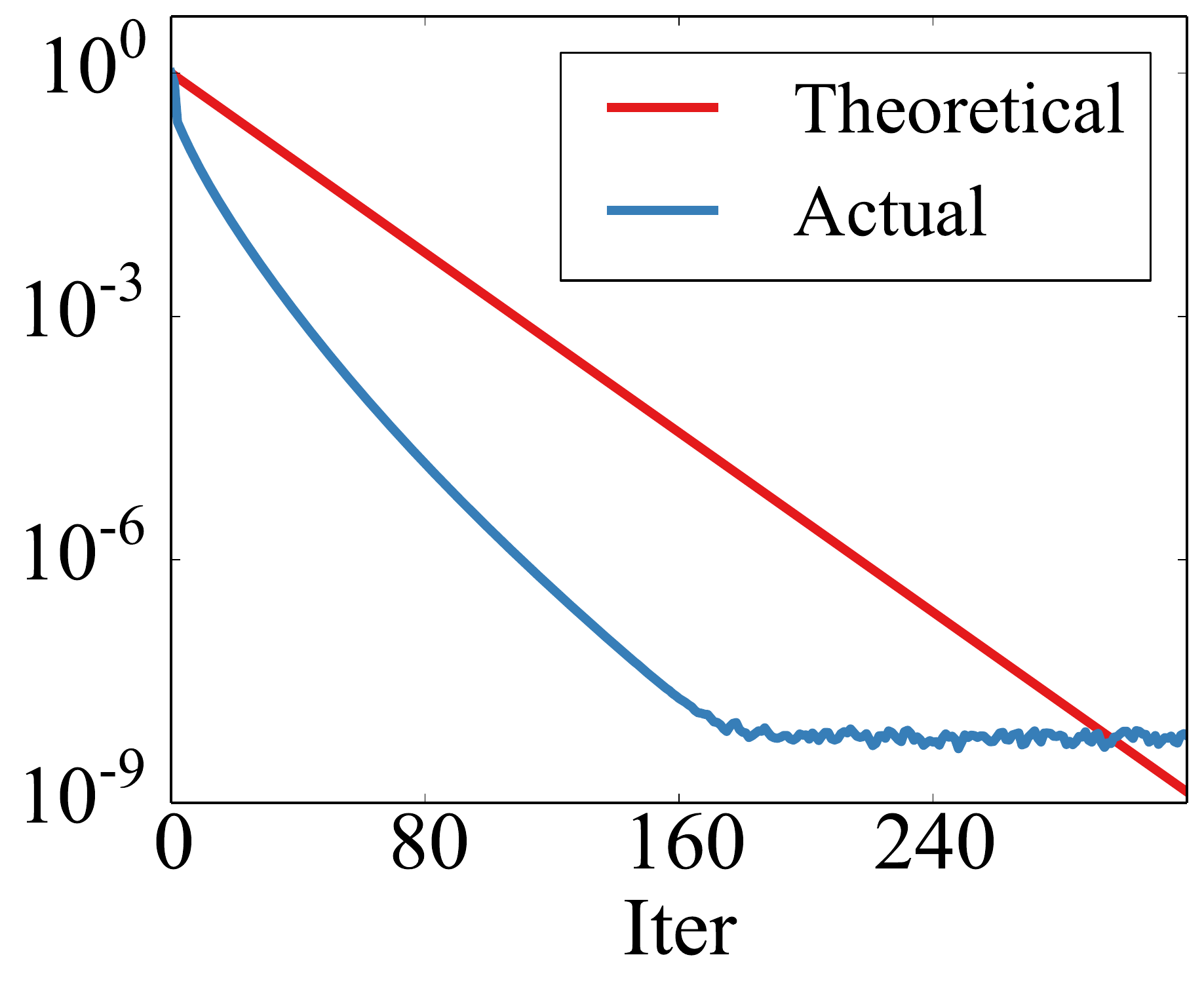}
	\caption{Comparison between theoretical and actual shrinkage of $\|\mathbf{B}\mathbf{z}^{k+1}-\mathbf{B}\mathbf{z}^k\|$ for verification of Theorem 3.3. The values are normalized by $\|\mathbf{B}\mathbf{z}^{1}-\mathbf{B}\mathbf{z}^0\|$ and plotted in logarithmic scale. The actual shrinkage ratio stays below the theoretical upper bound, before it oscillates due to numerical error when close enough to the solution.}
	\label{fig:xzu}
\end{figure}
In this subsection, we construct an example to verify the linear convergence theorem for the \xzu{} iteration (Theorem 3.3).
Assume the function $J:\mathbb{R}^9\rightarrow\mathbb{R}^{3\times3}$ assembles a vector into its matrix form. Assume $\mathbf{z}\in\mathbb{R}^{9n_t}$, $\mathbf{z}_i\in\mathbb{R}^{9}$ is its component, where $i\in[1,n_t]$. Then $g(\mathbf{z})$ becomes:
\begin{align}
\label{20}
g(\mathbf{z})=\sum_{i=1}^{n_t}v_i\psi(J(\mathbf{z}_i)).
\end{align}
The next proposition is just a corollary from the proofs of previous propositions so we omit its proof.
\begin{prop}
\label{prop2-1}
\begin{itemize}
  \item[(1):]The value $\max\limits_{i}v_i((2\lambda_1+\sqrt{3}\lambda_2)b_i^2+(\lambda_1+\frac{3}{2}\lambda_2)\sqrt{b_i^4+3})$ is a Lipschitz constant of $\nabla g(z)$ over the set $\conv{\lev{a}{g}}$,
  where $b_i=\max\{3\sqrt{3},\sqrt[4]{\frac{a/v_i}{\frac{16}{729}\lambda_1+\frac{72}{729}\lambda_2}}\}$.
  \item[(2):] $\sup\limits_{\mathbf{z}\in\lev{a}{g}}\|\nabla g(\mathbf{z})\|^2\leq \sum_{i=1}^{m}v_i^2      (2\lambda_1+3\lambda_2)^2b_i^2(\frac{1}{4}b_i^4+\frac{3}{4})$.
\end{itemize}
\end{prop}
For the sake of simplicity, we suppose $\mathbf{B}=\mathbf{I}$ in Eq.~\ref{eq:SeparableProblem}.
The optimization problem for the verification example is constructed via the following procedure:
\begin{enumerate}[label=\arabic*.]
\item Choose the initial value as suggested by Assumption 3.5.
\item Let $a=T^0+1$. Compute $L_c$ and $\sup\limits_{\mathbf{z}\in\lev{a}{g}}\|\nabla g(\mathbf{z})\|^2$ using Proposition~\ref{prop2-1}.
\item Choose a matrix $\mathbf{G}$ that is large enough such that $\rho(\mathbf{K})\leq\frac{1}{2L_c}$.
\item Choose a $\mu$ that is large enough such that $c_1\leq1$ as defined in Assumption 3.5, $\frac{\mu}{2}-\frac{L_c^2}{\mu}\geq \frac{L_c}{2}$ and $\mu>\max\{\frac{1}{\frac{1}{2L_c}-\rho(\mathbf{K})},
    \frac{1}{L_c}\}$.
\end{enumerate}
Fig.~\ref{fig:xzu} plots in logarithmic scale the value of $\|\mathbf{B}\mathbf{z}^{k+1}-\mathbf{B}\mathbf{z}^k\|$ throughout the iterations, as well as a straight line where the value changes according to the constant upper bound of shrinkage ratio given in Theorem 3.3. We can see that the actual shrinkage ratio stays below the theoretical upper bound before convergence.

\section{Verification for the \zxu{} iteration}
\begin{figure}[t]
	\centering
	\includegraphics[width=6.5cm]{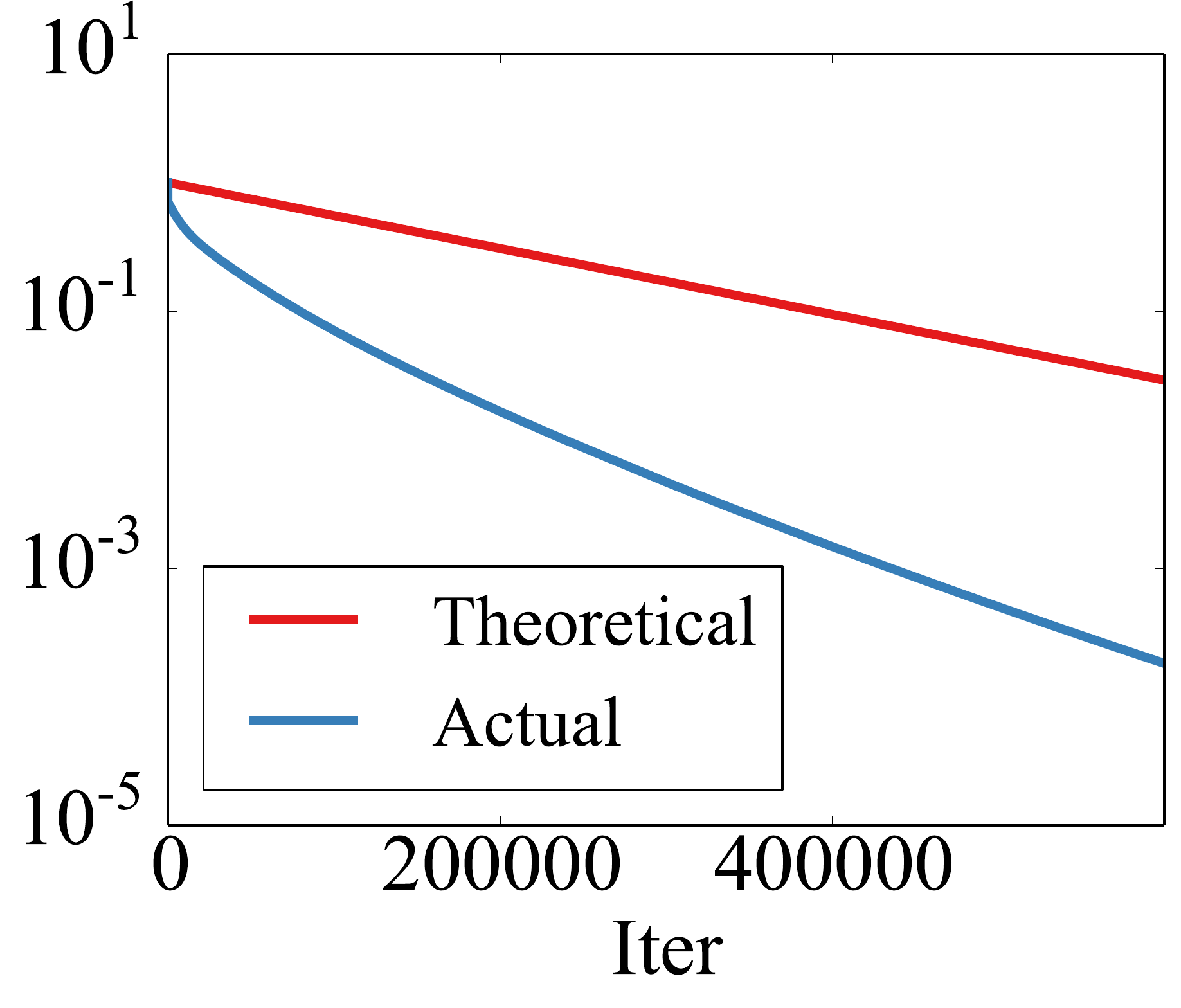}
	\caption{Comparison between theoretical and actual shrinkage of $\|\mathbf{v}^{k+1}-\mathbf{v}^k\|$ for verification of Theorem 3.4. The values are normalized by $\|\mathbf{v}^{1}-\mathbf{v}^0\|$ and plotted in logarithmic scale. The actual shrinkage ratio stays below the theoretical upper bound.}
	\label{fig:zxu}
\end{figure}

We now construct an example to verify the linear convergence theorem for the \zxu{} iteration (Theorem 3.4).
The first problem is to compute $\eta$. We first compute the SVD for $\mathbf{A}$. Suppose $\mathbf{A}=\mathbf{U}\mathbf{D}\mathbf{V}$ and let $r$ be the rank of $\mathbf{A}$. Let $\mathbf{U}_r$ be a sub-matrix of $\mathbf{U}$ consisting of its first $r$ columns. Then we know $R(\mathbf{A})=R(\mathbf{U}_r)$. Moreover, $\forall\mathbf{y}\in R(\mathbf{A})$, suppose $\mathbf{y}=\mathbf{U}_r\mathbf{z}$, then we have: $\|\mathbf{y}\|=\|\mathbf{z}\|$. Thus we choose $\eta$ to be the minimal eigenvalue of $\mathbf{U}_r^T\mathbf{K}\mathbf{U}_r$.

\begin{prop}
\label{prop3-1}
$\sup\limits_{x\in\lev{a}{f}}\|\mathbf{A}\mathbf{x}-\mathbf{A}\mathbf{\tilde{x}}\|^2\leq 2\|\mathbf{A}\|_2^2a/q $, where $q$ is the minimal eigenvalue of $\mathbf{G}$.
\end{prop}
\begin{proof}
By the definition of $f(\mathbf{x})$ we have:\begin{align}\label{21}
\lev{a}{f}=\{\mathbf{x}:\frac{1}{2}\|\mathbf{x}-\mathbf{\tilde{x}}\|_{\mathbf{G}}^2\leq a\}\subset \{\|\mathbf{x}-\mathbf{\tilde{x}}\|^2\leq 2a/q \}
\end{align}
Moreover, we have:\begin{align}\label{22}
\|\mathbf{A}\mathbf{x}-\mathbf{A}\mathbf{\tilde{x}}\|^2\leq\|\mathbf{A}\|_2^2\|\mathbf{x}-\mathbf{\tilde{x}}\|^2
\end{align}
which completes the proof.
\end{proof}
To determine $c_2$ defined in Assumption 3.6 we run one dimensional line-search for $t$ and compute the upper bound for $\sup\limits_{\mathbf{x}\in\lev{t}{f}}\frac{2}{\eta^2\mu}\|\mathbf{A}\mathbf{x}-\mathbf{A}\mathbf{\tilde{x}}\|^2+\sup\limits_{\mathbf{z}\in\lev{a-t}{g}}(\frac{2\rho(\mathbf{K})^2}{\mu\eta^2} +\frac{1}{\mu}) \| \mathbf{B}^{-T}\nabla g(\mathbf{z})\|^2$ by using Proposition~\ref{prop3-1} and Proposition~\ref{prop2-1}.

The optimization problem used of verification is constructed via the following procedure:
\begin{enumerate}[label=\arabic*.]
\item Choose initial value as suggested by Assumption 3.6. Then compute $\eta$.
\item Let $a=T^0+1$. Compute $L_d$ and $\sup\limits_{\mathbf{z}\in\lev{a}{g}}\|\nabla g(\mathbf{z})\|^2$ using Proposition~\ref{prop2-1}.
\item Choose a matrix $\mathbf{G}$ that is large enough such that $\rho(\mathbf{K})\leq\frac{1}{L_d}$.
\item Choose a $\mu$ that is large enough such that $c_2+c_3\leq1$ defined in Assumption 3.6, $\frac{\mu}{2}\geq\frac{4}{\eta^2\mu}$, $\frac{\mu-L_d}{2}\geq(\frac{4\rho(\mathbf{K})^2L_d^2}{\mu\eta^2}+\frac{2L_d^2}{\mu})$ and $\mu>\max\{\frac{1}{\frac{2}{L_d}-\rho(\mathbf{K})},
    \frac{1}{L_d}\}$.
\end{enumerate}
Fig.~\ref{fig:zxu} plots in logarithmic scale the value $\|\mathbf{v}^{k+1}-\mathbf{v}^k\|$ throughout the iterations, as well as a straight line where the value shrinks according to the constant theoretical upper bound given in Theorem 3.4. We can see that the actual shrinkage ratio stays below the upper bound.

\bibliographystyle{ACM-Reference-Format}
\bibliography{AA-ADMM}